\documentclass[12pt]{article}
\usepackage[margin=1in]{geometry}

\usepackage{setspace}
\usepackage{tikz}
\usetikzlibrary{arrows.meta,positioning,shapes.geometric}

\usepackage{array}
\newcolumntype{L}[1]{>{\raggedright\arraybackslash}p{#1}}
\newcolumntype{C}[1]{>{\centering\arraybackslash}p{#1}}

\usepackage{natbib}
\usepackage{graphicx}
\usepackage{amsmath,amssymb,mathtools}
\usepackage{amsthm}

\theoremstyle{plain}
\newtheorem{theorem}{Theorem}

\newtheorem{lemma}{Lemma}

\theoremstyle{definition}
\newtheorem{assumption}{Assumption}

\theoremstyle{remark}
\newtheorem{remark}{Remark}

\renewcommand{\P}{\mathbb{P}}

\DeclareMathOperator{\E}{\mathbb{E}}

\usepackage{centernot}
\usepackage{tikz}
\usepackage{url}
\usepackage{comment}
\usepackage{xcolor}
\usepackage{wrapfig}
\usepackage{xcite}
\usepackage{xr}
\usepackage{booktabs}
\usepackage{multirow}
\usepackage{siunitx}
\usepackage{amsfonts}
\usepackage{nicefrac}
\usepackage{algorithm}
\usepackage{algorithmicx}
\usepackage{algpseudocode}
\usepackage{adjustbox}
\usepackage{blkarray}
\usepackage{bm}
\usepackage{paralist}
\usepackage{enumitem}

\begin{document}

\title{Causal State-Dependent Local Projections}

\author{Joel M. David, Raffaella Giacomini, Xiyu Jiao, and Weining Wang\thanks{David: Federal Reserve Bank of Chicago (joel.david@chi.frb.org); 
Giacomini: Department of Economics, University College London (r.giacomini@ucl.ac.uk); Jiao: Department of Economics, University of Gothenburg (xiyu.jiao@economics.gu.se); Wang: Department of Economics, University of Bristol (weining.wang@bristol.ac.uk). We would like to thank Tincho Almuzara, Marco Brianti, Lutz Kilian and Wei Cui for useful comments and suggestions and William Pennington for excellent research assistance. The views expressed here are those of the authors and not necessarily those of the Federal Reserve Bank of Chicago or the Federal Reserve System.}}
\date{\today}

\maketitle
\begin{abstract}

State-dependent local projections (LPs) are widely used to study how causal effects vary as a function of economic states, but shock exogeneity alone does not identify this response function. We show that identification follows when the underlying conditional mean is linear in the shock with a state-dependent coefficient, a condition satisfied in canonical micro-macro environments, including first-order perturbation solutions of heterogeneous-agent and macro-finance models. Even then, standard linear-interaction LPs generally recover only a projection of the response function, motivating LPs with nonparametric state dependence. We develop a sieve estimator and establish pointwise and uniform inference for micro-macro panels, where a distinctive challenge is that the estimator can converge at different rates across the state space. Applied to firm investment, the method uncovers a hump-shaped response to monetary policy shocks and shows that standard linear-interaction LPs substantially understate the aggregate role of financial heterogeneity.
\end{abstract}




 \bigskip

\newpage
\section{Introduction}

State-dependent local projections (LPs) are widely interpreted as revealing how causal impulse responses vary with an observable state, but this functional interpretation is not generally warranted, even under the causal shock-exogeneity conditions maintained in the recent LP literature.\footnote{\citet{JordaTaylor2025} provide a survey of the empirical literature. State-dependent LP applications span both aggregate and micro-macro settings; \citet{gonccalves2024state} contains extensive references for aggregate applications. In micro-macro settings, state-dependent LPs are widely used to study heterogeneous transmission across a broad range of questions: how the effect of monetary policy shocks on firm investment varies with financial conditions such as distance to default \citep{ottonello2020financial}, balance sheet liquidity \citep{jeenas2023firm}, intangible asset intensity \citep{DottlingRatnovski2023}, or stock market turnover \citep{JeenasLagos2024}; how monetary transmission to household consumption and wealth differs by mortgage status \citep{CloyneFerreiraSurico2020,cloyne2023monetary} or race \citep{BarstscherKuhnSchularickWachtel2022}; and how the growth effects of climate shocks vary with countries' income level or temperature exposure \citep{nath2024global}. } \citet{kolesar2024dynamic} show that scalar LPs, which regress outcomes on shocks,
recover average causal effects even when the data-generating
process (DGP) is nonlinear. 
\citet{winkler2026} and \citet{AlmuzaraSancibrian2025MicroResponses} extend this insight to state-dependent LPs and show that shock-state interaction coefficients can have causal interpretations
as weighted averages or projections of causal effects. At the same time,
\citet{gonccalves2024state} show that state-dependent LPs can fail to recover causal
impulse responses, pointing to feedback from endogenous
states as a key source of the problem. These results reveal a fundamental tension: a shock-state interaction
coefficient can be a causal average or projection without recovering how the
causal response varies with the state, which is what empirical researchers
seek to estimate. Concretely, a positive interaction coefficient need not
imply that the causal response is larger at higher values of the state,
contrary to the usual empirical interpretation.

This paper resolves this tension and provides an organizing framework for a large empirical literature by developing identification and estimation for causal state-dependent impulse responses, understood as horizon-specific functions of the chosen observable state, together with inference tailored to micro-macro panels. We make three main contributions.

First, we establish that, beyond the causal shock-exogeneity conditions maintained in the existing literature, identifying the impulse response as a function of the state requires additional structure on the underlying DGP. We provide a sufficient semiparametric condition: at each horizon, the conditional mean is linear in the shock, with an unrestricted state-dependent coefficient and an otherwise unrestricted additive component. The condition holds in many canonical micro-macro
environments, including cases with endogenous states, but is less plausible in
purely aggregate settings. Moving from scalar to state-dependent LPs therefore
requires thinking about the DGP, but not specifying it fully: once the condition
holds, LPs can recover the causal state-dependent response by leveraging shock
exogeneity.

Second, we clarify the role of LP specification choices when the identifying conditions are satisfied. One need not choose the true state: any predetermined state defines a causal subpopulation impulse response. But functional form matters in a way that has no analogue for scalar LPs. Scalar LPs target an average causal effect, whereas state-dependent LPs are interpreted as describing how the causal response varies with the state. Unless that variation is exactly linear, we show that the standard linear-interaction specification recovers only a projection that does not support this interpretation.  

Third, we close the gap between identification and practice. Because the object of interest is a function, recovering it generally requires nonparametric estimation. We therefore develop a sieve estimator and establish pointwise and uniform inference for micro--macro panels, addressing the distinctive challenge that the estimator can converge at different rates across the state space.

The identification insights extend to purely aggregate applications, but we
focus throughout on micro-macro settings, where state-dependent LPs are used
to study how aggregate shocks are transmitted to heterogeneous households,
firms, and regions. The standard approach is to regress a micro outcome $Y_{i,t+h}$ on an empirical measure of an aggregate shock $X_t$ and the interaction term $Z_{i,t-1} X_t$ between the shock and a state $Z_{i,t-1}$, thereby imposing a linear form of state dependence. 

Our key sufficient condition for identification is that the
horizon-$h$ conditional mean of the micro outcome is linear in the aggregate
shock, with a functional coefficient $g_h(Z_{i,t-1})$:
\begin{equation}
\mathbb E[Y_{i,t+h}\mid \mathcal F_{t-1},X_t]
=
g_h(Z_{i,t-1})X_t+r_{i,h}(\mathcal F_{t-1}),
\label{ass3}
\end{equation}
where $\mathcal F_{t-1}$ is the information set. This is a semiparametric restriction on the underlying DGP: the
functional coefficient $g_h(Z_{i,t-1})$ and the nuisance component
$r_{i,h}(\mathcal F_{t-1})$ need not be modeled. We combine this restriction
with causal shock-exogeneity conditions closely related to those in the
existing causal LP literature, which link observed conditional means to
potential outcomes. We additionally impose an m.d.s.\ condition on $X_t$,
which allows $g_h(\cdot)$ to be identified without modeling $r_{i,h}(\mathcal F_{t-1})$.\footnote{If the empirical
shock is obtained by first removing predictable components, $X_t$ is replaced
by the residualized shock $\widetilde X_t$. If the measured series is treated
as an instrument for an endogenous aggregate innovation $\varepsilon_t$, then
the identifying restrictions apply to $\varepsilon_t$, and $X_t$ in the LP is
replaced by the fitted value of $\varepsilon_t$ from the first stage. In both
cases, the identification argument is unchanged, while our inference theory
treats the shock as observed and abstracts from first-stage uncertainty and
weak instruments.} The state $Z_{i,t-1}$ is predetermined but
may nevertheless be endogenous, in the sense of depending on past shocks.
Under these conditions, $g_h(z)$ equals the total causal impulse response at
horizon $h$ for units with state $z$, combining both direct and indirect
effects of the shock. Thus, the conditional-mean restriction represents the
causal response as a functional coefficient in an LP that is nonparametrically
identified. Importantly, we show that identification does not require observing the exact
state governing heterogeneity: any predetermined state $Z_{i,t-1}$ defines a
well-defined causal response for the subpopulation with $Z_{i,t-1}=z$.

We then characterize economically relevant classes of micro-macro DGPs under which condition \eqref{ass3} is satisfied. The two key restrictions are that the aggregate shock enters linearly with a state-dependent coefficient and that the propagation dynamics are stable over time. This admissible class is broad and includes heterogeneous-agents and macro-finance models in which the state represents fixed characteristics such as age or mortgage status or time-varying endogenous (but predetermined) ones such as wealth, leverage, or balance-sheet positions. In particular, first-order perturbation solutions of canonical heterogeneous-agent
(HA) and heterogeneous-agent New Keynesian (HANK) models satisfy these conditions
within the linearized approximation: a monetary shock may reduce consumption more for highly leveraged households, depress investment more for highly leveraged firms, and lower wealth more for holders of long-duration assets, even though the structural laws governing the dynamics of consumption smoothing, capital accumulation, and wealth remain unchanged. Our findings imply that, in these environments, LPs with exogenous shock measures
can recover causal state-dependent impulse responses directly, without requiring
a full specification of the DGP.

Departures from this structure can arise in a number of ways. One occurs when
the relationship between outcomes and the aggregate shock is nonlinear, as in
models with occasionally binding constraints, regime switching, sign
asymmetries, or global nonlinear
solutions of heterogeneous-agent models. A more subtle violation occurs when
shocks also change propagation dynamics. The micro-macro model of
\citet{almuzara2025nonlinear}, where aggregate conditions affect
household-income persistence, is one such environment.

That the condition holds in canonical micro-macro environments does not,
however, justify the standard LPs with linear state dependence used in empirical
work. Even under identification, the causal state-dependent
impulse response is the function \(g_h(\cdot)\). Unless the true state
dependence is exactly linear, we show that the standard LP specification does not
recover this function, but instead a projection; its interaction coefficient averages local derivatives of \(g_h(\cdot)\) rather than
levels. Thus, contrary to the usual interpretation in applications, a positive
coefficient need not imply that high-state units have larger causal responses,
while a small coefficient may reflect cancellation from hump-shaped state
dependence.

We develop a sieve estimator that directly recovers the causal state-dependent impulse response as a flexible function of the state. Inference in micro-macro panels is technically delicate because the relative importance of aggregate and idiosyncratic uncertainty, and hence the convergence rate of the estimator, can vary across the state space. Where aggregate variation remains relevant, cross-sectional averaging cannot eliminate it and the rate is governed by the time dimension; where it vanishes, averaging across units becomes effective and convergence is faster. The estimator can therefore converge at different rates
across the state space. We characterize these rates and establish valid
pointwise confidence intervals. For uniform inference, we studentize the
process by its state-specific standard error, yielding valid confidence
bands across these regimes under strong common dependence across units and
serial dependence from overlapping LP horizons.

Monte Carlo evidence shows that linear state-dependence specifications can
substantially distort nonlinear causal response functions. The proposed
uniform confidence bands exhibit finite-sample undercoverage in short panels,
with coverage improving as the time dimension grows. The simulations also illustrate the role of common LP-error components in inference. At horizons $h\geq1$, omitting intermediate-shock terms does not affect identification, but it leaves common variation in the LP error that cannot be averaged out across units, resulting in wider confidence bands.

Applying our approach to firm-level investment responses to monetary policy shocks, we revisit the widely used LP in \citet{ottonello2020financial}, which imposes linear state dependence in firms' financial conditions. We find that the limitation identified by our theory is quantitatively important: the linear interaction coefficient averages slopes of the true state-dependent response and thus misses a hump-shaped pattern in which the largest effects are concentrated near the mean of the distance-to-default distribution, rather than in the least risky right tail. Nonlinearity also matters for aggregate transmission. Because many firms lie in regions where the nonlinear response exceeds the
linear approximation, the standard specification underestimates the role of
financial heterogeneity: this channel raises the impact response of aggregate investment
roughly ten times more in our nonparametric estimates than under the standard
linear specification.

Our framework connects and clarifies several strands of the literature. For empirical work using state-dependent LPs, a central implication is that shock exogeneity alone is not enough to recover causal state-dependent impulse responses. An additional sufficient condition is that the total effect of the shock is representable as a coefficient on a shock that enters the true conditional mean linearly at each horizon. This structure is natural in micro-macro settings with heterogeneous exposure and stable propagation, including cases with endogenous predetermined states. It rules out nonlinear shock effects and shock-dependent propagation. Thus, the critique of \citet{gonccalves2024state} remains relevant in environments where shocks can change the propagation mechanism, especially aggregate time-series settings. Even within the admissible class, standard linear-interaction LPs recover the causal state-dependent response only if state dependence is truly linear - a strong functional-form restriction with no general justification. One should therefore generally estimate state dependence nonparametrically.

Our results build on the recent literature on causal interpretations of LPs
\citep{kolesar2024dynamic,AlmuzaraSancibrian2025MicroResponses,Bousquet2025_UncoveringNonlinearities,winkler2026}, which shows that LP coefficients can be interpreted as weighted averages or
projections of causal effects. We ask a different question: when do LPs recover the causal state-dependent response function itself? We show that this requires additional structure on the DGP and a shift in focus from interpreting a finite-dimensional LP coefficient to estimating a causal response function. This reveals the limitations of existing linear specifications and motivates our nonparametric estimator and pointwise and uniform inference procedures.

Our findings also complement the LP - VAR equivalence literature. \citet{PlagborgMollerWolf2021} and \citet{Ludwig2024} characterize the equivalence between LP and VAR representations of impulse responses. We establish a related equivalence for impulse responses as functions of observable states: we characterize a class of DGPs under which the functional coefficient in a state-dependent LP coincides with the causal state-dependent response.

\citet{gonccalves2024state} show that state-dependent LPs need not recover finite-shock impulse responses when the state is endogenous, even if they recover marginal responses to infinitesimal shocks. We show that the deeper obstacle is not state endogeneity, but shock-dependent propagation or nonlinearity in the shock. In micro-macro settings with linear shock exposure and stable propagation, finite-shock and marginal responses coincide, and LPs with nonparametric state dependence recover causal responses even when the predetermined state is endogenous.

\citet{gonccalves2024_wp} address the general nonlinear case by
nonparametrically estimating the conditional mean as a function of both the
shock and the state, and then constructing impulse responses as differences
across shock realizations. Our approach is also nonparametric, but only in the
state: under our conditions, the causal impulse response is
directly the LP functional coefficient, so no second-step transformation
is required. The two approaches also differ in their shock-exogeneity
assumptions. \citet{gonccalves2024_wp} impose conditional independence of
potential outcomes from the shock given past information, whereas we require
only the corresponding conditional-mean restriction, together with the standard
m.d.s.\ condition on the shock.

From an inference perspective, our setting combines the functional inference problem studied in the sieve literature \citep[e.g.][]{Chernozhukov2014gaussian,ZW15gaussian,lu2020kernel,chen2024inference} with the micro--macro dependence structure emphasized by \citet{AlmuzaraSancibrian2025MicroResponses} for finite-dimensional LP coefficients. Existing results do not cover the interaction of these two features: because
the target is a function, the loading on aggregate uncertainty can vary with
the state, so the effective sample size, and hence the convergence rate of the
estimator, can vary across the state space. We develop inference that
remains valid across these rate regimes.

Our empirical application is related to \citet{paranhos2025firms}, who uses random forests to flexibly estimate heterogeneous impulse responses in the setting of \citet{ottonello2020financial}. Her approach relaxes functional-form restrictions on heterogeneity, while taking identification of the impulse response as given. We characterize when the impulse response is identified; within those settings, random forests provide an alternative estimator, whereas our sieve implementation remains closer to standard LP practice and delivers uniform inference across the state space. \citet{drechsel2026investment} also study firm-level investment responses to
monetary policy in the same setting. Their clustering approach targets
unobserved heterogeneity in responsiveness under a finite-mixture structure,
whereas we study identification and inference for the causal response as a
function of a chosen observed state.

Figure~\ref{fig:sdlp-decision-tree} provides a summary and practical guide to when LP
specifications recover causal state-dependent impulse responses and when
nonparametric state dependence is needed.

\begin{figure}[t]
\centering
\resizebox{0.88\textwidth}{!}{%
\begin{tikzpicture}[
    every node/.style={font=\scriptsize, align=center},
    q/.style={draw, rounded corners, inner sep=2.8pt, text width=2.90cm},
    result/.style={draw, rounded corners, inner sep=2.8pt,
                   text width=3.65cm, font=\scriptsize\bfseries},
    arrow/.style={-{Latex[length=1.4mm]}, thick}
]

\node[q]      (q1)    at (0,0)       {Causal shock-exogeneity?};
\node[q]      (q2)    at (0,-1.35)   {DGP linear\\in the shock?};
\node[q]      (q3)    at (0,-2.70)   {Stable propagation\\in the DGP?};

\node[result] (id)    at (0,-4.15)
{Causal state-dependent\\IRF identified};

\node[q]      (q4)    at (0,-5.65)
{DGP implies linear\\state dependence?};

\node[result] (fail)  at (5.00,-2.00)
{Causal state-dependent IRF\\not guaranteed};

\node[result] (linok) at (-3.25,-8.15)
{Standard linear-interaction LPs recover IRF};

\node[result] (np) at (3.25,-8.15)
{Standard linear-interaction LPs do not recover IRF;\\estimate state dependence\\nonparametrically};

\draw[arrow] (q1) -- node[right] {\tiny yes} (q2);
\draw[arrow] (q2) -- node[right] {\tiny yes} (q3);
\draw[arrow] (q3) -- node[right] {\tiny yes} (id);
\draw[arrow] (id) -- (q4);

\draw[arrow] (q1.east) -- node[pos=0.22, above, yshift=6pt] {\tiny no} (fail.north west);
\draw[arrow] (q2.east) -- node[pos=0.30, above, yshift=3pt] {\tiny no} (fail.west);
\draw[arrow] (q3.east) -- node[pos=0.38, above, yshift=3pt] {\tiny no} (fail.south west);

\draw[arrow] (q4.south west) -- node[above left] {\tiny yes} (linok.north);
\draw[arrow] (q4.south east) -- node[above right] {\tiny no} (np.north);

\end{tikzpicture}%
}
\caption{When do LPs recover a causal state-dependent IRF?
Causal shock-exogeneity alone is not enough. If the DGP is linear in the shock
and propagation is stable, the causal state-dependent IRF is identified.
Standard linear-interaction LPs recover that IRF when the DGP implies
linear state dependence. Otherwise, state dependence should be estimated nonparametrically.}
\label{fig:sdlp-decision-tree}
\end{figure}

The rest of the paper is organized as follows.
Section~\ref{sec:toy} provides intuition through simple DGPs;
Section~\ref{identification} gives the general identification results;
Sections~\ref{sec:estimation}-\ref{sec:simulations} develop estimation,
inference, and simulations;
Section~\ref{sec:application} applies the method to monetary policy
transmission across firms; and Section~\ref{sec:conclusion} concludes.
The Online Appendix contains the formal asymptotic theory, proofs, and
additional simulation results.

\section{A Simple Framework: DGPs and LP Specifications}\label{sec:toy}

This section uses three stylized DGPs and three LP specifications to illustrate
when state-dependent LPs recover the intended response and when they fail.

Consider the following DGPs for a scalar outcome \(Y_{it}\), a causally exogenous aggregate shock \(X_t\) (as formalized in Assumption 1 below) and a predetermined micro state \(Z_{i,t-1}\):
\begin{align*}
Y_{it}
&= a(Z_{i,t-1})X_t + \rho Y_{i,t-1} + \varepsilon_{it},
&& \text{DGP 1: linear in shock, constant propagation} \\[0.5em]
Y_{it}
&= a(Z_{i,t-1})X_t + \rho_t Y_{i,t-1} + \varepsilon_{it},
&& \text{DGP 2: shock-dependent propagation} \\[0.5em]
Y_{it}
&= f(Z_{i,t-1},X_t) + \rho Y_{i,t-1} + \varepsilon_{it},
&& \text{DGP 3: nonlinear in shock}
\end{align*}
where: $\varepsilon_{it}$ is i.i.d.\ across $i$ and $t$, mean zero, and
independent of $X_t$; $a(\cdot)$ is an unrestricted function of $Z$; $f(\cdot,\cdot)$ is nonlinear in $X$, and in DGP~2 $\rho_t$ is a function of past aggregate shocks, so that perturbations to $X_t$ affect future propagation.

These DGPs isolate three distinct channels through which aggregate shocks affect outcomes: heterogeneous exposure to a shock that enters linearly, but with time-invariant propagation (DGP~1), shock-dependent propagation (DGP~2), and nonlinear dependence on the shock (DGP~3). As we show below, only in the first case do LPs recover causal state-dependent impulse responses.

The key distinction across the three DGPs is how the horizon-$h$ conditional
mean, $\mathbb E[Y_{i,t+h}\mid\mathcal F_{t-1},X_t]$, changes when the
shock is perturbed from $X_t$ to $X_t+\delta$. We focus here on the condition under which a state-dependent LP
recovers the functional response of interest, leaving the standard
causal shock-exogeneity assumptions that give this response a causal interpretation
to Section~\ref{identification}. The condition is that the change in the
conditional mean takes the multiplicative form
\[
g_h(Z_{i,t-1})\delta,
\]
so that the object of interest is the function $g_h(\cdot)$, which gives the state-dependent response per unit change in the shock.

Under DGP~1, increasing $X_t$ by $\delta$ changes the impact conditional mean
by $a(Z_{i,t-1})\delta$. With time-invariant propagation, the condition is satisfied because the horizon-$h$
change is $\rho^h a(Z_{i,t-1})\delta=g_h(Z_{i,t-1})\delta$. 

Under DGP~2, the condition generally fails because the shock also affects
future propagation. At $h=1$, the conditional mean contains
$\mathbb E[\rho_{t+1}Y_{it}\mid\mathcal F_{t-1},X_t]$, and increasing $X_t$
by $\delta$ changes this term through both $\rho_{t+1}$ and $Y_{it}$. The
resulting change is not generally proportional to $\delta$ with a coefficient
depending only on $Z_{i,t-1}$. This DGP is the micro-macro counterpart of that in 
\citet{gonccalves2024state}. Our formulation highlights that
the failure arises from shock-dependent propagation, rather than from state
endogeneity per se.

Under DGP~3, the condition can fail even at $h=0$ because the outcome is
nonlinear in the shock. With sign asymmetry, for example, $Y_{it}
=
\beta_1 Z_{i,t-1}X_t
+
\beta_2 Z_{i,t-1}X_t\mathbf 1\{X_t>0\}
+
\rho Y_{i,t-1}
+
\varepsilon_{it}$,
increasing $X_t$ by $\delta$ changes the shock-dependent part of the
conditional mean by
$\beta_1 Z_{i,t-1}\delta
+
\beta_2 Z_{i,t-1}
\left[
(X_t+\delta)\mathbf 1\{X_t+\delta>0\}
-
X_t\mathbf 1\{X_t>0\}
\right].$
Because the second term depends on $X_t$, the change cannot generally be
written as $g_0(Z_{i,t-1})\delta$.

Now consider different LP specifications:
\begin{subequations}
\begin{align}
Y_{i,t+h}
&= \beta_h^{\mathrm{scalar}} X_t + u_{i,t+h},
&& \text{Scalar LP} \label{slp}\\
Y_{i,t+h}
&= (\alpha_h+ \beta_h^{\mathrm{linearSD}} Z_{i,t-1}) X_t + u_{i,t+h},
&& \text{LP with linear state dependence} \nonumber\\
&&& \text{(standard approach)} \label{llp}\\
Y_{i,t+h}
&= b_h^{\top}\phi(Z_{i,t-1}) X_t + u_{i,t+h},
&& \text{LP with nonparametric state dependence} \nonumber\\
&&& \text{(our proposal)} \label{nlp}
\end{align}
\end{subequations}
where \(\phi(\cdot)\) is a vector of basis functions, allowing for flexible state dependence.

\citet{kolesar2024dynamic} show that scalar LPs as in \eqref{slp} retain an
average causal interpretation even when the DGP is nonlinear. By contrast, we show that the standard LP with linear state dependence in \eqref{llp} does not
generally recover the object of interest \(g_h(\cdot)\) (unless \(g_h(\cdot)\) is linear) and that \(\beta_h^{\mathrm{linearSD}}\) is
a weighted average of local derivatives of \(g_h(\cdot)\). Thus, a positive \(\beta_h^{\mathrm{linearSD}}\) need not imply larger causal responses for higher-state units, contrary to the usual empirical interpretation.

This reveals a fundamental asymmetry. Once interaction terms are introduced, the robustness of scalar LPs no longer carries over and functional-form misspecification matters. Standard state-dependent LPs therefore generally do not recover $g_h(\cdot)$; flexible nonparametric dependence on $Z_{i,t-1}$, as in \eqref{nlp}, permits the LP to recover it directly under DGP~1.

In DGP~2 and DGP~3, by contrast, the impulse response does not admit a multiplicative representation in the first place and state-dependent LPs generally give projections and not the object of interest. For DGP~3, note that including nonlinear shock terms in the LP does not by itself solve the problem because, even if \(X_t\) is an m.d.s., transformations such as \(X_t\mathbf 1\{X_t>0\}\) need not be. Hence the nonlinear shock regressors are not guaranteed to be orthogonal to the LP residual. This applies to scalar and state-dependent LPs alike: nonlinear-in-shock LPs therefore require stronger assumptions.

Table 1 summarizes what each LP specification recovers in the three stylized
DGPs. 

\begin{table}[t]
\centering
\caption{What LPs recover under three stylized DGPs}
\vspace{0.5em}
\label{tab:toyDGPs}
\renewcommand{\arraystretch}{1.4}
\begin{tabular}{L{0.28\textwidth} C{0.18\textwidth} C{0.18\textwidth} C{0.18\textwidth}}
\hline\hline
& \textbf{DGP 1} & \textbf{DGP 2} & \textbf{DGP 3} \\
\hline
\textbf{Scalar LP} &
Average causal effect &
Average causal effect &
Average causal effect \\[2em]

\textbf{LP with linear state dependence (standard approach)} &
Projection (IRF if state dependence is exactly linear)&
Projection &
Projection \\[4.2em]

\textbf{LP with nonparametric state dependence (our proposal)} &
Causal state-dependent IRF &
Projection &
Projection  \\
\hline\hline
\end{tabular}
\end{table}

In the remainder of the paper, we: (i) formalize conditions under which LPs recover causal state-dependent impulse responses; (ii) characterize when the conditions hold or fail, admitting environments with state-dependent exposure and constant propagation (such as DGP~1) while excluding those with shock-dependent propagation or nonlinear-in-shock effects (such as DGPs~2-3); and (iii) show that standard LPs with linear state dependence do not generally recover the object of interest, motivating a nonparametric approach. We argue that the conditions plausibly hold in many canonical micro-macro DGPs and develop estimation and inference of causal state-dependent impulse responses via LPs with nonparametric state dependence in these settings.


\section{Identification of Causal State-Dependent IRFs}\label{identification}
We now formally discuss identification and the environments where it holds or fails.
\subsection{Identification and the Limits of Standard State-Dependent LPs}

This section provides a general characterization of when LPs recover causal state-dependent impulse responses. Rather than specifying a full structural model or DGP, we impose a sufficient condition on the horizon-specific conditional mean. We then show that, even under this condition, standard LPs with linear state dependence generally fail to recover the causal state-dependent impulse response, motivating our nonparametric approach.

We observe a panel $\{Y_{it},X_t,Z_{it},W_{it}\}$ for $i=1,\dots,N$ and
$t=1,\dots,T$, where $Y_{it}$ is a micro outcome, $X_t$ is an aggregate shock,
$Z_{it}$ is a micro state, and $W_{it}$ collects control variables that may be
included in the LP specification. Let $\mathcal F_{t-1}$ denote the pre-shock
information set,
$\mathcal F_{t-1}
=
\sigma\big(\{Y_{is},Z_{is},W_{is}\}_{i,s\le t-1},
           \{X_s\}_{s\le t-1}\big).$

          Following the potential-outcome formulation used in the causal LP literature,
let $Y_{i,t+h}(x)$ denote the horizon-$h$ outcome when the date-$t$ shock is
set to $x$. Our object of interest is the finite-$\delta$ state-dependent
impulse response
\[
\mathrm{IRF}_h(z;\delta)
=
\mathbb E\!\left[
Y_{i,t+h}(X_t+\delta)-Y_{i,t+h}(X_t)
\mid Z_{i,t-1}=z
\right].\footnote{This is the micro-macro analogue of the state-dependent response in
\citet{gonccalves2024state}. See their paper for a discussion of alternative definitions.}
\]

For each horizon $h\geq0$, we consider the LP
\begin{equation}
Y_{i,t+h}
=
g_h(Z_{i,t-1})X_t
+
W_{i,t-1}^{\top}\gamma_h
+
u_{i,t+h},
\label{eq:LP}
\end{equation}
where $g_h(\cdot)$ is an unknown function of the state. We seek conditions under which
$g_h(\cdot)$ coincides with the causal state-dependent response per unit
shock.

The following assumptions suffice.

\noindent\begin{assumption}(\textbf{Causal shock-exogeneity}).\label{assum1}
For each horizon $h\geq0$ and every $x$:
\begin{enumerate}
\item[(i)]
$\mathbb E[Y_{i,t+h}(x)\mid\mathcal F_{t-1},X_t]
=
\mathbb E[Y_{i,t+h}(x)\mid\mathcal F_{t-1}];$
\item[(ii)]
$
\mathbb E[X_t\mid\mathcal F_{t-1}]=0.
$
\end{enumerate}
\end{assumption}

\begin{assumption}(\textbf{Predetermined state and controls}).\label{assum2}
$(Z_{i,t-1},W_{i,t-1})$ are $\mathcal F_{t-1}$-measurable.
\end{assumption}

\begin{assumption}(\textbf{Conditional mean condition}).\label{assum3}
For each horizon $h\geq0$,
\begin{equation}
\mathbb E[Y_{i,t+h}\mid\mathcal F_{t-1},X_t]
=
g_h(Z_{i,t-1})X_t
+
r_{i,h}(\mathcal F_{t-1}).
\label{eq:A3}
\end{equation}
\end{assumption}

Assumption~\ref{assum1}(i) provides the causal link between observed
conditional means and potential outcomes.\footnote{We maintain consistency:
$Y_{i,t+h}=Y_{i,t+h}(x)$ when $X_t=x$. Hence, for the shock values under
consideration, Assumption~\ref{assum1}(i) implies
$\mathbb E[Y_{i,t+h}(x)\mid\mathcal F_{t-1}]
=\mathbb E[Y_{i,t+h}\mid\mathcal F_{t-1},X_t=x]$.
For finite perturbations, we restrict attention to $\delta$ such that
$X_t+\delta$ lies in the support of the conditional distribution of $X_t$
given $\mathcal F_{t-1}$.}
Some version of such a shock-exogeneity condition is standard in the causal
LP literature, typically requiring the shock to be independent, or conditionally
independent, of the remaining determinants of future outcomes.
In the structural-function framework of \citet{kolesar2024dynamic}, for
example, shock independence implies
Assumption~\ref{assum1}(i). With pre-shock information
$\mathcal F_{t-1}$, the analogous conditional-independence restriction would
likewise imply our assumption. We require only this conditional-mean
implication, rather than full independence, because
Assumption~\ref{assum3} restricts the true conditional mean.
Assumption~\ref{assum1}(ii) is the usual m.d.s.\ condition with respect to
the pre-shock information set and provides the orthogonality used by the LP
moment conditions. Assumption~\ref{assum2} is the standard requirement that state and controls be predetermined.
Notice that the state may nevertheless be endogenous, in the sense of
depending on past shocks.

Assumption~\ref{assum3} is the additional restriction. It requires that the
horizon-$h$ conditional mean be linear in the shock, with a slope that depends
on the predetermined state. By Assumption~\ref{assum1}(i), consistency, and
Assumption~\ref{assum3} we have that 
$\mathbb E[Y_{i,t+h}(x)\mid\mathcal F_{t-1}]
=
g_h(Z_{i,t-1})x+r_{i,h}(\mathcal F_{t-1})$ for every relevant $x$. 
It follows from Assumption~\ref{assum2} and iterated expectations that the
finite-$\delta$ response satisfies
\[
\mathrm{IRF}_h(z;\delta)
=
\mathbb E\!\left[
Y_{i,t+h}(X_t+\delta)-Y_{i,t+h}(X_t)
\mid Z_{i,t-1}=z
\right]
=
g_h(z)\delta.
\]
Thus, $g_h(z)$ is the total causal response per unit shock, including both
direct effects and indirect effects operating through the subsequent
endogenous evolution of the system. The nuisance component
$r_{i,h}(\mathcal F_{t-1})$ is unrestricted and need not coincide with the
linear projection onto $W_{i,t-1}$ in \eqref{eq:LP}; it may contain arbitrary
lagged dynamics, fixed effects, and nonlinear dependence on past aggregate
and idiosyncratic variables.

Finally, Assumptions~\ref{assum1}(ii), \ref{assum2}, and~\ref{assum3} imply
\[
g_h(z)
=
\frac{\mathbb E[X_tY_{i,t+h}\mid Z_{i,t-1}=z]}
     {\mathbb E[X_t^2\mid Z_{i,t-1}=z]},
\qquad z\in\mathcal Z,
\]
so the causal state-dependent response per unit shock is nonparametrically
identified and coincides with the LP estimand without controls.\footnote{We assume the relevant moments are finite and $\mathbb E[X_t^2\mid Z_{i,t-1}=z]>0$ for all $z\in\mathcal Z$. Controls $W_{i,t-1}$ may be included for efficiency without affecting identification.}

The following result clarifies that a causal interpretation is still available
even when the chosen state variable $Z_{i,t-1}$ does not coincide with the true
state.

\bigskip
\noindent\textbf{Proposition 1 (State choice preserves causal interpretation).}
Suppose Assumptions~\ref{assum1}-\ref{assum2} hold and that, for some
$\mathcal F_{t-1}$-measurable state $s_{i,t-1}$,
\[
\E[Y_{i,t+h}\mid \mathcal F_{t-1},X_t]
=
\tilde g_h(s_{i,t-1})X_t+r_{i,h}(\mathcal F_{t-1}).
\]
Let $Z_{i,t-1}$ be any $\mathcal F_{t-1}$-measurable state. For any $z$ such
that $\E[X_t^2\mid Z_{i,t-1}=z]>0$, consider the functional LP estimand
\[
g_h(z)
=
\frac{\E[X_tY_{i,t+h}\mid Z_{i,t-1}=z]}
{\E[X_t^2\mid Z_{i,t-1}=z]}.
\]
Then
\[
g_h(z)
=
\E\!\left[
w_h(s_{i,t-1};z)\tilde g_h(s_{i,t-1})
\mid Z_{i,t-1}=z
\right],
\qquad
w_h(s;z)
=
\frac{\E[X_t^2\mid s_{i,t-1}=s,Z_{i,t-1}=z]}
{\E[X_t^2\mid Z_{i,t-1}=z]}.
\]
The weights satisfy $w_h(s;z)\ge0$ for all $s$ and
$\E[w_h(s_{i,t-1};z)\mid Z_{i,t-1}=z]=1$.
Thus, $g_h(z)$ is a convex variance-weighted average of the underlying
causal responses within the subpopulation $Z_{i,t-1}=z$.\footnote{With
predetermined controls in the LP, the same result holds after partialling
them out.}

\begin{proof}[Proof of Proposition~1]
Using the conditional-mean restriction and iterated expectations,
\[
\E[X_tY_{i,t+h}\mid Z_{i,t-1}=z]
=
\E[\tilde g_h(s_{i,t-1})X_t^2\mid Z_{i,t-1}=z],
\]
because, by Assumption~\ref{assum1}(ii), and since
$r_{i,h}(\mathcal F_{t-1})$ and $Z_{i,t-1}$ are
$\mathcal F_{t-1}$-measurable,
\[
\E[X_t r_{i,h}(\mathcal F_{t-1})\mid Z_{i,t-1}=z]
=
\E\!\left[
r_{i,h}(\mathcal F_{t-1})
\E[X_t\mid\mathcal F_{t-1}]
\mid Z_{i,t-1}=z
\right]
=0.
\]
Dividing by $\E[X_t^2\mid Z_{i,t-1}=z]$ and applying iterated expectations
gives
\[
g_h(z)
=
\E\!\left[
w_h(s_{i,t-1};z)\tilde g_h(s_{i,t-1})
\mid Z_{i,t-1}=z
\right].
\]
Moreover, $w_h(s;z)\ge0$, and
\[
\E[w_h(s_{i,t-1};z)\mid Z_{i,t-1}=z]
=
\frac{
\E[\E[X_t^2\mid s_{i,t-1},Z_{i,t-1}=z]
\mid Z_{i,t-1}=z]
}{
\E[X_t^2\mid Z_{i,t-1}=z]
}
=1.
\]
\end{proof}

\begin{remark}[Interpretation and robustness to state choice]
Proposition~1 shows that a functional LP based on any predetermined state
$Z_{i,t-1}$ identifies a variance-weighted average of the underlying causal
responses within the subpopulation $Z_{i,t-1}=z$, even if $Z_{i,t-1}$ does
not span all relevant heterogeneity. Hence $g_h(z)$ remains a well-defined
subpopulation causal response. In particular, if $g_h(z)$ is increasing,
subpopulations with higher values of $Z_{i,t-1}$ have larger causal responses.
\end{remark}

\begin{remark}[Heterogeneity]
The notation $g_h(\cdot)$ denotes a common estimand, not a common structural
effect. If the underlying responses are unit-specific,
$\tilde g_{h,i}(s_{i,t-1})$, then $g_h(z)$ is a variance-weighted average of
these responses within the subpopulation $Z_{i,t-1}=z$.
\end{remark}

The causal state-dependent impulse response is the function $g_h(\cdot)$. A natural estimator is therefore nonparametric. By contrast, the applied literature uses LPs with linear state dependence and interprets this linear function as a causal state-dependent impulse response. The next result shows that this interpretation is generally incorrect.

\bigskip
\noindent\textbf{Proposition 2 (LPs with linear state dependence do not generally recover state-dependent IRFs).}
\label{prop:coarse2}
Suppose Assumptions~\ref{assum1}-\ref{assum3} hold, the regularity
conditions of Lemma~1 in Appendix A are satisfied, and $g_h(\cdot)$ is absolutely
continuous on $\mathcal Z$. The slope coefficient in the LP with linear state dependence in
\eqref{llp} can be written as
\[
\beta_h^{\mathrm{linearSD}}
=
\int_{\mathcal Z} \omega_h(z)\,g_h'(z)\,dz,
\]
where \(g_h'(z)\) is the derivative of the state-dependent impulse response and
\(\omega_h(z)\) is a weighting function determined by the joint distribution of the data, with \(\omega_h(z) \ge 0\) for \(z \in \mathcal Z\) and
\(\int_{\mathcal Z}\omega_h(z)\,dz=1\). Lemma~\ref{misspecification beta} in Appendix~\ref{app:linear_misspec} provides the general result underlying this proposition and its derivation.
\bigskip

Recent work has shown that LP coefficients can retain causal content even under
misspecification: they can be interpreted as weighted averages or projections of
causal responses
\citep{kolesar2024dynamic,Bousquet2025_UncoveringNonlinearities,
winkler2026,AlmuzaraSancibrian2025MicroResponses}. Proposition~2 clarifies what this means in state-dependent applications, where
the object of interest is not a scalar average effect but the response function
over the state. The distinction is clearest in the
conditionally homoskedastic case,
\(\mathbb E[X_t^2\mid Z_{i,t-1}=z]= \mathbb E[X_t^2]\). For the scalar LP in \eqref{slp}, the result of
\citet{kolesar2024dynamic}, applied to our setting, implies
\[
\beta_h^{\mathrm{scalar}}
=
\mathbb E[g_h(Z_{i,t-1})]
=
\int_{\mathcal Z} g_h(z)\,dF_Z(z).
\]
Thus the scalar LP coefficient averages response levels: it is an average causal
effect.

The coefficient in the standard LP with linear state dependence is different. 
Appendix~\ref{app:linear_misspec} shows that, under conditional homoskedasticity,
$\beta_h^{\mathrm{linearSD}}
=
\frac{\operatorname{Cov}\!\left(Z_{i,t-1},g_h(Z_{i,t-1})\right)}
     {\operatorname{Var}(Z_{i,t-1})} ,$
so the coefficient is the slope in the linear projection of \(g_h(Z_{i,t-1})\) on a constant and \(Z_{i,t-1}\). Proposition~2
unpacks this projection slope further by showing that it is a weighted average of local derivatives of
\(g_h(\cdot)\), not of response levels. Consequently, unless \(g_h(\cdot)\) is linear, LPs
with linear state dependence do not generally recover the causal
state-dependent response function that empirical applications typically seek to
estimate. In particular, a positive \(\beta_h^{\mathrm{linearSD}}\) need not
imply larger responses at higher values of \(Z_{i,t-1}\), and a small coefficient may reflect cancellation across regions of the state space 
where the function has opposite slopes.

\medskip



\subsection{Economic Content of the Identifying Condition}

Assumption~\ref{assum3} is a restriction on the conditional mean of the outcome, rather than a full characterization of the DGP. It is not generally testable, but it is implied by economically meaningful classes of DGPs. Here we provide an economic interpretation of this condition and highlight classes of canonical micro-macro DGPs in which it is satisfied. Importantly, this characterization is not used for estimation; it clarifies when LPs can recover causal state-dependent impulse responses.

\subsubsection{Admissible DGPs.}

Assumption~\ref{assum3} holds in a broad class of micro-macro environments in which the aggregate shock enters linearly and propagation is governed by time-invariant laws of motion. 

\emph{(i) First-order perturbation solutions of heterogeneous-agent models.}

Assumption~\ref{assum3} is naturally satisfied within first-order perturbation solutions of HA/HANK
models. For any individual (grid point) \(i\) with predetermined state \(s_{i,t-1}\)
that may be endogenously affected
by past shocks, the linearized equilibrium law of motion implies that the conditional
mean of a micro outcome is linear in the aggregate innovation, with a
coefficient that may depend flexibly on the micro state.
Mapping to our notation, let $Y_{it}$ be any household-level outcome (e.g.\ consumption or earnings) and $X_t$ be the aggregate shock. At impact, the linearized solution implies
$\mathbb{E}\!\left[Y_{i,t}\mid \mathcal{F}_{t-1},X_t\right]
=
g_0(s_{i,t-1})X_t+r_{i,0}(\mathcal{F}_{t-1}).$
Because all subsequent propagation is governed by the same fixed linear law of motion and policy rules, the effect of the date-$t$ shock remains linear as the system is iterated forward. Hence, for each horizon $h$,
\[
\mathbb{E}\!\left[Y_{i,t+h}\mid \mathcal{F}_{t-1},X_t\right]
=
g_h(s_{i,t-1})X_t+r_{i,h}(\mathcal{F}_{t-1}).
\]
Thus, Assumption~\ref{assum3} holds exactly for the linearized model: the shock affects outcomes through a linear system whose propagation coefficients are fixed. Linearization does not make \(g_h(\cdot)\) linear in the micro state because the function reflects the underlying nonlinear household problem across the state space.

This structure can be given a microfoundation using the sufficient-statistics representation of \citet{Auclert2019}, in which an individual's consumption response to an aggregate shock decomposes, to first order, into a substitution effect and a wealth effect equal to the product of the marginal propensity to consume (MPC) and a balance-sheet exposure term - both unrestricted, potentially nonlinear functions of the micro state. This characterizes the cross-sectional structure of the response, while the linearized dynamics standard in perturbation-based HANK models ensure propagation is governed by time-invariant coefficients, so the horizon-$h$ effect is summarized by $g_h(s_{i,t-1})$. This also motivates the empirical choice of a low-dimensional state $Z_{i,t-1}$: while the full micro state \(s_{i,t-1}\) is high-dimensional, responses depend primarily on MPCs and balance-sheet exposures, which can often be summarized by low-dimensional predetermined characteristics such as liquid wealth.

When the linearized system is used as an approximation to a nonlinear economy, the relevant requirement is that the first-order approximation remain uniformly accurate over the range of shock realizations being compared. This is the same requirement that underlies impulse-response analysis in linearized heterogeneous-agent models: our condition does not impose restrictions beyond those needed to interpret impulse responses computed from first-order perturbation solutions \citep[e.g.][]{BoppartKrusellMitman2018,BayerLuetticke2020}. Moderate perturbations may therefore be well approximated when the induced equilibrium path remains in a region where the equilibrium mapping is smooth and close to linear. The approximation may become unreliable when the perturbation moves the economy into regions where nonlinearities become quantitatively important, for example because borrowing constraints bind for many agents, aggregate constraints such as the zero lower bound become active, or the shock induces large reallocations across heterogeneous agents \citep{AhnKaplanMollWinberryWolf2018,david2022risk}.\footnote{If $\mathbb E[Y_{i,t+h}\mid \mathcal F_{t-1},X_t]$ is twice continuously differentiable in $X_t$, the difference between the nonlinear finite-shock response and its first-order approximation is governed by a second-order remainder over the relevant range of shock realizations.}

\medskip
\emph{(ii) Macro and macro-finance models.}

Consider a panel VAR for the joint process $(Y_{it},W_{it})^{\top}$ with state-dependent exposure to the aggregate shock and time-invariant propagation:
\begin{equation}
\begin{pmatrix}
Y_{it}\\
W_{it}
\end{pmatrix}
=
\alpha_i+\lambda_t
+
A_i(L)
\begin{pmatrix}
Y_{i,t-1}\\
W_{i,t-1}
\end{pmatrix}
+
\begin{pmatrix}
\lambda_Y(Z_{i,t-1})\\
\lambda_W(Z_{i,t-1})
\end{pmatrix}
X_t
+
\begin{pmatrix}
\eta_{it}\\
\nu_{it}
\end{pmatrix},
\label{eq:admissible_panel_VAR}
\end{equation}
where $A_i(L)$ is a lag polynomial with possibly heterogeneous but time-invariant coefficients. The functions $\lambda_Y(\cdot)$ and $\lambda_W(\cdot)$ govern exposure of all variables to the aggregate shock and may be constant or arbitrary measurable functions of the predetermined state $Z_{i,t-1}$.

Because the system is linear in $X_t$ and the propagation coefficients are fixed, iterating \eqref{eq:admissible_panel_VAR} forward implies that, for each horizon $h$,
$\mathbb{E}\!\left[Y_{i,t+h}\mid \mathcal F_{t-1},X_t\right]
=
r_{i,h}(\mathcal F_{t-1})
+
g_h(Z_{i,t-1})\,X_t,$
for some measurable function $g_h(\cdot)$. When the relevant exposure and propagation heterogeneity is captured by the
predetermined state, the function \(g_h(Z_{i,t-1})\) summarizes both the direct
effect of \(X_t\) on \(Y_{i,t+h}\) through \(\lambda_Y\), and the indirect effects
operating through the endogenous post-\(t\) path of \(W\) and \(Y\), through
\(\lambda_W\) and the time-invariant propagation \(A_i(L)\). If some propagation
heterogeneity is not included in \(Z_{i,t-1}\), the exact response depends on a
richer state, and \(g_h(Z_{i,t-1})\) should instead be interpreted as the
subpopulation average response characterized in Proposition~1.

This structure is consistent with leading heterogeneous-firm monetary models. For example, in \cite{ottonello2020financial}, monetary policy shocks shift the intertemporal price system and firms' financing conditions. The model generates nonlinearity in exposure because firms' marginal propensity to invest depends on default risk and can vary  over the financial-state distribution (see their Figure~2), so the function \(\lambda(\cdot)\) need not be linear. Importantly, while the financial state evolves endogenously after the shock, the structural forces governing propagation (e.g., depreciation and adjustment costs) remain time-invariant. As a result, the shock shifts the level of investment without altering its propagation dynamics.

The mechanism in \cite{cui2025risk} also generates nonlinear exposure: the effect of an interest-rate shock on risk-taking depends on a predetermined leverage constraint, and the threshold for adopting risky projects can be hump-shaped in the interest rate. The nonlinearity arises from equilibrium choice conditions, not from changes in the propagation dynamics following the shock.

\subsubsection{Inadmissible DGPs.}

\medskip
\emph{(i) Nonlinear dependence on the shock.}
Assumption~\ref{assum3} excludes environments in which the conditional mean is nonlinear in the shock, e.g.,
\[
\mathbb{E}[Y_{i,t+h}\mid \mathcal F_{t-1},X_t]
=
f_h(Z_{i,t-1},X_t)+r_{i,h}(\mathcal F_{t-1}).
\]
Examples include regime-switching policy, zero lower bound nonlinearities, terms such as $1(X_t>0)$ or large discrete interventions. 

\medskip
\emph{(ii) Shock-dependent propagation.}
Assumption~\ref{assum3} also fails when aggregate shocks change the propagation of the outcome itself. A useful example is \citet{almuzara2025nonlinear}, whose nonlinear income process allows business-cycle conditions to affect both income exposures and the persistence and riskiness of future income dynamics. If such a model is the true DGP, then even observing exogenous aggregate shocks is not enough for state-dependent LPs of the form considered here to recover the model-implied impulse responses: the response depends on the fully nonlinear propagation mechanism, including feedback through endogenous micro states affected by past aggregate shocks, rather than only on the initial state.

\medskip
\emph{(iii) Nonlinear HA/HANK models.}
The previous mechanisms arise naturally in nonlinear heterogeneous-agent environments beyond first-order linearized approximations. In such models, aggregate shocks can move the economy across regions of the state space where local dynamics differ, for example because of occasionally binding constraints, kinked adjustment costs, or endogenous regime shifts. Shocks then affect not only the level of \(Y_{it}\), but also its subsequent propagation. The horizon-specific response depends on the nonlinear state path induced by the shock, and the conditional mean generally cannot be written in the multiplicative form of Assumption~\ref{assum3}.

\medskip
Such violations of Assumption~\ref{assum3} are economically plausible in aggregate time-series settings, where shocks can alter persistence. Our condition targets a different empirical setting: micro-macro panels in which cross-sectional heterogeneity operates through state-dependent exposure to common aggregate shocks that enter linearly, while the propagation of individual outcomes is governed by time-invariant structural primitives (e.g., depreciation, adjustment costs, or idiosyncratic risk processes). In these environments, LPs can recover causal state-dependent impulse responses (provided that state dependence is estimated nonparametrically, which we discuss next).

\section{Estimation and Inference}
\label{sec:estimation}
This section develops sieve estimation and inference for the state-dependent impulse response function \(g_h(\cdot)\). Although the procedures use familiar ingredients - sieve approximation, HAC covariance estimation, and Gaussian simulation - their joint validity requires additional analysis in our setting. In micro-macro panels, aggregate components of LP errors may survive cross-sectional averaging. Because the importance of these components can vary with the state, the effective sample size, and hence the convergence rate of the estimator, can differ across the state space. Overlapping LP horizons additionally induce serial dependence. The Online Appendix characterizes these rate regimes, establishes feasible HAC consistency, and proves pointwise and state-uniform inference.

\subsection{Sieve Estimator}
\label{sec:sieve}
For each horizon $h \in \{0,1,\ldots,H\}$, we estimate $g_h(\cdot)$ using a
sieve approximation. Let $\{\phi_{j,J}(z)\}_{j=1}^J$ denote a
basis of cubic B-splines defined on the support of $Z_{i,t-1}$, where interior
knots are located at equally spaced empirical quantiles and with a constant term included. For a
given sieve dimension $J$, approximate
\begin{equation}
g_h(z) \approx \phi_J(z)^{\top} b_h,
\qquad
\phi_J(z) = \bigl(\phi_{1,J}(z),\ldots,\phi_{J,J}(z)\bigr)^{\top}.
\label{eq:sieve_rep}
\end{equation}

Substituting \eqref{eq:sieve_rep} into the LP yields the
estimation equation
\begin{equation}
Y_{i,t+h}
=
\sum_{j=1}^{J}
b_{h,j}\,\bigl[\phi_{j,J}(Z_{i,t-1})\,X_t\bigr]
+
W_{i,t-1}^{\top}\gamma_h
+
u_{i,t+h}.
\label{eq:lp_estimation}
\end{equation}

\medskip

Let $D_{it}(J)$ denote the vector of regressors in \eqref{eq:lp_estimation}. The
OLS estimator is
\begin{equation}
\widehat{\theta}_h(J)
=
\left(
\sum_{i=1}^{N}\sum_{t=1}^{T-h}
D_{it}(J)\,D_{it}(J)^{\top}
\right)^{-1}
\left(
\sum_{i=1}^{N}\sum_{t=1}^{T-h}
D_{it}(J)\,Y_{i,t+h}
\right),
\label{eq:ols_new}
\end{equation}
with $\widehat{b}_h(J)$ denoting the subvector corresponding to the sieve
coefficients. The associated estimator of the impulse response is
\begin{equation}
\widehat{g}_{h,J}(z) = \phi_J(z)^{\top}\,\widehat{b}_h(J).
\label{eq:ghat_J}
\end{equation}

We select the sieve dimension $J$ using the Akaike information criterion (additional selection criteria are examined in the simulation section and in the Appendix and perform similarly):
\begin{equation}
\mathrm{AIC}_h(J)
=
N(T-h) \log\!\left(
\frac{1}{N(T-h)}
\sum_{i=1}^{N}\sum_{t=1}^{T-h}
\widehat{u}_{i,t+h}(J)^2
\right)
+
2K_J,
\label{eq:aic_new}
\end{equation}
where
$\widehat{u}_{i,t+h}(J)$ denote the OLS residuals from \eqref{eq:lp_estimation} and $K_J = J + \dim(W_{i,t-1})$ is the total number of estimated parameters.
The selected sieve dimension is
\begin{equation}
\widehat{J}_h = \arg\min_{J \in \mathcal{J}}\,\mathrm{AIC}_h(J).
\label{eq:Jhat}
\end{equation}
In practice, we take the candidate set $\mathcal J=\{4,\ldots,20\}$, with the upper bound serving as a finite-sample safeguard against an overly rich basis. Here $J=4$ gives a global cubic polynomial, while larger values of $J$ add
interior knots to the cubic B-spline basis and allow for more flexible nonlinearities. The final estimator is
\begin{equation}
\widehat{g}_h(z) = \phi_{\widehat{J}_h}(z)^{\top}\,\widehat{b}_h(\widehat{J}_h).
\label{eq:ghat_final}
\end{equation}
The following algorithm summarizes the steps for estimating the impulse response.
\paragraph{Algorithm 1 (Sieve LP estimator of the impulse response).}
For each horizon $h = 0,\ldots,H$:
\begin{enumerate}
  \item Construct cubic B-spline basis functions for $Z_{i,t-1}$ with a constant term included and interior
        knots equally placed at empirical quantiles.
  \item For each $J \in \mathcal{J}$, form the $J$ interaction regressors
        $\phi_{j,J}(Z_{i,t-1})\,X_t$, $j = 1,\ldots,J$, and stack them with
        $W_{i,t-1}$ to form $D_{it}(J)$.
  \item Estimate \eqref{eq:lp_estimation} by OLS and compute residuals
        $\widehat{u}_{i,t+h}(J)$.
  \item Compute $\mathrm{AIC}_h(J)$ for each $J \in \mathcal{J}$ and select
        $\widehat{J}_h$ via \eqref{eq:Jhat}.
  \item Set $\widehat{g}_h(z) = \phi_{\widehat{J}_h}(z)^{\top}\,\widehat{b}_h(\widehat{J}_h)$.
\end{enumerate}

\begin{remark}[Multivariate states]\label{multi}
The estimator can be implemented with a multivariate state \(Z_{i,t-1}\in\mathbb R^d\)
by replacing the univariate basis with tensor-product basis functions and interacting
these basis functions with \(X_t\). The remainder of the framework extends analogously. As in any nonparametric procedure, however, increasing the dimension of
the state worsens the bias-variance trade-off, so the approach is most useful for low-dimensional states. 
\end{remark}

\subsection{Inference}
\label{sec:inference_main}

We consider both pointwise and uniform inference for the impulse response $g_h(\cdot)$. 

Uniform inference across the state space is important because the empirical object is often the entire state-dependent response function, rather than the response at a single preselected state. Economic conclusions frequently depend on features of the whole function - such as whether responses {appear} monotone, hump-shaped, concentrated in the tails, or {differ} across regions of the state space. Uniform confidence bands {provide simultaneous coverage of the response function over the relevant support, allowing such features to be assessed while accounting for sampling uncertainty jointly across states}.

Let $D_{it}$, $\widehat b_h$, $\widehat u_{i,t+h}$ denote, respectively, the vector of regressors, the sieve coefficient estimator, the OLS residual in \eqref{eq:lp_estimation} corresponding to the AIC-selected number of basis functions $\widehat{J}_h$. Consider the sample second moment matrix
\[
\frac{1}{N(T-h)}
\sum_{i=1}^N\sum_{t=1}^{T-h}
D_{it}D_{it}^{\top}
=
\begin{pmatrix}
\widehat A_{11,h} & \widehat A_{12,h}\\
\widehat A_{21,h} & \widehat A_{22,h}
\end{pmatrix},
\]
where the first block corresponds to
$X_t\phi_{\widehat J_h}(Z_{i,t-1})$ and the second to $W_{i,t-1}$.
Form the partialled-out sieve regressor
\begin{equation}
\widetilde P_{h,i,t}
=
X_t\phi_{\widehat J_h}(Z_{i,t-1})
-
\widehat A_{12,h}\widehat A_{22,h}^{-1}W_{i,t-1},
\label{eq:Ptilde_main}
\end{equation}
and the score process
\begin{equation}
\widehat s_{h,t}
=
\frac{1}{\sqrt N} \sum_{i=1}^N
\widetilde P_{h,i,t}\,\widehat u_{i,t+h}.
\label{eq:score_main}
\end{equation}
We estimate the long-run covariance matrix of the score process by
\begin{equation}
\widehat \Omega_h
=
\widehat\Gamma_h(0)
+
\sum_{k=1}^{L_h}
w_k
\left\{
\widehat\Gamma_h(k)+\widehat\Gamma_h(k)^{\top}
\right\},
\qquad
\widehat\Gamma_h(k)
=
\frac{1}{T-h}
\sum_{t=k+1}^{T-h}
\widehat s_{h,t}\widehat s_{h,t-k}^{\top},
\label{eq:hac_main}
\end{equation}
where $w_k=1-k/(L_h+1)$. We use the
standard Newey-West bandwidth based on the effective time dimension,
\(L_h=\lfloor 4[(T-h)/100]^{2/9}\rfloor\).\footnote{As a robustness check, we also use \(L_h=h\), which is natural in LPs because overlapping horizons induce MA(\(h\))-type dependence in the errors. Results are similar across the two HAC bandwidth choices, so we report the standard Newey-West bandwidth. The alternative lag-augmentation approach of \citet{MontielOleaPlagborgMoller2021}, adapted to micro-macro panels by \citet{AlmuzaraSancibrian2025MicroResponses}, is less suitable in our setting because applying it to all shock-sieve interactions would substantially increase dimensionality relative to the available time dimension.}

The covariance matrix of $\widehat b_h$ is estimated by
\begin{equation}
\widehat V_h
=
\frac{1}{N(T-h)}
\widehat{\widetilde A}_{11,h}^{-1}
\widehat \Omega_h
\widehat{\widetilde A}_{11,h}^{-1},
\label{eq:Vhat_main}
\end{equation}
where $\widehat{\widetilde A}_{11,h}
=\widehat A_{11,h}-\widehat A_{12,h}\widehat A_{22,h}^{-1}\widehat A_{21,h}$.
\label{eq:schur_main}

The pointwise variance of $\widehat g_h(z)$ is
\begin{equation}
\widehat\sigma_h^2(z)
=
\phi_{\widehat J_h}(z)^{\top}
\widehat V_h
\phi_{\widehat J_h}(z).
\label{eq:ptwise_var}
\end{equation}
Thus, for fixed $z$ on the empirical support of $Z_{i,t-1}$, the $(1-\alpha)$ pointwise confidence interval is
\begin{equation}
\mathcal I_h(z)
=
\left[
\widehat g_h(z)
\pm
z_{1-\alpha/2}\widehat\sigma_h(z)
\right],
\label{eq:pointwise_ci_main}
\end{equation}
where $z_{1-\alpha/2}$ is the standard normal critical value.

For uniform inference, let $\mathcal Z_G$ be a fine grid covering the empirical
support of $Z_{i,t-1}$. For each bootstrap draw $b=1,\ldots,B$, generate $\xi_h^{(b)}\sim \mathcal N(0,I_{\widehat J_h})$, and form the Gaussian process
\begin{equation}
G_h^{(b)}(z_m)
=
\phi_{\widehat J_h}(z_m)^{\top}
\widehat V_h^{1/2}
\xi_h^{(b)},
\qquad z_m\in\mathcal Z_G.
\label{eq:bootstrap_process}
\end{equation}
Compute the studentized supremum statistic
\begin{equation}
T_h^{(b)}
=
\max_{z_m\in\mathcal Z_G}
\left|
\frac{G_h^{(b)}(z_m)}
{\widehat\sigma_h(z_m)}
\right|,
\label{eq:bootstrap maximum}
\end{equation}
and let
$c_{h,1-\alpha}^{\mathrm{stu}}$
be the empirical $(1-\alpha)$ quantile of
$\{T_h^{(b)}\}_{b=1}^B$. The $(1-\alpha)$ uniform confidence band is
\begin{equation}
\mathcal C_h(z)
=
\left[
\widehat g_h(z)
\pm
c_{h,1-\alpha}^{\mathrm{stu}}\widehat\sigma_h(z)
\right],
\qquad z\in\mathcal Z.
\label{eq:ucb_main}
\end{equation}
$\mathcal I_h(z)$ gives pointwise coverage at a fixed $z$, whereas
$\mathcal C_h(\cdot)$ gives simultaneous coverage over
$\mathcal Z$.\footnote{For cumulative specifications where the outcome is
$Y_{i,t+h}-Y_{i,t}$, the corresponding estimand is
$g_h(z)-g_0(z)$. Inference proceeds analogously by estimating the LP with
the transformed outcome.}

Studentization is important in our setting because the response estimator can converge at different rates across the state space. Scaling by the state-specific standard error keeps the Gaussian process nondegenerate across these rate regimes and allows a single uniform band to provide simultaneous coverage over the state space.

\subsection{Asymptotic Properties}
\label{sec:asymptotics}

We summarize the main theoretical results in the paper; formal conditions and proofs are provided in the Online Appendix.

To state these results, we maintain
Assumptions~\ref{assum1}-\ref{assum3}, so causal identification
remains as established in Section~\ref{identification}, and impose the
additional regularity conditions collected in Online Appendix
Assumption~A.1. The sampling conditions allow cross-sectional
dependence generated by aggregate innovations. Conditional on the full
path of these innovations, the unit-level covariate and LP-error
processes are independent across units. This conditional independence
does not, however, require the long-run covariance of the
\(N^{-1/2}\)-normalized effective score to remain bounded: aggregate
components of the LP error may generate nonzero conditional score means
and an effective-score covariance that grows with \(N\).

The conditions also accommodate time-varying conditional variance of
the aggregate shock and serial correlation in LP errors, including
correlation induced by overlapping horizons. The remaining requirements
concern moments and short-range temporal dependence, regressor-error
orthogonality, the possibly \(N\)-dependent long-run covariance of the
effective score, regularity of the sieve basis, and well-conditioned
population design matrices.

A key feature of these asymptotics is that \(N\) and \(T\) play
different roles. Cross-sectional averaging reduces idiosyncratic noise,
whereas strong aggregate components of the LP error are reduced only by
time-series averaging. Consequently, the effective sample size
can range from order \(NT\) under weak effective-score dependence to order
\(T\) under strong effective-score dependence. For the growing-dimensional
sieve coefficient vector, this corresponds to stochastic error of order
\(\sqrt{J/(NT)}\) and \(\sqrt{J/T}\), respectively. Because
the importance of aggregate uncertainty can vary across the
state space and may disappear at particular states, the pointwise convergence rate may
also vary with \(z\). Pointwise inference uses the corresponding
state-specific standard error, while uniform inference studentizes the
process so that the Gaussian approximation remains nondegenerate across
these different rate regimes.

We assume that $g_h(\cdot)$ is H\"older smooth on the support $\mathcal Z$ with smoothness parameter $\kappa>1/2$. This smoothness yields a sieve approximation error of order $J^{-\kappa}$ and the coefficient-bias term $J^{1/2-\kappa}$ appearing in Theorem~\ref{thm:pointwise_main}. Uniform inference additionally requires the stronger joint conditions stated in Online Appendix Theorem~A.3. Throughout the asymptotic analysis, the horizon $h$ is fixed, while $N$, $T$, and $J$ grow jointly subject to the rate restrictions stated below. We use $\Rightarrow$ to denote convergence in distribution.

Let \(s_{h,t}\) denote the population counterpart of
\(\widehat s_{h,t}\) in \eqref{eq:score_main}, and define its
finite-\((J,N)\) long-run covariance by
$\Omega_{h,J,N}
\coloneqq
\sum_{k\in\mathbb Z}
\E\!\left[
s_{h,t}s_{h,t-k}^{\top}
\right]$.
We retain the \(N^{-1/2}\) normalization of \(s_{h,t}\), but do not
require \(\Omega_{h,J,N}\) to remain bounded. Define
$\chi_{h,J,N}
\coloneqq
1\vee\lambda_{\max}(\Omega_{h,J,N})$. Let $A_{11,h}$ denote the population counterpart of
$\widehat A_{11,h}$ and define
$
V_h^\star
\coloneqq
A_{11,h}^{-1}
\Omega_{h,J,N}
A_{11,h}^{-1},
\sigma_h^{\star 2}(z)
\coloneqq
\phi_J(z)^\top
V_h^\star
\phi_J(z).
$
Accordingly, $\widehat\sigma_h(z)$ estimates
$\sigma_h^\star(z)/\sqrt{N(T-h)}$.
The HAC estimator
\(\widehat\Omega_h\) estimates the long-run covariance at its actual scale,
so no estimate of \(\chi_{h,J,N}\) is required for implementation.

\begin{theorem} \label{thm:pointwise_main}
Under Online Appendix Assumption~A.1, including $J\to\infty$, $J^2 / (NT)\to0$, $J\chi_{h,J,N}/(NT)\to0 $ and $J/T\to0$ as $N,T \to \infty$, the sieve coefficient estimator satisfies
\[
\|\widehat b_h-b_h\|_2
=
O_p\!\left(J^{1/2-\kappa}\right)
+
O_p\!\left(
\sqrt{\frac{J\chi_{h,J,N}}{NT}}
\right).
\]
If, in addition, the state-specific undersmoothing and
matrix-linearization conditions in Online Appendix
Theorem~A.1 hold, together with the feasible directional
covariance-consistency conditions in Online Appendix Theorem~A.2,
then, for each fixed \(z\in\mathcal Z\),
\[
\frac{\widehat g_h(z)-g_h(z)}{\widehat\sigma_h(z)}
\Rightarrow
\mathcal N(0,1).
\]
Consequently, the pointwise confidence interval $\mathcal I_h(z)$ in \eqref{eq:pointwise_ci_main} satisfies
\[
\P\!\left\{g_h(z)\in\mathcal I_h(z)\right\}
\to 1-\alpha.
\]
\end{theorem}
\bigskip
The two terms in the coefficient rate make the bias-variance trade-off explicit. Increasing $J$ reduces the sieve approximation bias $J^{1/2-\kappa}$, where $\kappa > 1/2$, but raises the sampling uncertainty $\sqrt{J\chi_{h,J,N}/(NT)}$.
{The dependence scale \(\chi_{h,J,N}\) implies an effective sample size of order \(NT/\chi_{h,J,N}\): bounded \(\chi_{h,J,N}\) corresponds to the root-\(NT\) regime, whereas \(\chi_{h,J,N}\asymp N\) corresponds to the root-\(T\) regime.}
Subject to the upper growth restrictions on $J$, the exact
state-specific undersmoothing condition is
$
\sqrt{N(T-h)}\,J^{1-\kappa}/\sigma_h^\star(z)\to0.
$
At states for which
$\sigma_h^\star(z)\asymp\sqrt{J\chi_{h,J,N}}$,
this reduces to
$
\sqrt{NT/\chi_{h,J,N}}\,J^{1/2-\kappa}\to0.
$
At other states, including switching points, the exact
state-specific condition must be used.

The HAC result applies to the feasible score used in practice: the Online Appendix shows that replacing the population errors and population partialling terms by their estimated counterparts does not affect the long-run covariance asymptotically.

\begin{theorem} \label{thm:uniform_main}
Under Online Appendix Assumption~A.1 and the additional conditions for uniform convergence and Gaussian approximation in Online Appendix Theorem~A.3,
\[
\sup_{z\in\mathcal Z}
|\widehat g_h(z)-g_h(z)|
=
O_p\!\left(J^{1-\kappa}\right)
+O_p\!\left(
J\sqrt{\frac{\chi_{h,J,N}}{NT}}
\right).
\]
If the feasible covariance-kernel consistency condition \textup{(U-Feasible)} in Online Appendix Remark~3 also holds, the uniform confidence band in \eqref{eq:ucb_main} satisfies
\[
\P\!\left\{
g_h(z)\in\mathcal C_h(z)
\text{ for all }z\in {\mathcal Z}
\right\}
\to 1-\alpha.
\]
\end{theorem}
\bigskip
Uniform inference requires stronger conditions than pointwise inference because it controls the entire response function rather than the response at a fixed state. 
{A central complication is that a common normalization need not remain nondegenerate when the convergence rate changes across states. Studentizing by the state-specific standard error resolves this problem and yields a valid Gaussian approximation to the supremum of the normalized estimation error.}
The proof combines a high-dimensional Gaussian approximation for the dependent sieve score with {a finite-net approximation and modulus-of-continuity bounds} to pass from a finite net to the continuum; the general Gaussian approximation result used in this step is stated in the Online Appendix, with its proof in the accompanying Technical Note. 
{The grid $\mathcal Z_G$ used in implementation is the numerical counterpart of this finite-net approximation, while the asymptotic coverage result applies to the full state space $\mathcal Z$.}

It is important to note that the formal asymptotic theory takes $J$ as a sequence satisfying the stated rate conditions and thus abstracts from uncertainty introduced by data-driven selection. The AIC-selected implementation is treated as a practical tuning rule and evaluated in the Monte Carlo experiments, together with alternative selectors.
\section{Simulation Evidence} \label{sec:simulations}

All simulations are based on a DGP satisfying Assumption 3. We run $300$ Monte Carlo replications of data generated as
\begin{equation} \label{Simu: DGP}
Y_{it} = g(Z_{i,t-1})X_t + 0.8\, Y_{i,t-1} + \varepsilon_{it}, \quad i=1,\ldots,N;\ t=1,\ldots,T,
\end{equation}
where $X_t \sim \mathcal{N}(0,1)$ and $\varepsilon_{it} \sim \mathcal{N}(0,1)$. The state variable evolves as
\[
Z_{it} = \mu_i + \xi_{it}, \quad \mu_i \sim \mathcal{N}(0,3), \quad \xi_{it} = 0.8\, \xi_{i,t-1} + \sqrt{1-0.8^2}\, v_{it}, \quad v_{it} \sim \mathcal{N}(0,1).
\]
We discard the first $500$ observations and retain $T$. Results are robust to including fixed effects (not reported). The sequences $\{X_t\}$, $\{\varepsilon_{it}\}$, and $\{v_{it}\}$,
and the unit effects $\{\mu_i\}$, are mutually independent, with the
unit-level variables independent across $i$ and the innovations
independent over $t$.

We focus on the cubic specification
\begin{equation} \label{Simu: true cubic g function}
g(z) = 0.5 z + 0.3 z^2 - 0.25 z^3.
\end{equation}
The corresponding horizon-$h$ causal impulse response is
\begin{equation} \label{true IRF} 
\mathrm{IRF}_Y^{\mathrm{true}}(h,\delta \mid Z_{i,t-1}=z) = 0.8^h g(z)\,\delta.
\end{equation}

We estimate the LP with nonparametric state dependence
\[
Y_{i,t+h} = g_h (Z_{i,t-1})X_t + \gamma_h Y_{i,t-1} + u_{i,t+h}, \quad h = 0,\ldots,H,
\]
using the spline sieve estimator described in Section \ref{sec:estimation}. The estimated impulse response is
\begin{equation} \label{sieve IRF} 
\mathrm{IRF}_Y^{\mathrm{sieve}}(h,\delta \mid Z_{i,t-1}=z) = \widehat{g}_h(z)\delta.
\end{equation}

We select the sieve dimension $J$ using AIC, GCV, and LASSO, alongside an oracle choice ($J=4$). For AIC and GCV, we search over $\mathcal{J}=\{4,\ldots,20\}$, with GCV given by
\[
\mathrm{GCV}_h(J) = 
\frac{\sum_{i,t}\widehat u_{i,t+h}(J)^2}{N(T-h)\left(1-\frac{K_J}{N(T-h)}\right)^2},
\quad 
\widehat J_h = \arg\min_{J\in\mathcal J}\mathrm{GCV}_h(J),
\]
where $K_J = J + 1$. For LASSO, we set a maximal dimension $J_{\mathrm{lasso}}=50$ and
estimate
\[
(\widehat b_h,\widehat\gamma_h)
=
\arg\min_{b_h,\gamma_h}
\left\{
\sum_{i,t}
\left(
Y_{i,t+h}
-
X_t\phi_{J_{\mathrm{lasso}}}(Z_{i,t-1})^\top b_h
-
\gamma_hY_{i,t-1}
\right)^2
+
\lambda_h\|b_h\|_1
\right\},
\]
where $\lambda_h$ is chosen by cross-validation. The selected dimension
$\widehat J_h$ is the number of nonzero elements of $\widehat b_h$.

\subsection{Sieve Estimation: Gains over Linear LPs}

We evaluate the performance of the sieve estimator for the state-dependent impulse response and its gains relative to standard LPs with linear state dependence. Under the DGP, which satisfies Assumption 3, the causal impulse response is $g_h(z)$, so the simulation isolates the role of functional-form restrictions. We set $N=500$, $T=200$, and select the number of basis functions using AIC.

For comparison, we estimate the LP with linear state dependence
\[
Y_{i,t+h} = (\alpha_h + \beta_h Z_{i,t-1})X_t + \gamma_h Y_{i,t-1} + u_{i,t+h},
\]
which implies
\begin{equation} \label{linear IRF} 
\mathrm{IRF}_Y^{\mathrm{linear}}(h,\delta \mid Z_{i,t-1}=z) = (\widehat{\alpha}_h + \widehat{\beta}_h z)\delta.
\end{equation}

Figure \ref{fig:spline-sieve-vs-linear} reports Monte Carlo averages and 10\% to 90\% percentile bands for the sieve and linear IRFs, together with the true response, for a unit shock ($\delta=1$) over a grid of $500$ points in $[-4.65,4.65]$.

\begin{figure}[htbp]
\centering
\includegraphics[width=\textwidth]{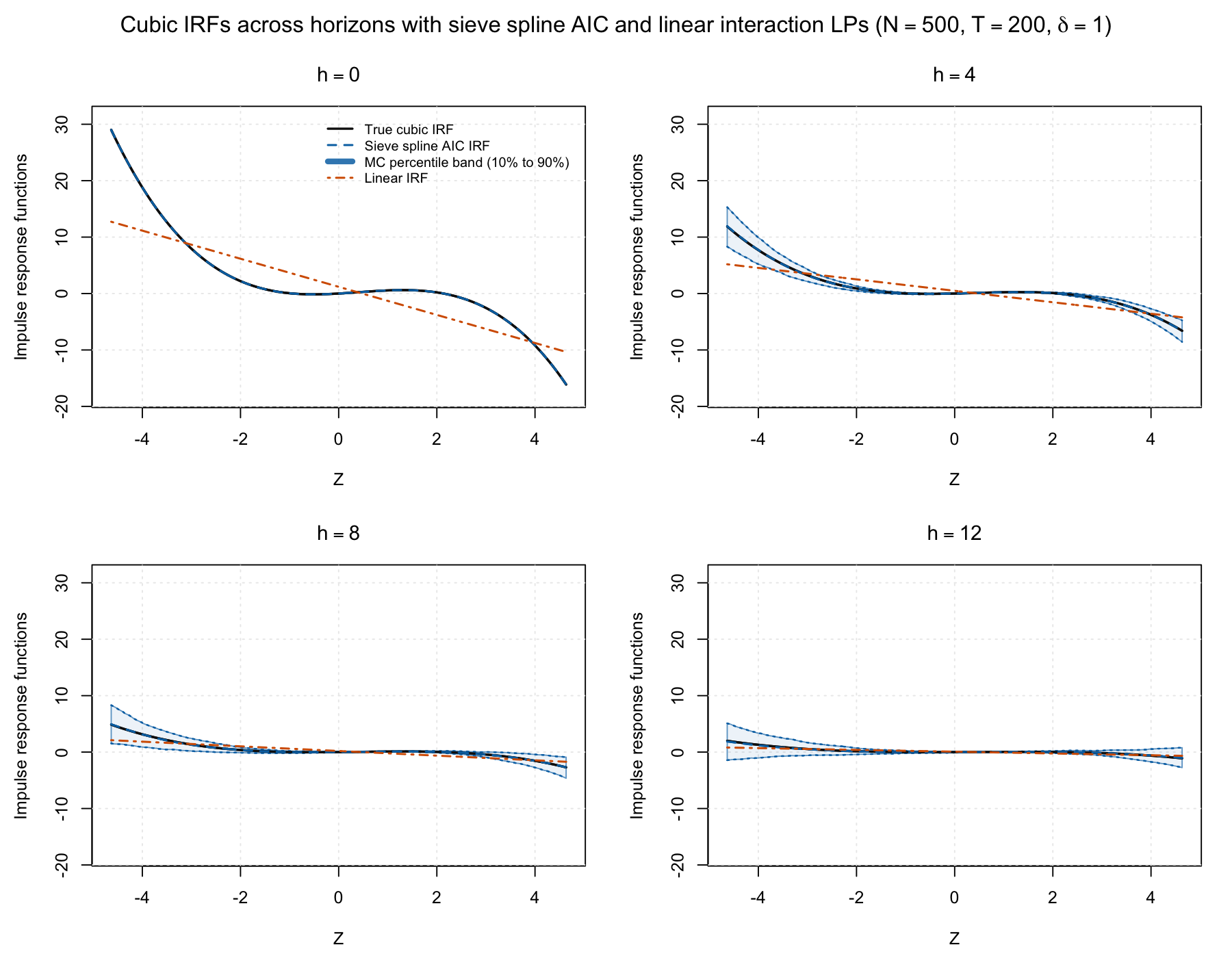}
\caption{Monte Carlo comparison of sieve and linear IRFs under the cubic DGP}
\label{fig:spline-sieve-vs-linear}
\end{figure}

Figure~\ref{fig:spline-sieve-vs-linear} shows that the linear specification is substantially biased, especially in the tails. The spline sieve instead tracks the true IRF closely across the support of $z$, with tight Monte Carlo dispersion.
The linear and sieve estimators can imply different, even opposite, economic conclusions. The discrepancy narrows at longer horizons as the true response is attenuated toward zero, but the key distinction remains one of estimands: the sieve targets the causal state-dependent IRF, whereas the linear specification generally recovers only a projection of it.

These results are not driven by the choice of basis. Additional simulation results in the Online Appendix show that the same conclusions hold when the true $g(z)$ is generated by a Fourier series but estimated using splines.

We next quantify the gains of nonparametric estimation using the root integrated mean squared error (RIMSE). Table \ref{tab:rimse_linear_vs_spline_aic} compares the sieve IRF to the linear IRF for different horizons $h$ and impulse sizes $\delta \in \{0.5\sigma_X,\sigma_X,2\sigma_X\}$, where $\sigma_X$ is the standard deviation of the shock $X_t$ (here equal to 1).

\begin{table}[htbp]
\centering
\caption{RIMSE comparison of sieve and linear IRFs}
\label{tab:rimse_linear_vs_spline_aic}
\setlength{\tabcolsep}{5pt}
\renewcommand{\arraystretch}{1.05}
\begin{tabular}{c *{3}{S[table-format=1.3] S[table-format=1.3]}}
\toprule
\multicolumn{1}{c}{$\delta$} & \multicolumn{2}{c}{$0.5\sigma_X$} & \multicolumn{2}{c}{$\sigma_X$} & \multicolumn{2}{c}{$2\sigma_X$} \\
\cmidrule(lr){2-3}\cmidrule(lr){4-5}\cmidrule(lr){6-7}
\multicolumn{1}{c}{$h$} & {Sieve} & {Linear} & {Sieve} & {Linear} & {Sieve} & {Linear} \\
\midrule
0  & 0.004 & 2.253 & 0.009 & 4.506 & 0.018 & 9.011 \\
4  & 0.418 & 0.981 & 0.836 & 1.962 & 1.672 & 3.923 \\
8  & 0.402 & 0.499 & 0.804 & 0.997 & 1.609 & 1.995 \\
12 & 0.397 & 0.342 & 0.795 & 0.683 & 1.589 & 1.366 \\
\bottomrule
\end{tabular}
\end{table}

The RIMSE results mirror the visual evidence: the sieve estimator improves on the linear specification when nonlinearities are pronounced, but its advantage vanishes as the response becomes nearly linear. At \(h=12\), the linear specification performs slightly better. At horizons where nonlinearities matter, the gains also widen with $\delta$, indicating that the cost of misspecifying nonlinear state dependence becomes more pronounced for larger shocks.  Overall, the findings confirm that accounting for nonlinearities is important for accurately recovering the impulse response.

\subsection{Sieve Inference: Uniform-band Performance}

We evaluate the finite-sample performance of sieve-based uniform confidence bands across sample sizes $(N,T) \in \{1,100,500\} \times \{40,120,200\}$ and horizons $h \in \{0,4,8,12\}$. For the spline sieve, we consider an oracle specification ($J=4$) and data-driven choices of $J$ based on AIC, GCV, and LASSO. The main text focuses on AIC; results for all selectors are reported in the Online Appendix.

Given $J$, we construct uniform confidence bands for $g_h(z)\delta$ (with $\delta=1$) using \eqref{eq:bootstrap_process}-\eqref{eq:ucb_main} with $B=2000$ bootstrap draws and nominal coverage $1-\alpha=0.95$. Performance is evaluated over a grid of $500$ points spanning the 5th to 95th percentiles of $\{ Z_{it} \}$ using coverage (fraction of replications in which the true IRF lies within the band over the grid) and average band width.\footnote{Coverage is lower when evaluated over the full empirical support, reflecting the greater difficulty of nonparametric uniform inference in sparsely populated boundary regions.}



Table \ref{tab:mc_dynamic_cubic_spline_aic} reports results for
$N \in \{100,500\}$ and $T \in \{40,120,200\}$ at horizons
$h \in \{0,4\}$ under AIC selection; additional results are reported in
the Online Appendix. Coverage generally improves with $T$, while band
widths decline. Differences across $N$ are limited, with the larger $N$
mainly producing narrower bands. Without intermediate-shock controls, the bands at $h=4$ are substantially
wider than at impact because intermediate aggregate shocks remain in the LP error.

\begin{table}[htbp]
\centering
\caption{Uniform-band performance under specifications without and with intermediate shocks for the sieve with AIC selection at nominal coverage $1-\alpha=0.95$}
\label{tab:mc_dynamic_cubic_spline_aic}
\vspace{0.2cm}
\setlength{\tabcolsep}{3.5pt}
\renewcommand{\arraystretch}{1.05}
\begin{tabular}{l *{4}{S} *{4}{S}}
\toprule
& \multicolumn{4}{c}{No Intermediate Shocks}
& \multicolumn{4}{c}{Intermediate Shocks} \\
\cmidrule(lr){2-5}\cmidrule(lr){6-9}
& \multicolumn{2}{c}{$N=100$}
& \multicolumn{2}{c}{$N=500$}
& \multicolumn{2}{c}{$N=100$}
& \multicolumn{2}{c}{$N=500$} \\
\cmidrule(lr){2-3}\cmidrule(lr){4-5}
\cmidrule(lr){6-7}\cmidrule(lr){8-9}
& {Coverage} & {Width}
& {Coverage} & {Width}
& {Coverage} & {Width}
& {Coverage} & {Width} \\
\midrule
\multicolumn{9}{l}{\textit{Panel A: Horizon $h=0$}} \\
\midrule
$T = 40$
& 0.73 & 0.14
& 0.72 & 0.06
& 0.73 & 0.14
& 0.72 & 0.06 \\
$T = 120$
& 0.83 & 0.08
& 0.84 & 0.04
& 0.83 & 0.08
& 0.84 & 0.04 \\
$T = 200$
& 0.88 & 0.06
& 0.86 & 0.03
& 0.88 & 0.06
& 0.86 & 0.03 \\
\midrule
\multicolumn{9}{l}{\textit{Panel B: Horizon $h=4$}} \\
\midrule
$T = 40$
& 0.73 & 3.13
& 0.75 & 2.56
& 0.60 & 0.22
& 0.67 & 0.10 \\
$T = 120$
& 0.82 & 2.01
& 0.87 & 1.64
& 0.83 & 0.13
& 0.81 & 0.06 \\
$T = 200$
& 0.92 & 1.66
& 0.89 & 1.30
& 0.86 & 0.10
& 0.83 & 0.05 \\
\bottomrule
\end{tabular}
\end{table}

Overall, the AIC-based bands exhibit finite-sample undercoverage, particularly when $T$ is small. Additional simulation results in the Online Appendix show that coverage is closer to nominal for the oracle selector, indicating an important finite-sample role for tuning-parameter selection, although some undercoverage remains even under the oracle. Coverage generally improves as the time dimension increases.

The table also illustrates the role of common LP-error components emphasized by the asymptotic theory in Section~\ref{sec:asymptotics}. For $h\geq1$, omitting intermediate terms driven by future aggregate shocks does not affect identification of $g_h(\cdot)$, but leaves common variation in the LP error that is not eliminated by cross-sectional averaging. Including these interactions therefore makes increases in $N$ more effective at reducing sampling uncertainty, as reflected by the sharp decline in band widths in Table~\ref{tab:mc_dynamic_cubic_spline_aic}. We use this comparison to illustrate the efficiency implications of common LP-error components, but retain the specification without intermediate-shock controls as the baseline elsewhere. Interacting these additional shocks with the sieve basis substantially increases dimensionality and can lead to numerical instability.

\subsection{Sieve Robustness: Sensitivity to Selector Choice}

We assess the sensitivity of sieve estimation to the choice of selector for the number of basis functions, using the design with $N = 500$ and $T = 200$ at horizon $h=4$. Estimation results are qualitatively similar across selectors and are not reported. Table \ref{tab:mc_dynamic_h4_n500_t200} reports the corresponding inference results. The oracle specification delivers coverage closest to the nominal level, while AIC, GCV, and LASSO perform similarly, with modest differences in coverage and width. Additional results are reported in the Online Appendix. Overall, the results are fairly robust to selector choice, although data-driven selection entails some finite-sample coverage loss relative to the oracle.

\begin{table}[htbp]
\centering
\caption{Uniform-band performance across selector choices at horizon $h=4$ for $N=500$ and $T=200$ under nominal coverage $1-\alpha=0.95$}
\label{tab:mc_dynamic_h4_n500_t200}
\setlength{\tabcolsep}{5pt}
\renewcommand{\arraystretch}{1.05}
\begin{tabular}{l S S}
\toprule
Selector & {Coverage} & {Width} \\
\midrule
Oracle & 0.92 & 1.17 \\
AIC    & 0.89 & 1.30 \\
GCV    & 0.89 & 1.30 \\
LASSO  & 0.90 & 1.42 \\
\bottomrule
\end{tabular}
\end{table}

\section{Monetary Policy and Firm Investment} \label{sec:application}

To demonstrate the empirical relevance of our approach, we revisit firm-level investment responses to monetary policy shocks, as in \citet{ottonello2020financial}, \citet{cloyne2023monetary}, and \citet{jeenas2023firm}, among others. Our nonparametric state-dependent LP recovers a causal state-dependent response to the extent that the monetary policy shock satisfies the causal exogeneity conditions in Assumption~1 and Assumption~\ref{assum3} is satisfied: monetary policy affects firms through heterogeneous balance-sheet exposure, while propagation dynamics are well approximated as stable over the relevant horizon. The application also illustrates the limitation of the standard linear specification used in this literature: even under these conditions, it estimates a projection rather than the causal state-dependent response.

Following \citet{ottonello2020financial} (OW), we use distance to default
(D2D) - an estimate of the probability of default using leverage ratios and
equity price volatility - as our measure of firm financial conditions. We
compute D2D following the approach in \citet{gilchrist2012credit}.\footnote{
D2D is equal to the log ratio of firm value to debt, adjusted for the volatility
of firm value. It can be interpreted as the number of standard deviations by
which the log ratio of firm value to debt must fall below its mean for the firm
to default. Although OW also study leverage, we focus on D2D for two main
reasons: first, there is a large empirical literature documenting the
performance of D2D as an indicator of default risk
\citep{duffie2007multi,schaefer2008structural,duffie2011measuring,atkeson2017measuring};
second, OW find weaker results using leverage, perhaps because it is a less
precise indicator of firms' financial soundness.} Firms with high D2D are at
lower risk of default and hence safer, i.e., more financially sound. We follow
OW in using monetary policy shocks identified from high-frequency changes in
Federal Funds rate futures in a narrow window around FOMC announcements,
measuring firm investment with Compustat quarterly data as the change in the log
book value of the firm's capital stock, \(k_{i,t}\equiv \log K_{i,t}\),
standardizing D2D in units of standard deviations relative to the sample mean,
and applying the same sample selection criteria. We estimate the following LP specification:
\begin{equation}
k_{i,t+h} - k_{i,t-1}
=
\alpha_{sth}
+
g_h\left(Z_{i,t-1}\right)X_t
+
\rho_h \Delta k_{i,t-1}
+
 W_{i,t-1}^{\top}\gamma_h
+
u_{i,t+h}.
\label{eq:main_spec}
\end{equation}
The dependent variable is the firm's cumulative net investment from quarter
\(t\) through \(t+h\), so \(g_h(\cdot)\) is interpreted as a cumulative
investment response at horizon \(h\). Here, \(\alpha_{sth}\) denotes
sector-time-horizon fixed effects; \(Z_{i,t-1}\) is the firm's pre-shock
financial condition, measured by D2D; \(X_t\) is the monetary policy shock;
\(\Delta k_{i,t-1}\) is lagged investment; \(W_{i,t-1}\) is a vector of
firm-level controls; and \(u_{i,t+h}\) is an error term.\footnote{As in OW, we include as controls
\(Z_{i,t-1}\), total assets, sales growth, current assets as a share of total
assets, a fiscal quarter indicator, and the interaction of \(Z_{i,t-1}\) with
lagged GDP growth. Our specification differs from OW in three ways: we control
for the firm's lagged investment rate; they use the firm's demeaned financial
position relative to its average level; they include firm fixed effects. We have
verified that these differences do not qualitatively affect our conclusions.} Because the
specification includes sector-time-horizon fixed effects, \(g_h(\cdot)\)
captures differential investment responses across firms with different
pre-shock financial conditions, net of the component common to firms in the
same sector and period. Accordingly, in this specification the constant component of $g_h(\cdot)$
is absorbed by the sector-time fixed effects, so we normalize the response
function relative to this common component.

We estimate two versions of \eqref{eq:main_spec}: a linear specification,
\(g_h(Z_{i,t-1})=\alpha_h+\beta_h Z_{i,t-1}\), and the flexible nonlinear
specification developed in Section~\ref{sec:estimation}. Figure~\ref{fig:IRF_1}
plots the resulting differential investment responses to a 25 basis point surprise interest
rate cut on impact and at horizons of four, eight, and 12 quarters. Each panel
shows the response as a function of firms' pre-shock financial conditions, with
uniform confidence bands around the nonlinear estimate constructed using \eqref{eq:bootstrap_process}-\eqref{eq:ucb_main}. Because the bands are studentized, their width varies with the state-specific
standard error. This state-specific scaling allows the construction to remain
valid when convergence rates differ across the state space.

\subsection{Nonlinear response heterogeneity}

Panel A of Figure \ref{fig:IRF_1} sharply illustrates our main findings. Imposing a linear specification yields a positive impact coefficient, which implies a monotonic upward-sloping relationship between a firm's investment response and financial condition - less risky, more financially sound firms with higher D2D are more responsive to monetary policy shocks. The result corroborates the influential findings in OW, who also estimate a positive coefficient. In contrast, the nonlinear IRF shows quite a different pattern: a non-monotonic hump-shaped response that peaks at approximately the mean of the distribution, rather than at the least risky right tail. Thus, although responsiveness is increasing for firms that are close to default, as found in previous work, the relationship turns negative around the mean, implying that there is a large portion of the distribution where riskier, lower distance to default firms are more responsive to monetary policy shocks, not less.\footnote{Due to right-skewness in the distribution, the median D2D is slightly below the mean.}


Although our econometric model does not reveal the precise economic forces driving these patterns, one interpretation is that firms close to default are investment constrained and respond little, firms near the middle of the distribution are most sensitive to interest rates, and the safest firms less so. Thus, our method uncovers a richer, more nuanced relationship between shock sensitivity and financial conditions that is masked by the linear specification, underscoring the importance of allowing for flexible forms of state-dependence when analyzing microeconomic responses to macroeconomic shocks.\footnote{Our results are qualitatively unchanged under various modifications to our main specification; for example, including firm fixed-effects somewhat dampens responsiveness across the distribution, but the hump-shaped pattern we uncover remains.}


\begin{figure}[htbp]
	\centering
    \includegraphics[scale=0.53]{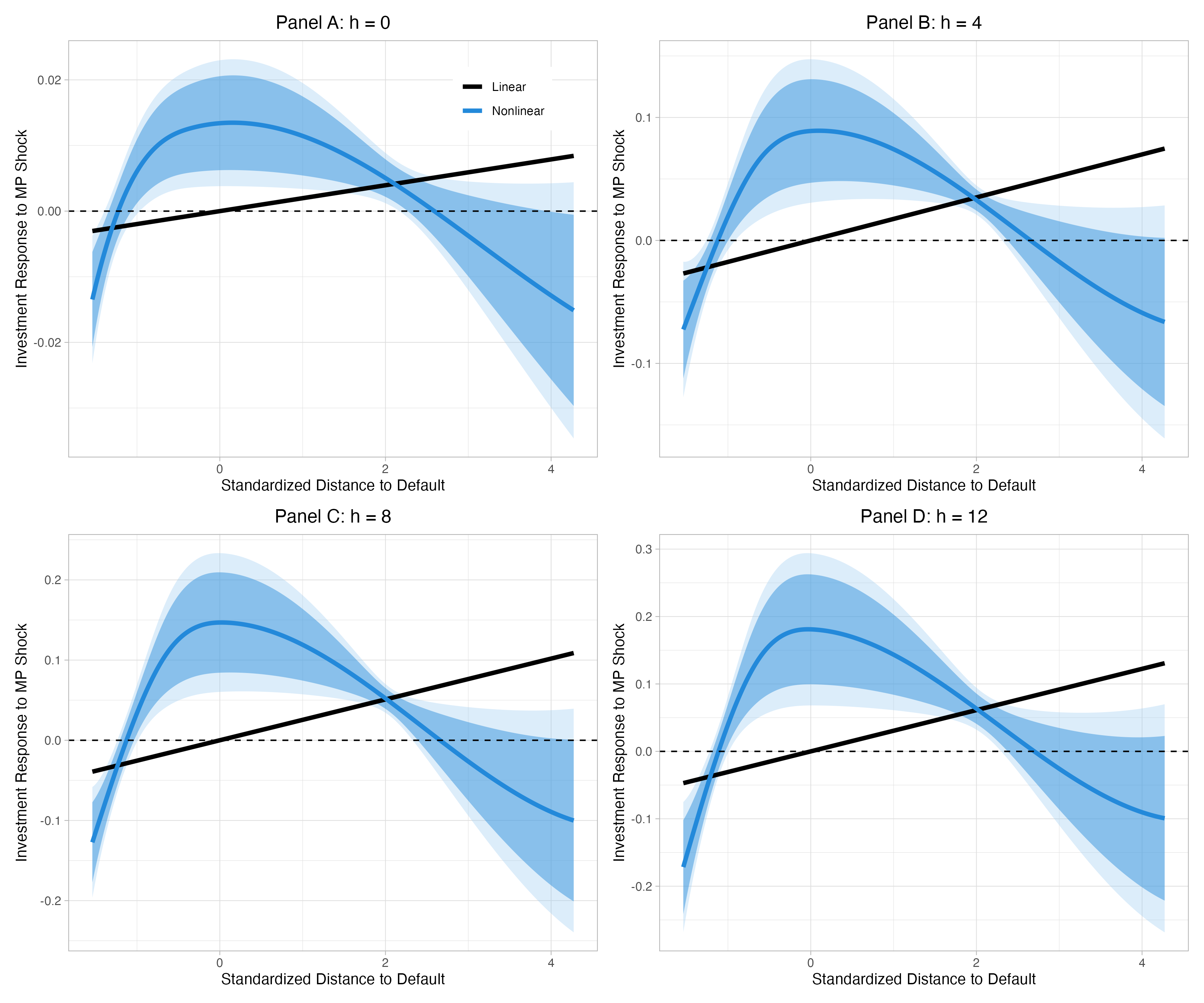}
	\caption{Impulse Response of Firm-Level Investment to a Monetary Policy Shock}
	\label{fig:IRF_1}
	\flushleft	\scriptsize \textit{Notes}: This figure displays local projection impulse responses to a 25 b.p. monetary policy shock across firms with different levels of distance to default. The units of distance to default are standard deviations relative to the sample mean. The linear specification imposes a linear interaction between the shock and distance to default. The nonlinear specification implements the nonlinear function of distance to default. Dark (light) shaded regions report 68\% (90\%) uniform confidence bands. 
\end{figure}
Panels B-D of Figure \ref{fig:IRF_1} report IRFs at four-quarter intervals up to 12 quarters after the shock. The non-monotonicity becomes more pronounced over time: firms near the center of the D2D distribution exhibit increasingly larger cumulative investment responses than those at either tail.\footnote{Although not shown, the cumulative IRF changes little beyond 12 quarters, indicating that the effects of the shock have largely dissipated.} In contrast, the linear specification misses these dynamics and suggests that the response of high-D2D firms continues to grow relative to that of firms closer to default.

To illustrate the economic implications of response heterogeneity, Figure \ref{fig:IRF_2} plots the differential in cumulative impulse responses across selected points of the D2D distribution under the linear and nonlinear specifications, which directly captures how the transmission of shocks varies across firms over the post-shock horizon. We report pointwise confidence intervals using \eqref{eq:ptwise_var} and \eqref{eq:pointwise_ci_main}, since the objects of interest are differences in impulse responses at specific states.\footnote{Specifically, we consider $\widehat g_h(Z_a) - \widehat g_h(Z_b)$, where $Z_a$ and $Z_b$ denote two fixed points of the state distribution, such as the 50th and 5th percentiles, or the 95th and 50th percentiles, of $Z_{i,t-1}$. Accordingly, instead of using the pointwise variance estimator in \eqref{eq:ptwise_var} directly, we use the variance
$( \phi_{\widehat J_h}(Z_a) - \phi_{\widehat J_h}(Z_b) )^{\top}
\widehat V_h
( \phi_{\widehat J_h}(Z_a) - \phi_{\widehat J_h}(Z_b) )$ and construct pointwise confidence intervals as in \eqref{eq:pointwise_ci_main}.} In Panel A, we compare a firm at the 50th percentile of the D2D distribution to one at the 5th percentile. Although both sets of estimates predict that the firm at the 50th percentile will invest more than the firm at the 5th percentile in response to the monetary policy shock, the linear estimates substantially underpredict the difference; for example, after 16 quarters, the linear estimates yield a cumulative investment differential of just under 5\%, whereas the nonlinear estimates suggest a cumulative differential exceeding 20\%. In Panel B we similarly compare a firm at the 95th percentile of the state distribution to the 50th; in this case, the two sets of estimates yield virtually opposite results: the linear estimates predict a positive differential of almost 10\% at 16 quarters following the shock, whereas the nonlinear estimates show that this difference is in fact negative and close to -10\%. Thus, the linear specification underestimates the strength of the positive relationship between responsiveness and D2D in the left tail of the distribution, and misses entirely on the negative relationship in the right tail. Similar results hold at the shorter horizons as well. In short, imposing a linear specification can lead to misleading conclusions regarding the investment response of firms across the financial distribution.


\begin{figure}[htbp]
	\centering
	\includegraphics[scale=0.53]{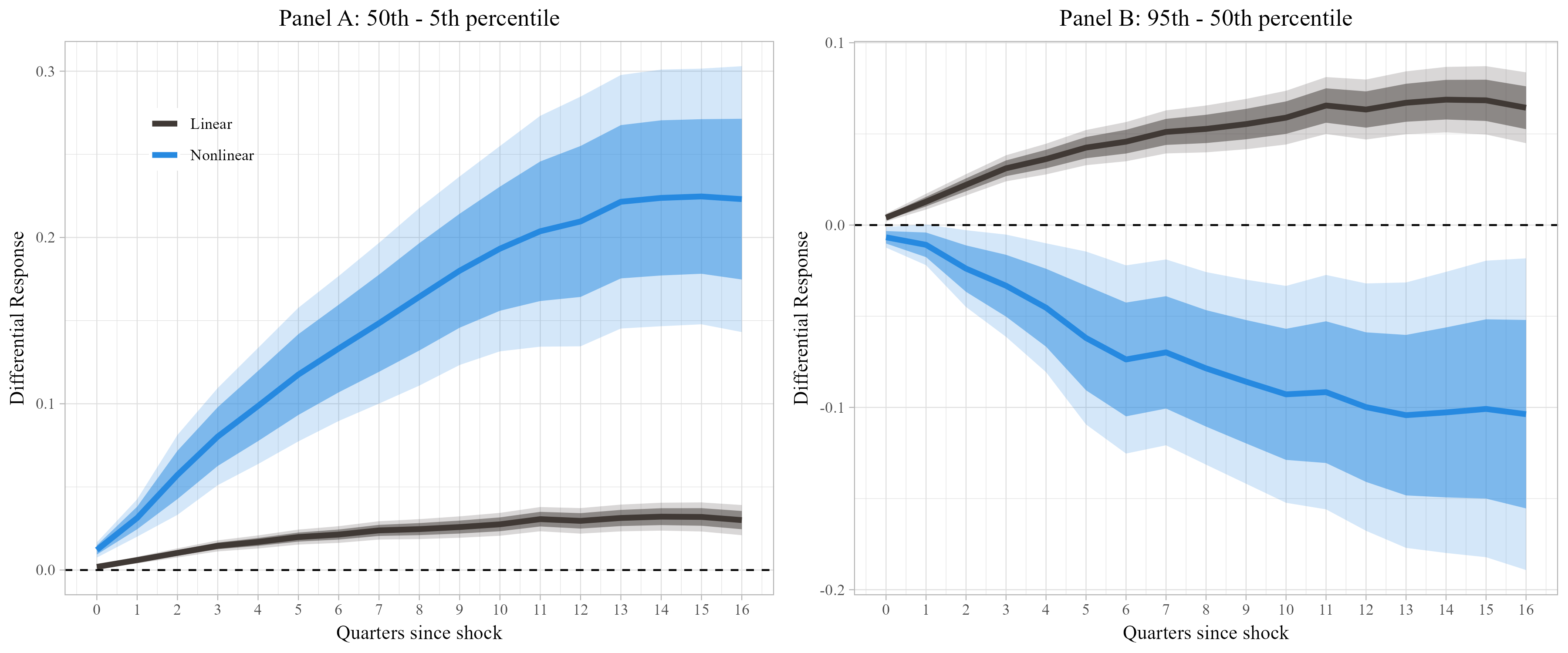}
	\caption{Differential Investment Responses to a Monetary Policy Shock}
	\label{fig:IRF_2}
	\flushleft	\scriptsize \textit{Notes}: This figure displays the differential in impulse responses (IRF) to a 25 b.p. monetary policy shock for firms with different levels of distance to default. The linear specification imposes a linear interaction between the shock and distance to default. The nonlinear specification implements the nonlinear function of distance to default. Panel A plots the difference in IRFs for a firm at the 50th percentile of the distance to default distribution relative to a firm at the 5th percentile; Panel B plots the difference in IRFs for a firm at the 95th percentile relative to a firm at the 50th percentile. Dark (light) shaded regions report 68\% (90\%) pointwise confidence intervals.
\end{figure}

\subsection{Aggregate implications of nonlinear responsiveness }

Our results have important quantitative implications for the role of financial heterogeneity in the transmission of monetary policy: because most firms are concentrated in regions of the distribution where the nonlinear response exceeds the linear approximation, the latter effectively averages across states and attenuates the estimated impact of monetary policy shocks.


We use the estimated cross-sectional response function to summarize aggregate
responsiveness through heterogeneity in firms' financial conditions. Specifically, we
average the estimated firm-level responses using firms' shares of aggregate capital at
the time of the shock:
\begin{equation*}
\Delta k_{t+h}
\ = \
\underbrace{
\sum_i
\frac{K_{i,t-1}}{\sum_j K_{j,t-1}}
g_h\left(Z_{i,t-1}\right)
}_{\text{aggregate responsiveness}}
\ X_t \;,
\end{equation*}
where \(\Delta k_{t+h}=\log K_{t+h}-\log K_{t-1}\) and
\(K_t=\sum_i K_{it}\) denotes the aggregate capital stock.\footnote{This calculation captures only the incremental component of the response associated with cross-sectional heterogeneity in firms' financial conditions. Because equation \eqref{eq:main_spec} includes sector-time fixed effects, the component of the monetary-policy response common to firms within a sector and period is absorbed by the fixed effects. The calculation therefore aggregates the estimated responses relative to the sector-time benchmark implicit in equation \eqref{eq:main_spec}, rather than the level aggregate response to monetary policy.}  As the expression highlights, aggregate responsiveness depends not only on individual responses, but the joint distribution of those responses with firm shares of aggregate capital.\footnote{A similar result is used in \cite{david2023rise} to decompose the impact of monetary policy on aggregate investment across tangible and intangible capital.}

We calculate aggregate responsiveness upon shock impact under both the linear and nonlinear specifications of $g_h\left(\cdot\right)$ using the estimates from Panel A of Figure \ref{fig:IRF_1} and the distribution of firm capital on a quarter-by-quarter basis. We plot the results in Figure \ref{fig:agg_inv}. There are two main takeaways: first, the linear specification substantially underestimates the sensitivity of aggregate investment to monetary policy. On average, the linear specification implies that financial heterogeneity adds about 0.09 percentage points (p.p.) to the responsiveness of aggregate investment. In contrast, the nonlinear specification implies that financial heterogeneity adds about 1.0 p.p., roughly an order of magnitude larger.\footnote{The result is consistent with previous work finding that aggregate non-residential private fixed investment tends to be highly sensitive to monetary policy relative to other categories of GDP \citep{david2023rise}.}  The finding implies that not only does the linear specification miss on important aspects of the distribution of micro-level responses, but also that those micro-level biases affect ``bottom-up'' calculations of the implications of micro-level responses for the macro-level response. The result stems largely from the fact that most firms are concentrated in regions of the distribution where the nonlinear response exceeds the linear approximation and hence the linear specification attenuates the estimated aggregate impact of the macroeconomic shock.

The second takeaway from Figure \ref{fig:agg_inv} is that aggregate responsiveness differs sharply over time across the two specifications, with a correlation of about -0.7. The linear specification implies weaker monetary policy precisely when the nonlinear specification implies stronger effects. As a stark example, consider the Great Financial Crisis (GFC) and its aftermath from late 2007 through 2009. As firms’ financial conditions deteriorated and more firms moved closer to default, the positive coefficient from the linear LP predicts that monetary policy would be less effective in stimulating investment. In contrast, because pre-recession financial health was unusually strong (mean D2D was high), worsening conditions moved firms toward and then below the long-run mean, where our estimates imply responsiveness is highest. Thus, the two approaches yield very different implications regarding the states of micro-level financial conditions when monetary policy is most effective: the linear estimates imply monetary policy is most effective at times when firms tend to be financially the strongest (responsiveness from the linear estimates also tends to be high when dispersion in D2D is high), whereas the nonlinear estimates imply monetary policy is most effective when firms tend to be closer to the long-run average level of D2D (responsiveness from the nonlinear estimates also tends to be high when dispersion in D2D is low). In sum, allowing for a flexible approach to estimating the impact of macro shocks on micro variables can be critical for properly assessing the macro-level implications of micro-level heterogeneity.

\begin{figure}[htbp]
	\centering
	\includegraphics[scale=0.45]{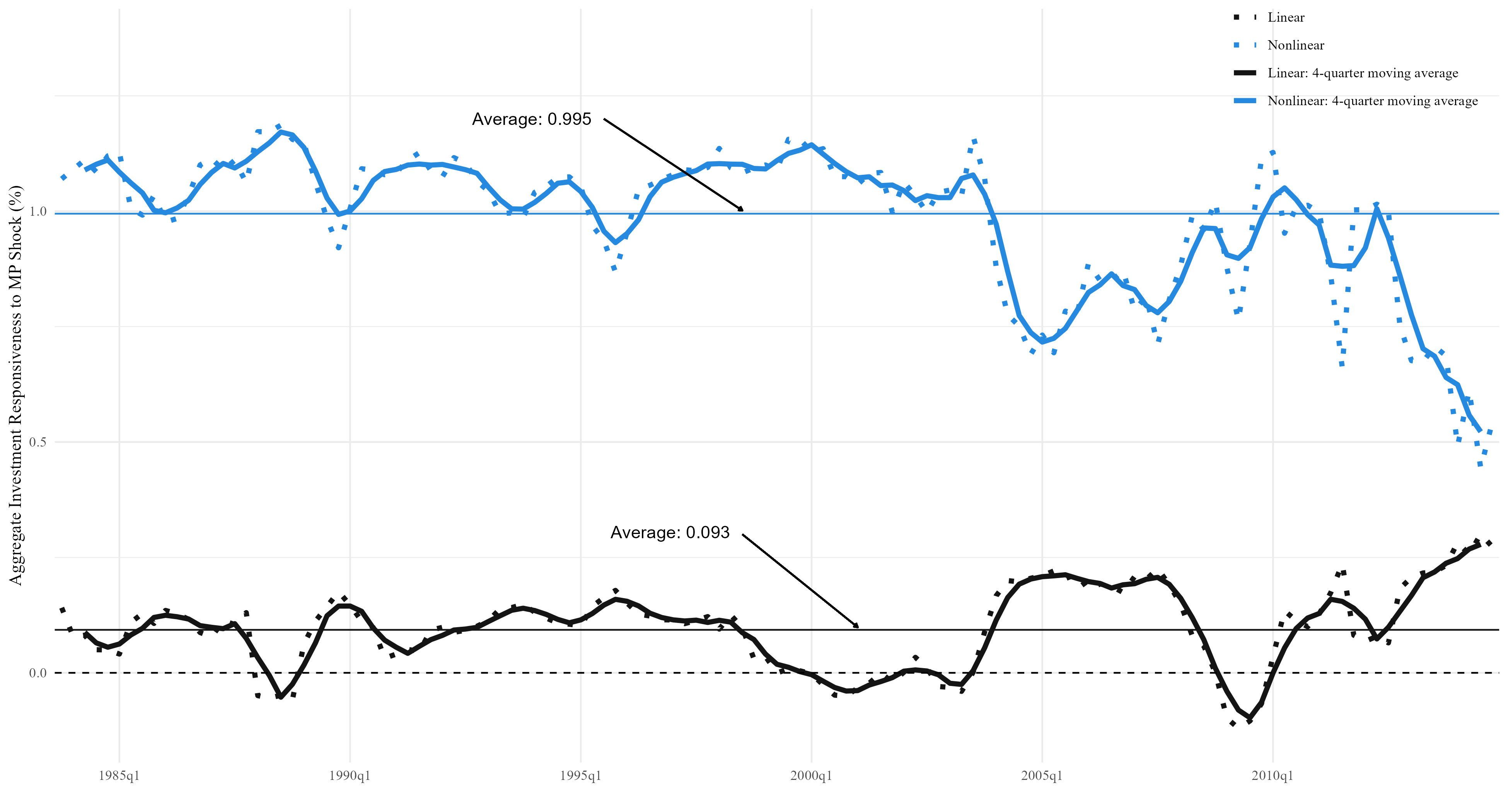}
	\caption{Aggregate Investment Response to a Monetary Policy Shock}
	\label{fig:agg_inv}
\end{figure}

\section{Conclusion}\label{sec:conclusion}

State-dependent LPs are widely used to study how causal responses vary across economic states, but that interpretation is not generally valid. The broader lesson is that state-dependent LPs can recover causal response functions in micro-macro environments with linear shock exposure and stable propagation, but doing so generally requires estimating state dependence nonparametrically rather than imposing a linear interaction. Our framework therefore recasts state-dependent LPs from interpreting shock-state interaction coefficients to identifying, estimating, and conducting valid inference on causal response functions in micro-macro panels.

\appendix

\section{Appendix: What LPs with Linear State Dependence Estimate}
\label{app:linear_misspec}

To characterise the consequences of misspecifying \(g_h(z)\) as linear, consider the simplified setup where the researcher estimates the LP with linear state dependence
\begin{equation}
\label{linearlemma1}
Y_{i,t+h}
=
(\alpha_h +\beta_h Z_{i,t-1})X_t+u_{i,t+h},
\end{equation}
while the true state-dependence is nonlinear:
\begin{equation}
\label{mainlemma1}
Y_{i,t+h}
=
g_h(Z_{i,t-1})X_t+u_{i,t+h}.
\end{equation}

We analyse the resulting misspecification in Lemma~\ref{misspecification beta}, which characterizes the coefficient $\beta_h$.

\begin{assumption}
\label{assumption DGP}
\leavevmode \\
\((a)\) \(\{(Z_{i,t-1},X_t)\}\) satisfies: \(\{Z_{i,t-1}\}\) is independent and identically distributed across \(i\) and stationary in \(t\); \(\{X_t\}\) is stationary; \((Z_{i,t-1},X_t)\) has joint distribution \(F_{Z,X}\).  \\[4pt]
\((b)\) \(Z_{i,t-1}\) has a continuous distribution with compact support \(\mathcal Z=[\underline z,\bar z]\subset\mathbb R\). \\[4pt]
\((c)\) \(g_h(z)\) is continuously differentiable on \(\mathcal Z\). \\[4pt]
\((d)\) The regression error satisfies \(\mathbb E[X_tu_{i,t+h}]=0\) and \(\mathbb E[Z_{i,t-1}X_tu_{i,t+h}]=0\).
\end{assumption}

\begin{assumption}
\label{assumption moments}
\leavevmode \\
\((a)\) \(Z_{i,t-1}\) and \(X_t\) have finite fourth moments; the second-moment matrix of \((X_t,Z_{i,t-1}X_t)\) is nonsingular. \\[4pt]
\((b)\)
\( \mathbb E\!\left[ |g_h(Z_{i,t-1})X_t| \{1+|X_t|+|Z_{i,t-1}X_t|\} \right]<\infty. \) \\[4pt] 
\((c)\) \(\int_{\mathcal Z}|\omega_h(z)g_h'(z)|\,dz<\infty\), where \(\omega_h(z)\) is defined in Lemma~\ref{misspecification beta}.
\end{assumption}

The conditions imposed here are local to Lemma~\ref{misspecification beta}.
They weaken the main identifying restrictions while retaining the moment and
smoothness conditions needed to characterize the population linear projection
coefficient. Assumption~\ref{assumption DGP}\textup{(d)} replaces the
conditional mean restriction with the orthogonality conditions required for
the OLS probability limit, while Assumption~\ref{assumption moments} ensures
that the relevant moments and integrals are well defined. These conditions
characterize the population projection coefficient but do not
nonparametrically identify \(g_h(\cdot)\). Assumption~\ref{assumption DGP}\textup{(a)} imposes a stronger sampling condition used only for the argument here; the formal asymptotic theory in the Online Appendix allows weaker conditions. By contrast, Assumption~\ref{assumption DGP}\textup{(c)} requires only continuous differentiability, while the stronger smoothness conditions in the Online Appendix imply the local condition used here.

\begin{lemma}[Estimand of LPs with linear state dependence]
\label{misspecification beta}
Suppose Assumptions~\ref{assumption DGP} and~\ref{assumption moments} hold,
and the DGP is~\eqref{mainlemma1}. Consider the LP with linear state
dependence~\eqref{linearlemma1}. Let
\[
\mu_{Z,X}
=
\frac{\mathbb E[X_t^2Z_{i,t-1}]}{\mathbb E[X_t^2]},
\qquad
D_{Z,X}
=
\mathbb E\!\left[
X_t^2(Z_{i,t-1}-\mu_{Z,X})^2
\right].
\]
Then the population coefficient in \eqref{linearlemma1} on \(Z_{i,t-1}X_t\) satisfies
\[
\beta_h
=
\frac{
\mathbb E\!\left[
X_t^2(Z_{i,t-1}-\mu_{Z,X})g_h(Z_{i,t-1})
\right]
}{
D_{Z,X}
}.
\]
Under the additional conditional homoskedasticity condition
\(\mathbb E[X_t^2\mid Z_{i,t-1}]= \mathbb E[X_t^2]\), then \(\mu_{Z,X}= \mathbb E[Z_{i,t-1}]\) and \( D_{Z,X} = \mathbb E[X_t^2] \mathbb E \left[(Z_{i,t-1}-\mathbb E[Z_{i,t-1}])^2 \right]  \), and the expression simplifies to
\[
\beta_h
=
\frac{
\mathbb E\!\left[(Z_{i,t-1}- \mathbb E[Z_{i,t-1}])g_h(Z_{i,t-1})\right]
}{
\mathbb E\!\left[(Z_{i,t-1}- \mathbb E[Z_{i,t-1}])^2\right]
}
=
\frac{\operatorname{Cov}(Z_{i,t-1},g_h(Z_{i,t-1}))}
{\operatorname{Var}(Z_{i,t-1})}.
\]

Moreover, the population coefficient admits the representation
\[
\beta_h
=
\int_{\mathcal Z}\omega_h(z)g_h'(z)\,dz,
\]
where
\[
\omega_h(z)
=
\frac{
\mathbb E\!\left[
X_t^2(Z_{i,t-1}-\mu_{Z,X})\mathbf 1\{Z_{i,t-1}\ge z\}
\right]
}{
D_{Z,X}
}
\]
with $\omega_h(z)\ge 0$
for all $z\in\mathcal Z$,
$\int_{\mathcal Z}\omega_h(z)\,dz=1$.
\end{lemma}

\begin{proof}[Proof of Lemma~\ref{misspecification beta}]
By Assumption~\ref{assumption DGP}\((a)\), we can let \(Z=Z_{i,t-1}\), \(X=X_t\), and \(g_h(Z)=g_h(Z_{i,t-1})\). Then, due to Assumption~\ref{assumption DGP}\((d)\) and the true DGP in~\eqref{mainlemma1}, the population coefficients in~\eqref{linearlemma1} solve
\[
\mathbb E[X^2g_h(Z)]
=
\alpha_h \mathbb E[X^2]+\beta_h \mathbb E[X^2Z],
\qquad
\mathbb E[X^2Zg_h(Z)]
=
\alpha_h \mathbb E[X^2Z]+\beta_h \mathbb E[X^2Z^2].
\]
Solving these normal equations gives
\[
\beta_h
=
\frac{
\mathbb E[X^2] \mathbb E[X^2Zg_h(Z)] - \mathbb E[X^2Z] \mathbb E[X^2g_h(Z)]
}{
\mathbb E[X^2] \mathbb E[X^2Z^2]-\{ \mathbb E[X^2Z]\}^2
}
=
\frac{
\mathbb E[X^2(Z-\mu_{Z,X})g_h(Z)]
}{
D_{Z,X}
},
\]
where \(\mu_{Z,X}= \mathbb E[X^2Z]/ \mathbb E[X^2]\) and \(D_{Z,X} = \mathbb E[X^2(Z - \mu_{Z,X})^2]\). Assumption~\ref{assumption moments}\((a)\) ensures that \(\mathbb E[X^2] > 0\) and \(D_{Z,X} > 0\). This proves the first claim.

By Assumption~\ref{assumption DGP}\((b,c)\), we can write
\[
g_h(Z)
=
g_h(\underline z)
+
\int_{\underline z}^{\bar z}
\mathbf 1\{Z\ge z\}g_h'(z)\,dz .
\]
Then, it holds
\begin{align*}
\mathbb E[X^2(Z-\mu_{Z,X})g_h(Z)] & =_{(1)} \mathbb E \!\left[ X^2(Z-\mu_{Z,X}) \int_{\underline z}^{\bar z} \mathbf 1\{Z\ge z\}g_h'(z)\,dz \right] \\
& =_{(2)} \int_{\underline z}^{\bar z} \mathbb E \!\left[ X^2(Z-\mu_{Z,X})\mathbf 1\{Z\ge z\} \right]
g_h'(z)\,dz .
\end{align*}
where \(=_{(1)}\) is by \(\mathbb E[X^2(Z-\mu_{Z,X})]=0\) and \(=_{(2)}\) by Fubini's theorem. Dividing by \(D_{Z,X}\) yields
\[
\beta_h
=
\int_{\mathcal Z}\omega_h(z)g_h'(z)\,dz, \quad 
\omega_h(z)
=
\frac{
\mathbb E\!\left[
X^2(Z - \mu_{Z,X})\mathbf 1\{Z \ge z\}
\right]
}{
D_{Z,X}
}.
\]

It remains to verify the properties of \(\omega_h(z)\). We first show that \(\omega_h(z)\ge 0\) for any \(z\in\mathcal Z\). Let \(A_z=\mathbf 1\{Z\ge z\}\) and \(A_z^c=\mathbf 1\{Z<z\}\). As \(A_z+A_z^c=1\), we have
\[
\mathbb E[X^2]=\mathbb E[X^2A_z]+\mathbb E[X^2A_z^c],
\qquad
\mathbb E[X^2Z]=\mathbb E[X^2ZA_z]+\mathbb E[X^2ZA_z^c].
\]
Substituting these decompositions and simplifying gives
\[
\mathbb E[X^2(Z-\mu_{Z,X})A_z]
=
\frac{\mathbb E[X^2ZA_z] \mathbb E[X^2A_z^c]
-
\mathbb E[X^2ZA_z^c] \mathbb E[X^2A_z]}{\mathbb E[X^2]}.
\]
Using
\[
\begin{aligned}
\mathbb E[X^2ZA_z]
&=
\mathbb E[X^2\{z+(Z-z)\}A_z] 
=
z\mathbb E[X^2A_z]+\mathbb E[X^2(Z-z)A_z],
\\[4pt]
\mathbb E[X^2ZA_z^c]
&=
\mathbb E[X^2\{z-(z-Z)\}A_z^c]
=
z\mathbb E[X^2A_z^c]-\mathbb E[X^2(z-Z)A_z^c],
\end{aligned}
\]
it follows that
\[
\mathbb E[X^2(Z-\mu_{Z,X})A_z]
=
\frac{\mathbb E[X^2(Z-z)A_z]\mathbb E[X^2A_z^c]
+
\mathbb E[X^2(z-Z)A_z^c]\mathbb E[X^2A_z]}
{\mathbb E[X^2]}.
\]
The numerator is nonnegative, since \(X^2\ge 0\), \(A_z,A_z^c\in\{0,1\}\),
\(Z-z\ge 0\) on \(\{Z\ge z\}\), and \(z-Z > 0\) on \(\{Z<z\}\). Hence, \(\omega_h(z) \ge 0\). 

Second, we show that the weights integrate to one. Using the definition of \(\omega_h(z)\) and Fubini's theorem, we obtain
\[
\begin{aligned}
\int_{\underline z}^{\bar z}\omega_h(z)\,dz
&=
\frac{1}{D_{Z,X}}
\mathbb E\!\left[
X^2(Z-\mu_{Z,X})
\int_{\underline z}^{\bar z}\mathbf 1\{Z\ge z\}\,dz
\right] \\
&=_{(1)}
\frac{
\mathbb E[X^2(Z-\mu_{Z,X})(Z-\underline z)]
}{
D_{Z,X}
}
=_{(2)}
\frac{
\mathbb E[X^2(Z-\mu_{Z,X})^2]
}{
D_{Z,X}
}
=
1,
\end{aligned}
\]
where \(=_{(1)}\) is by \(\int_{\underline z}^{\bar z}\mathbf 1\{Z\ge z\}\,dz=Z-\underline z\) and \(=_{(2)}\) by \(\mathbb E[X^2(Z-\mu_{Z,X})]=0\).
\end{proof}

Lemma~\ref{misspecification beta} shows that \(\beta_h\) averages local
derivatives \(g_h'(z)\), not response levels \(g_h(z)\). Hence
\(\beta_h>0\) does not imply that higher-\(z\) units have larger
causal responses. The fitted linear response \(\alpha_h+\beta_h z\)
coincides with the causal state-dependent response \(g_h(z)\) only when
\(g_h(\cdot)\) is linear.

\section*{Data Availability Statement}

Firm-level data are from Compustat and require a subscription. The monetary
policy shock series and replication code for the empirical analysis and
simulations will be made publicly available.
\section*{Conflict of interest}
The authors declare no conflict of interest.
\section*{AI use}
We used OpenAI ChatGPT (GPT-5.6 Sol) for editorial assistance in shortening and polishing
the manuscript to meet the journal’s submission-length requirements. All substantive content, analysis, and
conclusions are the authors’ own responsibility.

\bibliographystyle{apalike}
\bibliography{nonlinearLP_bib}


\end{document}


\title{Online Appendix to ``Causal State-Dependent Local Projections''}

\author{
Joel M. David, Raffaella Giacomini, Xiyu Jiao, and Weining Wang%
\thanks{David: Federal Reserve Bank of Chicago
(joel.david@chi.frb.org);
Giacomini: Department of Economics, University College London
(r.giacomini@ucl.ac.uk);
Jiao: Department of Economics, University of Gothenburg
(xiyu.jiao@economics.gu.se);
Wang: Department of Economics, University of Bristol
(weining.wang@bristol.ac.uk).}
}

\date{\today}
\maketitle

\setcounter{tocdepth}{2}

\section{Technical Details and Asymptotic Theory}
\label{app:sieve_tech}

This Online Appendix provides the formal setup and asymptotic theory for the
estimation and inference procedures in the paper, including regularity
conditions, consistency, asymptotic normality, pointwise and uniform inference,
proofs, and additional simulation results.

\subsection{Estimation and Inference}
\label{app:sieve_estimation_tech}

For each horizon $h\in\{0,1,\ldots,H\}$, recall the LP
$Y_{i,t+h}=g_h(Z_{i,t-1})X_t+W_{i,t-1}^{\top}\gamma_h+u_{i,t+h}$.
We estimate and conduct inference on the potentially nonlinear function
$g_h(\cdot)$ using a sieve approximation. Let
$\{\phi_{1,J},\ldots,\phi_{J,J}\}$ denote a sieve basis; the main text uses
cubic B-splines, though the notation also covers, for example, Fourier and
wavelet bases. Suppressing the dependence on $J$ when clear, write
$\phi(z)=(\phi_1(z),\ldots,\phi_J(z))^{\top}$ and
$g_h(z)=\phi(z)^{\top}b_h+R_h(z)$, where
$b_h=(b_{h,1},\ldots,b_{h,J})^{\top}$ and $R_h(z)$ is the approximation error.

The sieve dimension $J$ is selected in the main text by AIC over a finite
candidate set. The simulations also consider generalized cross-validation
(GCV) and LASSO-based selection to assess finite-sample sensitivity to the
selection rule.

Substituting the sieve approximation yields estimation equation
\begin{equation}
Y_{i,t+h}
=
X_t\phi(Z_{i,t-1})^{\top}b_h
+
W_{i,t-1}^{\top}\gamma_h
+
R_h(Z_{i,t-1})X_t
+
u_{i,t+h}.
\label{eq:lp_sieve_est}
\end{equation}
For notation stacked across individuals $i$, define
\begin{align*}
Y_{t+h}
& =
(Y_{1,t+h},\ldots,Y_{N,t+h})^{\top},
&
Z_{t-1}
& =
(Z_{1,t-1},\ldots,Z_{N,t-1})^{\top}, \\
W_{t-1}
& =
\begin{pmatrix}
W_{1,t-1}^{\top}\\
\vdots\\
W_{N,t-1}^{\top}
\end{pmatrix}
\in\mathbb R^{N\times Q},
&
u_{h,t+h}
& =
(u_{1,t+h},\ldots,u_{N,t+h})^{\top}.
\end{align*}
Let
\[
\Phi(Z_{t-1})
=
\begin{pmatrix}
\phi(Z_{1,t-1})^{\top}\\
\vdots\\
\phi(Z_{N,t-1})^{\top}
\end{pmatrix}
\in\mathbb R^{N\times J},
\qquad
R_h(Z_{t-1})
=
(R_h(Z_{1,t-1}),\ldots,R_h(Z_{N,t-1}))^{\top}.
\]
Then the stacked estimation equation is
\begin{equation}
Y_{t+h}
=
X_t\Phi(Z_{t-1})b_h
+
W_{t-1}\gamma_h
+
R_h(Z_{t-1})X_t
+
u_{h,t+h}.
\label{eq:lp_sieve_vector}
\end{equation}

Define
$v_{h,i,t+h}=R_h(Z_{i,t-1})X_t+u_{i,t+h}$ and
$v_{h,t+h}=R_h(Z_{t-1})X_t+u_{h,t+h}$.
Here $v_{h,i,t+h}$ is the composite sieve-regression error; the corresponding
error and residual are denoted simply by $u_{i,t+h}$ and $\widehat u_{i,t+h}$
in the finite-dimensional regression in the main text.
Let $D_{it}^\top=(X_t\phi(Z_{i,t-1})^\top,W_{i,t-1}^\top)$,
$\theta_h^\top=(b_h^\top,\gamma_h^\top)$, and $D_t=
\bigl(X_t\Phi(Z_{t-1}),\,W_{t-1}\bigr)
\in\mathbb R^{N\times(J+Q)}$.
Then \eqref{eq:lp_sieve_est}--\eqref{eq:lp_sieve_vector} become
\begin{equation}
Y_{i,t+h}=D_{it}^{\top}\theta_h+v_{h,i,t+h},
\qquad
Y_{t+h}=D_t\theta_h+v_{h,t+h}.
\label{eq:linear_form_est}
\end{equation}

Estimating \eqref{eq:linear_form_est} by OLS gives
$\widehat\theta_h
=
\left(\sum_{t=1}^{T-h}D_t^\top D_t\right)^{-1}
\left(\sum_{t=1}^{T-h}D_t^\top Y_{t+h}\right).$ Let $\widehat\theta_h^\top=(\widehat b_h^\top,\widehat\gamma_h^\top)$,
$\widehat g_h(z)=\phi(z)^\top\widehat b_h$, with residuals $\widehat v_{h,i,t+h}
=
Y_{i,t+h}-D_{it}^\top\widehat\theta_h,
\widehat v_{h,t+h}
=
Y_{t+h}-D_t\widehat\theta_h.$
When $\widehat J_h$ is data-driven, all quantities are evaluated at
$J=\widehat J_h$.

To isolate the sieve coefficients, define
\[
\widehat A_h
=
\frac{1}{N(T-h)}
\sum_{t=1}^{T-h}D_t^\top D_t
=
\begin{pmatrix}
\widehat A_{11,h} & \widehat A_{12,h}\\
\widehat A_{21,h} & \widehat A_{22,h}
\end{pmatrix},
\qquad
\widehat B_h
=
\frac{1}{N(T-h)}
\sum_{t=1}^{T-h}D_t^\top Y_{t+h}
=
\begin{pmatrix}
\widehat B_{1,h}\\
\widehat B_{2,h}
\end{pmatrix}.
\]
The corresponding population matrix is 
$A_h
=
\E\!\left[\frac1N D_t^\top D_t\right]
=
\begin{pmatrix}
A_{11,h} & A_{12,h}\\
A_{21,h} & A_{22,h}
\end{pmatrix},$
with
\[
A_{11,h}
=
\E\!\left[\frac1N X_t^2\Phi(Z_{t-1})^\top\Phi(Z_{t-1})\right],
\qquad
A_{12,h}
=
\E\!\left[\frac1N X_t\Phi(Z_{t-1})^\top W_{t-1}\right],
\]
\[
A_{22,h}
=
\E\!\left[\frac1N W_{t-1}^\top W_{t-1}\right],
\qquad
A_{21,h}=A_{12,h}^\top.
\]

As in the main text, define the sample and population Schur complements
\[
\widehat{\widetilde A}_{11,h}
=
\widehat A_{11,h}
-
\widehat A_{12,h}\widehat A_{22,h}^{-1}\widehat A_{21,h},
\qquad
\widetilde A_{11,h}
=
A_{11,h}
-
A_{12,h}A_{22,h}^{-1}A_{21,h}.
\]
Then $\widehat b_h
=
\widehat{\widetilde A}_{11,h}^{-1}
\left(
\widehat B_{1,h}
-
\widehat A_{12,h}\widehat A_{22,h}^{-1}\widehat B_{2,h}
\right)$,
equivalently the Frisch--Waugh--Lovell coefficient from partialling out
$W_{i,t-1}$. Hence
\begin{equation}
\widehat b_h-b_h
 =
\widehat{\widetilde A}_{11,h}^{-1}
\left[
\frac{1}{N(T-h)}
\sum_{t=1}^{T-h}
\left(
X_t\Phi(Z_{t-1})^{\top}
-
\widehat A_{12,h}\widehat A_{22,h}^{-1}W_{t-1}^{\top}
\right)
v_{h,t+h}
\right].
\label{eq:estimation_error_for_sieve}
\end{equation}

For feasible inference, define the partialled-out sieve regressors
\[
\widetilde P_{h,t}
=
X_t\Phi(Z_{t-1})^\top
-
\widehat A_{12,h}\widehat A_{22,h}^{-1}W_{t-1}^\top
=
\bigl(\widetilde P_{h,1,t},\ldots,\widetilde P_{h,N,t}\bigr),
\]
where $\widetilde P_{h,i,t}
=
X_t\phi(Z_{i,t-1})
-
\widehat A_{12,h}\widehat A_{22,h}^{-1}W_{i,t-1}.$
The score is $\widehat s_{h,t}
=
\frac{1}{\sqrt N}\widetilde P_{h,t}\widehat v_{h,t+h}
=
\frac{1}{\sqrt N}\sum_{i=1}^N
\widetilde P_{h,i,t}\widehat v_{h,i,t+h}.$
Recall the population
effective score
$s_{h,t}
\coloneqq
\frac{1}{\sqrt N}
X_t\Phi(Z_{t-1})^\top u_{h,t+h}$. For notational convenience, write
\(\Omega_h\equiv\Omega_{h,J,N}\).
As in the main text, $\{\widehat s_{h,t}\}_{t=1}^{T-h}$ is used to construct the HAC estimator $\widehat\Omega_h$ of the long-run covariance $\Omega_h$ of the population score 
$\{s_{h,t}\}_{t=1}^{T-h}$, defined in Subsection~\ref{app:sieve_formal}. The covariance estimator for $\widehat b_h$ is $\widehat V_h$, and pointwise and uniform inference for $g_h(\cdot)$ proceeds as in the paper.

\subsection{Asymptotic Theory} \label{app:sieve_formal}
The results below provide the formal basis for the two theorems summarized in Section~4.3 of the main paper. Theorem~1 in the main manuscript combines the coefficient convergence rate and pointwise asymptotic normality established in Theorem~\ref{thm:main_final} with the HAC consistency result in Theorem~\ref{lem:HAC_consistency}, which together justify feasible pointwise inference. Theorem~2 combines the uniform convergence and Gaussian approximation established in Theorem~\ref{thm:uniform_GA}, the oracle uniform-band result in Corollary~\ref{cor:uniform_band}, and the feasible-band result in Remark~\ref{rem:feasible_band}.

Denote by \(\|x\|_p\) the \(\ell^p\) norm of a vector \(x\), for \(p\in\{1,2,\infty\}\). Unless otherwise stated, \(\|x\|=\|x\|_2\), and \(\|A\|\) denotes the operator norm induced by the Euclidean norm (equivalently, the spectral norm) of a matrix \(A\). For a random vector \(Y\), define $\|Y\|_{L^q} \coloneqq \left\{\E\!\left[\|Y\|_2^q\right]\right\}^{1/q}$. Let \(\lambda_{\min}(A)\) and \(\lambda_{\max}(A)\) denote the smallest and largest eigenvalues of a square matrix \(A\), respectively. Throughout the appendix, the horizon \(h\) is held fixed as \(T\to\infty\), so that \((T-h)/T\to1\). Accordingly, we use the effective sample size \(T-h\) in exact finite-sample definitions, while using \(T\) in asymptotic rate statements, growth conditions, and limiting arguments.

Let $V_{it}\coloneqq(Z_{it},W_{it}^{\top})^{\top}$,
$\varepsilon_{it}\in\mathbb R^{d_\varepsilon}$, and
$\xi_t\in\mathbb R^{d_\xi}$, where $d_\varepsilon$ and $d_\xi$ are fixed.
Define $\mathcal A_t\coloneqq\sigma(\xi_s:s\le t)$ and
$\mathcal A\coloneqq\sigma(\xi_s:s\in\mathbb Z)$, and let
$\{\mathcal F_t\}$ denote the filtration maintained in Assumptions~1--3.
Finally, let
$\sigma_{X,t}^2\coloneqq\E[X_t^2\mid\mathcal F_{t-1}]$ and
$U_{it}\coloneqq(\sigma_{X,t}^2,V_{i,t-1}^{\top})^{\top}$.
Let \(\{\varepsilon_{it}^*\}\) and \(\{\xi_t^*\}\) be independent
copies of the corresponding innovation sequences, mutually independent
of each other and of all original innovations. For \(k\geq0\), let
\(U_{it}^{*(\varepsilon,t-1-k)}\) denote the coupled version of
\(U_{it}\) obtained by replacing \(\varepsilon_{i,t-1-k}\) with
\(\varepsilon_{i,t-1-k}^*\), while leaving all other innovations
unchanged. Similarly, let \(U_{it}^{*(\xi,t-1-k)}\) denote the coupled
version obtained by replacing \(\xi_{t-1-k}\) with
\(\xi_{t-1-k}^*\), while leaving all other innovations unchanged.
Define
\[
\delta_{q,U}^{\mathrm{id}}(k)
\coloneqq
\sup_{N\geq1}
\sup_{i\leq N}
\sup_{t\in\mathbb Z}
\left\|
U_{it}
-
U_{it}^{*(\varepsilon,t-1-k)}
\right\|_{L^q},
\]
\[
\delta_{q,U}^{\mathrm{agg}}(k)
\coloneqq
\sup_{N\geq1}
\sup_{i\leq N}
\sup_{t\in\mathbb Z}
\left\|
U_{it}
-
U_{it}^{*(\xi,t-1-k)}
\right\|_{L^q},
\]
and
$\delta_{q,U}(k)
\coloneqq
\delta_{q,U}^{\mathrm{id}}(k)
+
\delta_{q,U}^{\mathrm{agg}}(k)$.
The associated dependence-adjusted norm is defined by
\[
\|U\|_{q,\alpha}
\coloneqq
\sup_{m\geq0}
(m+1)^\alpha\Delta_{m,q,U},
\qquad
\Delta_{m,q,U}
\coloneqq
\sum_{k=m}^{\infty}\delta_{q,U}(k).
\]
$
\delta_{q,X}(k)
\coloneqq
\sup_{t\in\mathbb Z}
\left|
X_t-X_t^{*(\xi,t-k)}
\right|_{L^q}.
$ and $\|X\|_{q,\alpha}\leq C_X^{\mathrm{dep}}$ is similarly defined.
For \(\kappa>1/2\) and \(M<\infty\), define
$\mathcal G_M
\coloneqq
\left\{
g\in\mathcal H^\kappa(\mathcal Z):
\|g\|_{\mathcal H^\kappa}\leq M
\right\}$.
Writing \(r=\lceil\kappa\rceil-1\), let
\[
\|g\|_{\mathcal H^\kappa}
\coloneqq
\sum_{s=0}^{r}
\sup_{z\in\mathcal Z}|D^s g(z)|
+
\sup_{z\neq z'}
\frac{|D^r g(z)-D^r g(z')|}
     {|z-z'|^{\kappa-r}},
\]
where \(D^s g\) denotes the \(s\)th derivative of \(g\) with \(D^0g=g\). For the sieve basis, define
\[
\mathbf S_t
\coloneqq
\frac1N\sum_{i=1}^N
\phi(Z_{i,t-1})\phi(Z_{i,t-1})^\top,
\qquad
Q_{J,N}
\coloneqq
\E[\mathbf S_t]
=
\E\!\left[
\frac1N\sum_{i=1}^N
\phi(Z_{i,t-1})\phi(Z_{i,t-1})^\top
\right].
\]
Suppose that the following autocovariance series converges in operator norm.
Recall that the effective score, its lag-$k$ autocovariance matrix, and its
finite-$(J,N)$ long-run covariance are defined by
\[
s_{h,t}
\coloneqq
\frac1{\sqrt N}
X_t\Phi(Z_{t-1})^\top u_{h,t+h},
\qquad
\Gamma_{h,J,N}(k)
\coloneqq
\E[s_{h,t}s_{h,t-k}^{\top}],
\qquad
\Omega_{h,J,N}
\coloneqq
\sum_{k\in\mathbb Z}\Gamma_{h,J,N}(k).
\]

 Thus, \(\Omega_h\) denotes the
deterministic \(J\times J\) population long-run covariance at the current
\((J,N)\), with its dependence on \((J,N)\) suppressed. It need not
converge to a fixed bounded matrix. Recall from the main text that $\chi_{h,J,N}
\coloneqq
1\vee\lambda_{\max}(\Omega_h)$.  For later use, define
$\ell_J\coloneqq1\vee\log J,
\bar r_{A,T}
\coloneqq
\sqrt{\frac{\ell_J}{T}}
+
\sqrt{\frac{J\ell_J}{NT}}
+
\frac{J\ell_J}{NT},
\bar a_{A,T}
\coloneqq
\bar r_{A,T}+\frac JT$.

\begin{assumption}
\label{ass:C}
Maintain Assumptions~1--3 in the paper. Suppose:

\begin{enumerate}
\item \textbf{(Sampling and temporal dependence)} The array $\{\varepsilon_{it}:i\ge1,t\in\mathbb Z\}$ is i.i.d.\ across
$i,t$, $\{\xi_t:t\in\mathbb Z\}$ is i.i.d.\ over $t$, and the two
innovation sequences are mutually independent. Moreover,
$V_{it}$ is $\mathcal F_t$-measurable,
$X_t=G_X(\ldots,\xi_{t-1},\xi_t)$, and
\[
\sigma_{X,t}^2
=
\E[X_t^2\mid\mathcal F_{t-1}]
=
\E[X_t^2\mid\mathcal A_{t-1}]
=
G_\sigma(\ldots,\xi_{t-2},\xi_{t-1}).
\]
For each $i$,
$V_{it}=G_{V,i}(\ldots,\varepsilon_{i,t-1},\varepsilon_{it};
\ldots,\xi_{t-1},\xi_t)$.
For every $N\ge1$,
$\{(X_t,V_{1t},\ldots,V_{Nt})\}_{t\in\mathbb Z}$ is strictly stationary.
Moreover, $U_{it}$ is $\mathcal F_{t-1}$-measurable, with first component
depending only on the aggregate innovations.
For some $q>4$, $\rho\in(0,1)$, and finite constants
$C_U$ and $C_{\mathrm{dep}}$,
$
\sup_{N\geq1}\sup_{i\leq N}\sup_{t\in\mathbb Z}
|U_{it}|_{L^q}
\leq C_U,
\delta_{q,U}(k)+\delta_{q,X}(k)
\leq
C_{\mathrm{dep}}\rho^k,
 k\geq0.
$
Each $Z_{it}$ has a continuous distribution on the common compact interval
$\mathcal Z$, and
$0<\underline{\sigma}_X^2\le\sigma_{X,t}^2
\le\overline{\sigma}_X^2<\infty$ a.s.
For the same $q$,
$\E[|X_t|^q\mid\mathcal F_{t-1}]\le C_X$ a.s.
\item \textbf{(Error process and orthogonality)}
For every \(i\) and \(h\leq H\),
$
\E\!\left[
X_t\phi(Z_{i,t-1})u_{i,t+h}
\right]=0,
\E\!\left[
W_{i,t-1}u_{i,t+h}
\right]=0$.
Conditional on the full aggregate-innovation path \(\mathcal A\), the
individual processes
$\left\{
\bigl(V_{it},(u_{i,t+h})_{h=0}^{H}\bigr)
\right\}_{t\in\mathbb Z}$
are independent across \(i\).
For each \(i\) and \(h\leq H\),
$
u_{i,t+h}
=
\sum_{\ell=0}^{\infty}
a_{i,h,\ell}\eta_{i,t+h-\ell}.
$ For some $\rho_u\in(0,1)$ and $C_u<\infty$,
$
\sup_{i,h}
\sum_{\ell=m}^{\infty}|a_{i,h,\ell}|
\leq
C_u\rho_u^m,
 m\geq0$. Conditional on $\mathcal A$, the joint innovation array
$(\eta_{it},\varepsilon_{it}):i\geq1,\ t\in\mathbb Z\}$
is independent across $(i,t)$, with
$\E[\eta_{it}\mid\mathcal A]=0,
\sup_{i,t}\|\eta_{it}\|_{L^q}<\infty$.
Moreover,
$\lambda_{\min}(\Omega_h)\geq c_\Omega>0,
1\leq\chi_{h,J,N}\leq C_\chi N,$
uniformly over \(h\leq H\) and along the asymptotic sequence. Let
$q_{h,t}^{W}
\coloneqq
N^{-1/2}W_{t-1}^{\top}u_{h,t+h}$.
For every $h\leq H$ and each $(J,N)$, suppose that
$\{s_{h,t}\}_{t\in\mathbb Z}$ and
$\{q_{h,t}^{W}\}_{t\in\mathbb Z}$ are strictly stationary and 
$
\sup_{h\leq H}
\sup_{J,N}
\frac{1}{\chi_{h,J,N}}
\sum_{k\in\mathbb Z}
\left\|
\E\!\left[
q_{h,t}^{W}q_{h,t-k}^{W\top}
\right]
\right\|
<\infty.
$ Suppose also that
$
\sup_{h\leq H}
\sup_{J,N}
\frac{1}{\chi_{h,J,N}}
\sum_{k\in\mathbb Z}
\|\Gamma_{h,J,N}(k)\|
<\infty.
$
Moreover, for some $q_0\in(2,q/2]$,
$
\sup_{h\leq H}
\sup_{J,N}
\sup_{\|a\|_2=1}
\left\|
a^\top\Omega_h^{-1/2}s_{h,0}
\right\|_{L^{q_0}}
<\infty.
$ For some $\rho_s\in(0,1)$ and $C_s<\infty$, the effective score
satisfies
$\sup_{h\leq H}
\sup_{J,N}
\max_{j\leq J}
\sup_{r\in\mathbb Z}
\frac{1}{\sqrt{\chi_{h,J,N}}}
\left\|
s_{h,r-h,j}
-
s_{h,r-h,j}^{*(r-m)}
\right\|_{L^{q_0}}
\leq
C_sL_J\rho_s^{(m-h)_+}, m\geq0$.

\item \textbf{(Smoothness)} For some \(\kappa>1/2\) and \(M<\infty\),
$g_h\in\mathcal G_M$
uniformly over \(h\leq H\).

\item \textbf{(Sieve basis and matrix regularity)}
The sieve basis satisfies
$\sup_{z\in\mathcal Z}\|\phi(z)\|_2\le C_\phi\sqrt J$ and
$\|\phi(z)-\phi(z')\|_2\le L_J|z-z'|$ for all $z,z'\in\mathcal Z$.
There exist $0<c_Q<C_Q<\infty$ such that
$c_Q\le\lambda_{\min}(Q_{J,N})
\le\lambda_{\max}(Q_{J,N})\le C_Q$
uniformly in $J,N$.
The eigenvalues of $A_{22,h}$ are uniformly bounded away from zero and infinity
over $h\le H$, and
$\sup_{i,t,J}\|
\E[\phi(Z_{i,t-1})\phi(Z_{i,t-1})^\top\mid\mathcal A_{t-1}]
\|\le C_Q^{\mathcal A}<\infty$ a.s. In addition, the growing-dimensional sample Gram matrix satisfies
$
\sup_{h\leq H}
\left\|
\frac{1}{T-h}
\sum_{t=1}^{T-h}
\left\{
X_t^2\mathbf S_t
-
\E[X_t^2\mathbf S_t]
\right\}
\right\|
=
O_p\!\left(
\sqrt{\frac{\ell_J}{T}}
+
\sqrt{\frac{J\ell_J}{NT}}
+
\frac{J\ell_J}{NT}
\right).
$

\item \textbf{(Sieve growth)}
As $N,T\to\infty$, $J\to\infty$,
$J^2/(NT)\to0$, $\frac{J\chi_{h,J,N}}{NT}\to0$, $J/N\to0$ and $J/T\to0$.

\end{enumerate}
\end{assumption}

Assumptions~1--3 in the paper provide the structural identification conditions:
shock exogeneity, predetermined states and controls, and a state-dependent linear
conditional mean identifying the causal response $g_h(\cdot)$.
Assumptions~\ref{ass:C}(1)--(2) allow cross-sectional dependence through
aggregate innovations: conditional on the path \(\mathcal A\), the joint
covariate and error processes
$\left\{
\bigl(V_{it},(u_{i,t+h})_{h=0}^{H}\bigr)
\right\}_{t\in\mathbb Z}$
are independent across \(i\).
Because the aggregate innovation may be multidimensional, the maintained
representation allows multiple observed or latent common components. Recall that
$\chi_{h,J,N}
=
1\vee\lambda_{\max}(\Omega_h)$
measures the strength of cross-sectional dependence in the effective score.
The formal theory allows
\(1\leq\chi_{h,J,N}\leq C_\chi N\), so the worst-direction effective sample
size is of order \(NT/\chi_{h,J,N}\), and the stochastic component of the
sieve coefficient rate is
$O_p\!\left(
\sqrt{\frac{J\chi_{h,J,N}}{NT}}
\right)$.
Bounded \(\chi_{h,J,N}\) gives the root-\(NT\) regime, whereas
\(\chi_{h,J,N}\asymp N\) gives the root-\(T\) regime emphasized by
\citet{AlmuzaraSancibrian2025MicroResponses}; intermediate values produce
rates between these cases. Because strong common components may affect only
a subset of directions, this rate is a worst-direction upper bound.
Pointwise inference adapts to direction-specific rates through the full
long-run covariance matrix \(\Omega_h\). The HAC estimator estimates
\(\Omega_h\) at its actual scale, so the implementation is unchanged;
feasible inference requires directional relative consistency of the
estimated long-run variance. Temporal dependence in $\{U_{it}\}_{t\in\mathbb Z}$ is controlled by a functional-dependence condition, providing the short-memory
restrictions needed for laws of large numbers, root-$T$ bounds, and
concentration of sample moment matrices. Conditional heteroskedasticity is permitted because \(\sigma_{X,t}^{2}=\E[X_t^2\mid\mathcal F_{t-1}]\) need not be constant. 

Assumption~\ref{ass:C}(2) imposes exogeneity, weak dependence, finite moments,
and long-run variance nondegeneracy for the structural error. The moving-average
representation accommodates serial correlation from overlapping horizons, while
weighted summability, together with Assumption~\ref{ass:C}(1), ensures
short-range dependence of the effective score. 
The stronger score-level dependence conditions required for
uniform Gaussian approximation are imposed separately in
Theorem~\ref{thm:uniform_GA}. The moment and eigenvalue
conditions support the central limit theorem, consistent long-run variance estimation, and a
nondegenerate Gaussian limit for pointwise inference.

Assumption~\ref{ass:C}(3) imposes smoothness of the impulse response
function \(g_h(\cdot)\). For standard spline or wavelet bases, classical
approximation theory
\citep[e.g.,][]{Schumaker1981,DeVoreLorentz1993} gives
$\sup_{z\in\mathcal Z}|R_h(z)|
=
O(J^{-\kappa})$,
uniformly over \(h\leq H\). The sieve approximation induces a coefficient bias of order
\(J^{1/2-\kappa}\). Together with the growth conditions
$
\frac{J^2}{NT}\to0,
\frac JT\to0,
\frac{J\chi_{h,J,N}}{NT}\to0$, $J/N\to0$
in Assumption~\ref{ass:C}(5), these bounds yield
$\|\widehat b_h-b_h\|_2
=
O_p(J^{1/2-\kappa})
+
O_p\!\left(\sqrt{\frac{J\chi_{h,J,N}}{NT}}\right).$
Formally balancing the two components gives
$J^{-\kappa}
\asymp
\sqrt{\frac{\chi_{h,J,N}}{NT}}$,
provided this choice also satisfies the growth conditions in
Assumption~\ref{ass:C}(5). Thus, bounded \(\chi_{h,J,N}\) gives
\(J\asymp(NT)^{1/(2\kappa)}\), whereas
\(\chi_{h,J,N}\asymp N\) gives \(J\asymp T^{1/(2\kappa)}\).

Assumption~\ref{ass:C}(4) imposes regularity conditions on the sieve
basis and the population moment matrices. The envelope and Lipschitz
conditions are standard requirements for sieve approximation and
uniform inference, while the eigenvalue bounds on
$Q_{J,N}=\E[\mathbf S_t]$
ensure that the basis remains uniformly well conditioned as \(J\)
grows. Moreover,
$A_{11,h}
=
\E[X_t^2\mathbf S_t]
=
\E[\sigma_{X,t}^2\mathbf S_t].$
Therefore, the conditional-variance bounds in
Assumption~\ref{ass:C}(1) imply
$\underline{\sigma}_X^2Q_{J,N}
\le
A_{11,h}
\le
\overline{\sigma}_X^2Q_{J,N},$ 
so \(A_{11,h}\) is uniformly nonsingular. The conditional Gram bound
controls fluctuations induced by the aggregate history. Finally, the
nonsingularity of
$A_{22,h}
=
\E\!\left[\frac1N W_{t-1}^{\top}W_{t-1}\right]$
ensures that the partialling-out step is well defined.
This high-level concentration condition can be verified from primitive
functional-dependence conditions together with the local-support,
bounded-overlap, and envelope properties of the spline basis. Because
the verification requires a lengthy concentration argument for the
growing-dimensional transformed Gram process but provides no additional
insight for the main results, we impose the condition directly here.

\begin{theorem}[Rate of convergence and pointwise asymptotic normality]
\label{thm:main_final}
Under Assumption~\ref{ass:C}, the sieve estimator
$\widehat g_h(z)=\phi(z)^\top\widehat b_h$ satisfies the coefficient rate
$\|\widehat b_h-b_h\|_2
=
O_p(J^{1/2-\kappa})
+
O_p\!\left(\sqrt{J\chi_{h,J,N}/(NT)}\right)$.
If, in addition,
$\frac{
\sqrt{N(T-h)}\,J^{1-\kappa}
}{
\sigma_{h}^\star(z)
}
\to0$ and $\bar a_{A,T}
\frac{
\|\phi(z)\|_2\sqrt{J\chi_{h,J,N}}
}{
\sigma_{h}^\star(z)
}
\to0$, then, for each fixed $z\in\mathcal Z$,
\[
\frac{\sqrt{N(T-h)}\{\widehat g_h(z)-g_h(z)\}}
{\sigma_h^\star(z)}
\Rightarrow\mathcal N(0,1),
\]
where
$\sigma_h^{\star 2}(z)
=
\phi(z)^\top A_{11,h}^{-1}\Omega_h
A_{11,h}^{-1}\phi(z)$.
\end{theorem}

\begin{remark}[Multivariate states and the curse of dimensionality]
\label{rem:dimension}
Theorem~\ref{thm:main_final} is stated for a scalar state. With a
$d$-dimensional state $Z_{i,t-1}\in\mathcal Z\subset\mathbb R^d$, consider a
tensor-product sieve with $m$ basis functions per coordinate, so $J=m^d$.
For $g_h\in\mathcal H^\kappa(\mathcal Z)$,
$\sup_{z\in\mathcal Z}|R_h(z)|=O(m^{-\kappa})$ and
$\|\widehat b_h-b_h\|_2
=
O_p(m^{d/2-\kappa})+O_p(\sqrt{m^d\chi_{h,J,N}/(NT)})$,
requiring $\kappa>d/2$ for the bias term to vanish.
\end{remark}
\begin{remark}[Refinement for localized sieve bases]
\label{rem:localized_refinement}
Theorem~\ref{thm:main_final} uses global Euclidean-norm bounds and thus requires
$\frac{\sqrt{NT}\,J^{1-\kappa}}
{\sigma_h^\star(z)}
\to0.$
For an $L^2$-normalized fixed-order B-spline basis with quasi-uniform knots,
$\sup_z\#\{j:\phi_j(z)\neq0\}\lesssim1$,
$\max_{j\leq J}|\phi_j|_\infty\lesssim\sqrt J$, and
$\|\phi(z)\|_2\asymp\sqrt J$; its population Gram matrix is uniformly banded
and well conditioned. Define the approximation-score vector
$\widetilde B_{h,T}
\coloneqq
\frac{1}{N(T-h)}
\sum_{t=1}^{T-h}
\widetilde P_{h,t}X_tR_{h,t}$.
If
$\|\widehat{\widetilde A}_{11,h}^{-1}\|_\infty=O_p(1)$
and
$\|\widetilde B_{h,T}\|_\infty=O_p(J^{-\kappa})$,
where, for matrices, $\|\cdot\|_\infty$ denotes the maximum absolute row-sum norm, localization sharpens the pointwise sieve-approximation
contribution from $O_p(J^{1-\kappa})$ to
$O_p(J^{1/2-\kappa})$.
At states for which
$\sigma_h^\star(z)\asymp\sqrt J$, the standardized bias is
$O_p\!\left(
\frac{\sqrt{NT}\,J^{1/2-\kappa}}
{\sigma_h^\star(z)}
\right)$, so the pointwise undersmoothing condition can be weakened to
$\sqrt{NT}\,J^{-\kappa}\to0$ without changing the coefficient rate.
In the bounded-dependence regime $\chi_{h,J,N}=O(1)$, if
$N\asymp T$ and $J\asymp(NT)^a$, this yields
$1/(2\kappa)<a<1/3$ up to logarithmic factors, which is nonempty for
$\kappa>3/2$.
This refinement applies only to the pointwise result; the uniform result
continues to use the global Euclidean-norm bounds and therefore requires
stronger growth and smoothness conditions.
\end{remark}

For clarity, starred quantities denote asymptotic normalizations. Thus
$\sigma_h^2(z)=\sigma_h^{\star 2}(z)/[N(T-h)]$ and
$\widehat\sigma_h^2(z)=\widehat\sigma_h^{\star 2}(z)/[N(T-h)]$,
while
$V_h=V_h^\star/[N(T-h)]$ and
$\widehat V_h=\widehat V_h^\star/[N(T-h)]$.
Here $V_h^\star$ is the asymptotic covariance of
$\sqrt{N(T-h)}(\widehat b_h-b_h)$.
The studentized uniform-band critical value
$c_{h,1-\alpha}^{\mathrm{stu}}$ in
Corollary~\ref{cor:uniform_band} below is dimensionless and therefore
does not require rescaling.

\begin{theorem}[Consistency of the HAC variance estimator]
\label{lem:HAC_consistency}
Suppose Assumption~\ref{ass:C} holds. Assume $q_0> 4$.
Suppose also that
$\lim_{K\to\infty}
\sup_{J,N}
\frac1{\chi_{h,J,N}}
\sum_{|k|>K}
\|\Gamma_{h,J,N}(k)\|
=0,$
and $\sqrt{\frac{NT}{\chi_{h,J,N}}}\,J^{-\kappa}\to0.$ 
Recall the score process
$\widehat s_{h,t}
=\frac{1}{\sqrt N}\bigl(
X_t\Phi(Z_{t-1})^\top
-\widehat A_{12,h}\widehat A_{22,h}^{-1}W_{t-1}^\top
\bigr)\widehat v_{h,t+h}$,
and the HAC estimator
\[
\widehat \Omega_h
=
\frac{1}{T-h}\sum_{t=1}^{T-h}
\widehat s_{h,t}\widehat s_{h,t}^{\top}
+
\sum_{k=1}^{L} w_k
\left(
\frac{1}{T-h}\sum_{t=k+1}^{T-h}
\bigl(
\widehat s_{h,t}\widehat s_{h,t-k}^{\top}
+
\widehat s_{h,t-k}\widehat s_{h,t}^{\top}
\bigr)
\right),
\]
with Bartlett weights $w_k=1-k/(L+1)$, $k=1,\ldots,L$, where
$d_J\coloneqq1+\log(1+L_J)$ and $L=L_T$ satisfies
$L\to\infty,
J^2d_J^4\frac{L}{T}\to0,
\frac{LJ}{\sqrt T}\to0$.
Then, as $N,T\to\infty$,
$\frac{
\|\widehat\Omega_h-\Omega_h\|
}{
\chi_{h,J,N}
}
\xrightarrow{p}0$ and
$\frac{
\|\widehat V_h^\star-V_h^\star\|
}{
\chi_{h,J,N}
}
\xrightarrow{p}0$, where
$\widehat V_h^\star
=\widehat{\widetilde A}_{11,h}^{-1}\widehat\Omega_h
\widehat{\widetilde A}_{11,h}^{-1}$
and
$V_h^\star=A_{11,h}^{-1}\Omega_h A_{11,h}^{-1}$.
In particular, for any fixed $z\in\mathcal Z$ satisfying
$\phi(z)^\top V_h^\star\phi(z)
\gtrsim
\chi_{h,J,N}\|\phi(z)\|_2^2$,
$\left|
\frac{\phi(z)^\top\widehat V_h^\star\phi(z)}
{\phi(z)^\top V_h^\star\phi(z)}
-1
\right|=o_p(1)$.
\end{theorem}

Recall that
$\widetilde\phi_h(z)=A_{11,h}^{-1}\phi(z)$
and
$
\sigma_h^{\star 2}(z)
=
\widetilde\phi_h(z)^\top
\Omega_h
\widetilde\phi_h(z).
$
Define
$\tau_J\coloneqq1\vee\sup_{z\in\mathcal Z}
\|\widetilde\phi_h(z)\|_1$
and
$
\psi_h(z)
\coloneqq
\frac{\widetilde\phi_h(z)}
{\sigma_h^\star(z)}
$
for use below.
Define
$L_{\psi,J,N}
\coloneqq
\sup_{z\neq z'}
\frac{
\|\psi_h(z)-\psi_h(z')\|_2
}{
|z-z'|
}$.
Assumption~\ref{ass:C}(4), condition~\eqref{eq:U_normalization},
and the eigenvalue bounds on $\Omega_h$ imply
$L_{\psi,J,N}
\lesssim
\frac{
\sqrt{\chi_{h,J,N}}\,L_J
}{
\tau_J
}$.

\begin{theorem}[Uniform convergence and studentized Gaussian approximation]
\label{thm:uniform_GA}
Suppose Assumption~\ref{ass:C} holds and that
\begin{equation}
\inf_{z\in\mathcal Z}
\frac{\|\widetilde\phi_h(z)\|_2}{\tau_J}
\geq c>0.
\tag{U-Norm}\label{eq:U_normalization}
\end{equation}

Let
$
y_{h,r}(z)
\coloneqq
\psi_h(z)^\top s_{h,r-h}
$
denote the studentized score reindexed by outcome date. For $q_s=q_0>4$ and some
$\alpha_s>\frac12-\frac1{q_s}$, let
$y_{h,i,\{0\}}(z)$ denote the coupled version obtained by replacing
the primitive innovations at date zero with independent copies. Suppose
\begin{equation}
\sup_{J,N}
\sup_{m\geq0}
(m+1)^{\alpha_s}
\sum_{i=m}^{\infty}
\left\|
\sup_{z\in\mathcal Z}
\left|
y_{h,i}(z)-y_{h,i,\{0\}}(z)
\right|
\right\|_{L^{q_s}}
<\infty.
\tag{U-FD}\label{eq:U_score_dependence}
\end{equation}

Let $(\varepsilon_T\downarrow0)$, and let
$\mathcal Z_T=\{z_1,\ldots,z_{K_T}\}$
be a deterministic $(\varepsilon_T)$-net of $\mathcal Z$, with
$K_T\lesssim\varepsilon_T^{-1}$ and
$\ell_T\coloneqq1\vee\log(2K_T)$.
Choose the block sequences in Theorem~\ref{strong} so that its
block-size and moment-exponent conditions hold with
$n=T-h$, $p=K_T$, $q=q_s$, and $\alpha=\alpha_s$.

Let $(\delta_T)$ denote the corresponding
Gaussian-approximation rate in~\eqref{eq:delta-rate}.
Recall that $\Omega_h\equiv\Omega_{h,J,N}$ and define
$
G_h(z)
\coloneqq
\psi_h(z)^\top
\Omega_h^{1/2}\mathcal N_J,
$
where $\mathcal N_J\sim\mathcal N(0,I_J)$.

If
\begin{equation}
\sqrt{\ell_T}
\left[
\delta_T
+
\bar a_{A,T}\sqrt{J\chi_{h,J,N}}
+
\sqrt{NT}\,J^{1/2-\kappa}
+
L_{\psi,J,N}\varepsilon_T
\sqrt{J\chi_{h,J,N}}
+
\sqrt{\chi_{h,J,N}}\,
L_{\psi,J,N}\varepsilon_T
\sqrt{\log(1/\varepsilon_T)}
\right]
\to0,
\tag{U-Growth}\label{eq:U_growth}
\end{equation}
where $(\bar a_{A,T})$ is defined in
Lemma~\ref{lem:gram_rates}, then the following conclusions hold.

\smallskip
\noindent\textbf{(a) Sup-norm rate.}
$\sup_{z\in\mathcal Z}
|\widehat g_h(z)-g_h(z)|
=
O_p(J^{1-\kappa})
+
O_p\!\left(
J\sqrt{\frac{\chi_{h,J,N}}{NT}}
\right)
.$

\smallskip
\noindent\textbf{(b) Studentized Gaussian approximation.}
$$
\sup_{x\in\mathbb R}
\left|
\P\!\left(
\sup_{z\in\mathcal Z}
\left|
\frac{
\sqrt{N(T-h)}
\{\widehat g_h(z)-g_h(z)\}
}{
\sigma_h^\star(z)
}
\right|
\leq x
\right)
-
\P\!\left(
\sup_{z\in\mathcal Z}|G_h(z)|
\leq x
\right)
\right|
\to0.
$$
\end{theorem}

\smallskip
\noindent
The following oracle confidence-band result is immediate.
\begin{corollary}[Oracle studentized uniform confidence band]
\label{cor:uniform_band}
Under the hypotheses of Theorem~\ref{thm:uniform_GA}, let
$c_{h,1-\alpha}^{\mathrm{stu}}$
be defined by
$$
c_{h,1-\alpha}^{\mathrm{stu}}
\coloneqq
\inf\left\{
c>0:
\P\!\left(
\sup_{z\in\mathcal Z}|G_h(z)|
\leq c
\right)
\geq1-\alpha
\right\}.
$$
Then
$$
\P\!\left(
\sup_{z\in\mathcal Z}
\left|
\frac{
\sqrt{N(T-h)}
\{\widehat g_h(z)-g_h(z)\}
}{
\sigma_h^{\star}(z)
}
\right|
\leq
c_{h,1-\alpha}^{\mathrm{stu}}
\right)
\to1-\alpha.
$$
Equivalently,
$$
\widehat g_h(z)
\pm
\frac{
c_{h,1-\alpha}^{\mathrm{stu}}
\sigma_h^\star(z)
}{
\sqrt{N(T-h)}
},
\qquad
z\in\mathcal Z,
$$
has asymptotic simultaneous coverage $1-\alpha$.
This band is oracle because $A_{11,h}$ and $\Omega_h$ are unknown.
\end{corollary}
\begin{remark}[Feasible implementation]
\label{rem:feasible_band}
The oracle critical value in
Corollary~\ref{cor:uniform_band} may be replaced by the conditional
Gaussian critical value computed from $\widehat V_h^\star$, provided
\begin{equation}
\ell_T^2
\sup_{z,z'\in\mathcal Z}
\frac{
\left|
\phi(z)^\top
\bigl(\widehat V_h^\star-V_h^\star\bigr)
\phi(z')
\right|
}{
\sigma_h^\star(z)\sigma_h^\star(z')
}
\to_p0.
\tag{U-Feasible}\label{eq:U_feasible}
\end{equation}
Under condition~\eqref{eq:U_feasible}, the feasible studentized
confidence band has asymptotic simultaneous coverage $1-\alpha$.
Condition~\eqref{eq:U_feasible} is a direction-uniform strengthening
of the operator-norm HAC consistency result in
Theorem~\ref{lem:HAC_consistency}.
\end{remark}

\begin{remark}[Admissible oracle rates]
\label{rem:uniform_GA_rates}
For an $L^2$-normalized fixed-order B-spline basis with quasi-uniform knots,
$\sup_{z\in\mathcal Z}\|\phi(z)\|_2\asymp\sqrt J$,
$L_J\lesssim J^{3/2}$, and $\tau_J\asymp\sqrt J$.
These orders are consistent with the uniform eigenvalue bounds on
$Q_{J,N}$ under the maintained $L^2$ normalization.
In the bounded-dependence regime $\chi_{h,J,N}=O(1)$,
a sufficient corresponding order for the standardized direction is
$L_{\psi,J,N}\lesssim J$.
\end{remark}

The Gaussian-approximation rate $\delta_T$ in Equation
\eqref{eq:delta-rate} is the maximum of all five terms stated in
Theorem~\ref{strong}. If $L\asymp T^{2/3}$, reducing this rate to
the single leading term $T^{-1/9}\ell_T^{7/6}$ requires the
additional compatibility conditions discussed in
Remark~\ref{rem:rates}. Without those conditions, the full rate
$\delta_T$ must be retained.

For a concrete sufficient regime, suppose $N\asymp T$,
$\chi_{h,J,N}=O(1)$, and use an
$L^2$-normalized B-spline basis with
$\tau_J\asymp\sqrt J$, $L_J\lesssim J^{3/2}$,
$L_{\psi,J,N}\lesssim J$, and
$\varepsilon_T=J^{-3}$. Choose the block sequences in
Theorem~\ref{strong} as
$L\asymp T^{2/3}$, $M\asymp T^{1/3}$, and
$m\asymp T^{1/(3\alpha_s)}$.
Under \eqref{eq:U_normalization} and
\eqref{eq:U_score_dependence}, and provided the remaining
block-size and covariance-approximation conditions of
Theorem~\ref{strong} hold, if
\[
J\asymp(NT)^a,\qquad
\frac1{2\kappa-1}<a<\frac13,\qquad
\alpha_s>1,\qquad
q_s>18a+2,
\]
then \eqref{eq:U_growth} holds, with logarithmic factors absorbed by the
strict inequalities. The upper bound $a<1/3$ comes from
$J^{3/2}/T\to0$, while the lower bound
$a>1/(2\kappa-1)$ comes from
$\sqrt{NT}\,J^{1/2-\kappa}\sqrt{\ell_T}\to0$.
Hence the oracle range is nonempty if and only if $\kappa>2$.

\section{Proofs}
\subsection{Proof of Theorem \ref{thm:main_final}}
\begin{lemma}[Rates for the Gram matrix and Schur complement]
\label{lem:gram_rates}

Suppose Assumption~\ref{ass:C} holds.
Recall that
$\ell_J\coloneqq 1\vee\log J$, $\bar r_{A,T}
\coloneqq
\sqrt{\frac{\ell_J}{T}}
+
\sqrt{\frac{J\ell_J}{NT}}
+
\frac{J\ell_J}{NT}$, and
$\bar a_{A,T}
\coloneqq
\bar r_{A,T}
+
\frac{J}{T}$. Under Assumption~\ref{ass:C}(5),
$\bar r_{A,T}=o(1)$ and $\bar a_{A,T}=o(1)$. Then:

\smallskip
\noindent\textbf{(a) Gram matrix.}
$\left\|
\widehat A_{11,h}-A_{11,h}
\right\|
=
O_p(\bar r_{A,T}).$

\smallskip
\noindent\textbf{(b) Schur complement.}
Recall $\widetilde A_{11,h}
=
A_{11,h}
-
A_{12,h}A_{22,h}^{-1}A_{21,h}$. By Assumption~1(ii) and the predetermination of the states and controls,
$A_{12,h}=A_{21,h}^\top=0$, and hence
$\widetilde A_{11,h}=A_{11,h}$. Moreover,
\[
\left\|
\widehat{\widetilde A}_{11,h}
-
\widetilde A_{11,h}
\right\|
=
O_p(\bar a_{A,T}).
\]

\smallskip
\noindent\textbf{(c) Inverse bounds.}
$
\|\widehat A_{11,h}^{-1}\|
+
\|\widehat{\widetilde A}_{11,h}^{-1}\|
+
\|\widehat A_{22,h}^{-1}\|
=
O_p(1).
$

\smallskip
\noindent\textbf{(d) Inverse expansions.}
$\left\|
\widehat A_{11,h}^{-1}
-
A_{11,h}^{-1}
\right\|
=
O_p(\bar r_{A,T})$
and
$\left\|
\widehat{\widetilde A}_{11,h}^{-1}
-
\widetilde A_{11,h}^{-1}
\right\|
=
O_p(\bar a_{A,T}).$
\end{lemma}

\begin{proof}[Proof of Lemma~\ref{lem:gram_rates}]
Recall $\mathbf S_t
=
\frac1N\sum_{i=1}^N
\phi(Z_{i,t-1})\phi(Z_{i,t-1})^\top$.
Assumption~\ref{ass:C}(1) shows that $0<\underline\sigma_X^2
\leq
\sigma_{X,t}^2
\leq
\overline\sigma_X^2
<\infty$ almost surely, where recall $\sigma_{X,t}^2
=
\E[X_t^2\mid\mathcal F_{t-1}]$. Since $h$ is fixed, $T-h\asymp T$ throughout.

\smallskip
\noindent\textbf{Step 0: Preliminary bounds.}
By the shock exogeneity in Assumption~1,
$\E[X_t\mid\mathcal F_{t-1}]=0$.
Since
$\phi(Z_{i,t-1})W_{i,t-1}^\top$ is
$\mathcal F_{t-1}$-measurable,
$A_{12,h}
=
\E\!\left[
X_t\phi(Z_{i,t-1})W_{i,t-1}^\top
\right]
=
0$. Consequently,
$A_{21,h}=0,
\widetilde A_{11,h}=A_{11,h}$. In particular,
$\lambda_{\min}(\widetilde A_{11,h})
=
\lambda_{\min}(A_{11,h})
\geq c_A$.

\emph{(i) Rate for $\widehat A_{12,h}$.}
Write
$\widehat A_{12,h}
=(T-h)^{-1}\sum_{t=1}^{T-h}X_tM_t$,
where
$M_t\coloneqq
N^{-1}\sum_{i=1}^N\phi(Z_{i,t-1})W_{i,t-1}^\top$.
Because $M_t$ is $\mathcal F_{t-1}$-measurable,
$\{X_tM_t\}_t$ is a matrix martingale-difference sequence. By
orthogonality of martingale differences under the Frobenius inner
product and
$\sigma_{X,t}^2\leq\overline\sigma_X^2$,
\[
\E\|\widehat A_{12,h}\|_F^2
=
\frac1{(T-h)^2}
\sum_{t=1}^{T-h}
\E\!\left[
\sigma_{X,t}^2\|M_t\|_F^2
\right]
\leq
\frac{\overline\sigma_X^2}{T-h}
\max_t\E\|M_t\|_F^2.
\]
By Jensen's inequality and the sieve-envelope condition,
\[
\begin{aligned}
\E\|M_t\|_F^2
&\leq
\frac1N\sum_{i=1}^N
\E\left\|
\phi(Z_{i,t-1})W_{i,t-1}^\top
\right\|_F^2    =
\frac1N\sum_{i=1}^N
\E\!\left[
\|\phi(Z_{i,t-1})\|_2^2
\|W_{i,t-1}\|_2^2
\right]                                           \\
&\leq
C_\phi^2J
\frac1N\sum_{i=1}^N
\E\|W_{i,t-1}\|_2^2
\lesssim J.
\end{aligned}
\]
Here
$\frac1N\sum_{i=1}^N
\E\|W_{i,t-1}\|_2^2
=
\operatorname{tr}(A_{22,h})<\infty$
because the dimension of $W_{i,t-1}$ is fixed and the eigenvalues of
$A_{22,h}$ are uniformly bounded. Hence, by Markov's inequality,
\begin{equation} \label{eq:A12_rate}
\|\widehat A_{12,h}\|
\leq
\|\widehat A_{12,h}\|_F
=
O_p\!\left(\sqrt{\frac JT}\right).
\end{equation}

\emph{(ii) Rates for $\widehat A_{22,h}$.}
The block $\widehat A_{22,h}$ is a $Q\times Q$ sample average with
$Q$ fixed. Each entry is a stationary average with finite
$q/2>2$ moments. The functional-dependence condition in
Assumption~\ref{ass:C}(1), together with conditional
cross-sectional independence, therefore gives
\begin{equation} \label{eq:A22_rate}
\|\widehat A_{22,h}-A_{22,h}\|
=
O_p(T^{-1/2}).
\end{equation}
Because $\lambda_{\min}(A_{22,h})$ is bounded away from zero,
Weyl's inequality gives
$\lambda_{\min}(\widehat A_{22,h})
\geq \frac12\lambda_{\min}(A_{22,h})$
with probability approaching one. Therefore,
\begin{equation} \label{eq:A22_inverse}
\|\widehat A_{22,h}^{-1}\|
=
O_p(1).
\end{equation}
\smallskip
\noindent\textbf{Part \textbf{(a)}.}
By definition,
$
\widehat A_{11,h}
=
\frac{1}{T-h}
\sum_{t=1}^{T-h}
X_t^2\mathbf S_t
$
and
$
A_{11,h}
=
\E[X_t^2\mathbf S_t].
$
Because $\mathbf S_t$ is $\mathcal F_{t-1}$-measurable,
$
A_{11,h}
=
\E[\sigma_{X,t}^2\mathbf S_t].
$
It follows that
$
\widehat A_{11,h}-A_{11,h}
=
\frac{1}{T-h}
\sum_{t=1}^{T-h}
\left\{
X_t^2\mathbf S_t
-
\E[X_t^2\mathbf S_t]
\right\}.
$
Therefore, the high-level Gram-concentration condition in
Assumption~\ref{ass:C}(4) gives
$
\|
\widehat A_{11,h}-A_{11,h}
\|
=
O_p\!\left(
\sqrt{\frac{\ell_J}{T}}
+
\sqrt{\frac{J\ell_J}{NT}}
+
\frac{J\ell_J}{NT}
\right)
=
O_p(\bar r_{A,T}),
$
which proves part~\textbf{(a)}.

\smallskip
\noindent\textbf{Part \textbf{(b)}.}
By Step~0,
$A_{12,h}=0$ and $\widetilde A_{11,h}=A_{11,h}$. Hence
$\widehat{\widetilde A}_{11,h}-\widetilde A_{11,h}
=
(\widehat A_{11,h}-A_{11,h})
-
\widehat A_{12,h}
\widehat A_{22,h}^{-1}
\widehat A_{21,h}$.
By submultiplicativity of the operator norm and
\eqref{eq:A12_rate}--\eqref{eq:A22_inverse},
\begin{equation*}
\left\|
\widehat A_{12,h}
\widehat A_{22,h}^{-1}
\widehat A_{21,h}
\right\|
\leq
\|\widehat A_{12,h}\|^2
\|\widehat A_{22,h}^{-1}\|
=
O_p\!\left(\frac JT\right).
\end{equation*}
Together with part~\textbf{(a)}, this gives
$\|
\widehat{\widetilde A}_{11,h}
-
\widetilde A_{11,h}
\|
=
O_p(\bar r_{A,T}+J/T)
=
O_p(\bar a_{A,T})$,
which proves part~\textbf{(b)}.

\smallskip
\noindent\textbf{Part \textbf{(c)}.}
Because $\bar a_{A,T}\to0$, parts~\textbf{(a)} and~\textbf{(b)}
give
$\|\widehat A_{11,h}-A_{11,h}\|=o_p(1)$
and
$\|
\widehat{\widetilde A}_{11,h}
-
\widetilde A_{11,h}
\|
=o_p(1)$.
By Weyl's inequality and Step~0,
$\lambda_{\min}(\widehat A_{11,h})
\geq
c_A-o_p(1)
\geq
c_A/2$, and
$\lambda_{\min}(\widehat{\widetilde A}_{11,h})
\geq
c_A-o_p(1)
\geq
c_A/2$,
each with probability approaching one.

Together with \eqref{eq:A22_inverse}, this yields
$\|\widehat A_{11,h}^{-1}\|
+
\|\widehat{\widetilde A}_{11,h}^{-1}\|
+
\|\widehat A_{22,h}^{-1}\|
=
O_p(1)$,
which proves part~\textbf{(c)}.

\smallskip
\noindent\textbf{Part \textbf{(d)}.}
The resolvent identity gives
$\widehat A_{11,h}^{-1}-A_{11,h}^{-1}
=
-\widehat A_{11,h}^{-1}
(\widehat A_{11,h}-A_{11,h})
A_{11,h}^{-1}$.
Parts~\textbf{(a)} and~\textbf{(c)}, together with
$\|A_{11,h}^{-1}\|\leq c_A^{-1}$, imply
$\|\widehat A_{11,h}^{-1}-A_{11,h}^{-1}\|
=O_p(\bar r_{A,T})$.
Similarly,
$\widehat{\widetilde A}_{11,h}^{-1}
-\widetilde A_{11,h}^{-1}
=
-\widehat{\widetilde A}_{11,h}^{-1}
(\widehat{\widetilde A}_{11,h}-\widetilde A_{11,h})
\widetilde A_{11,h}^{-1}$.
Since
$\|\widetilde A_{11,h}^{-1}\|
=\|A_{11,h}^{-1}\|\leq c_A^{-1}$,
parts~\textbf{(b)} and~\textbf{(c)} imply
$\|
\widehat{\widetilde A}_{11,h}^{-1}
-\widetilde A_{11,h}^{-1}
\|
=O_p(\bar a_{A,T})$.
This proves part~\textbf{(d)} and completes the proof.
\end{proof}


\begin{proof}[Proof of Theorem~\ref{thm:main_final}]
Using \eqref{eq:estimation_error_for_sieve}, we obtain the
linearization
\[
\widehat{b}_h - b_h
=
\widehat{\widetilde A}_{11,h}^{-1}
\left[
\frac{1}{N(T-h)}
\sum_{t=1}^{T-h}
\left(
X_t \Phi(Z_{t-1})^\top
-
\widehat A_{12,h} \widehat A_{22,h}^{-1}W_{t-1}^\top
\right)
v_{h,t+h}
\right].
\]
Here
$v_{h,t+h}=R_h(Z_{t-1})X_t+u_{h,t+h}$ and
$R_h(Z_{t-1})\coloneqq(R_h(Z_{i,t-1}))_{i=1}^{N}$,
with $R_h(z)=g_h(z)-\phi(z)^\top b_h$ the sieve approximation
residual. By the smoothness condition in Assumption~\ref{ass:C}(3) and the
corresponding sieve approximation property,
$\sup_{z\in\mathcal Z}|R_h(z)|=O(J^{-\kappa})$.

Define the partialled-out sieve-bias term
\[
\widetilde B_{h,T}
\coloneqq
\frac1{N(T-h)}
\sum_{t=1}^{T-h}
\left\{
X_t\Phi(Z_{t-1})^\top
-
\widehat A_{12,h}
\widehat A_{22,h}^{-1}
W_{t-1}^\top
\right\}
X_tR_h(Z_{t-1}).
\]
For its first component, Assumption~\ref{ass:C}(4) and
$\sup_z|R_h(z)|=O(J^{-\kappa})$ give
\[
\begin{aligned}
\left\|
\frac1{N(T-h)}
\sum_{t=1}^{T-h}
X_t^2\Phi(Z_{t-1})^\top R_h(Z_{t-1})
\right\|_2
&\leq
\frac{1}{T-h}\sum_{t=1}^{T-h}
X_t^2
\frac1N\sum_{i=1}^N
\|\phi(Z_{i,t-1})\|_2
|R_h(Z_{i,t-1})|\\
&\lesssim
J^{1/2-\kappa}
\frac{1}{T-h}\sum_{t=1}^{T-h}X_t^2=
O_p(J^{1/2-\kappa}).
\end{aligned}
\]
For the control component, the maintained moment and law-of-large-numbers
conditions imply
\[
\begin{aligned}
\left\|
\frac{1}{N(T-h)}
\sum_{t=1}^{T-h}
W_{t-1}^\top X_tR_h(Z_{t-1})
\right\|_2
&\leq
J^{-\kappa}
\frac{1}{T-h}\sum_{t=1}^{T-h}
|X_t|
\frac1N\sum_{i=1}^N
\|W_{i,t-1}\|_2\\
&=
O_p(J^{-\kappa}).
\end{aligned}
\]
Since
$\|\widehat A_{12,h}\|
=O_p(\sqrt{J/T})$ and
$\|\widehat A_{22,h}^{-1}\|=O_p(1)$,
the partialling contribution is
$O_p(J^{1/2-\kappa}T^{-1/2})
=o_p(J^{1/2-\kappa})$.
Consequently,
$\|\widetilde B_{h,T}\|_2=O_p(J^{1/2-\kappa})$.
Because
$\|\widehat{\widetilde A}_{11,h}^{-1}\|=O_p(1)$,
it follows that
$\|
\widehat{\widetilde A}_{11,h}^{-1}
\widetilde B_{h,T}
\|_2
=O_p(J^{1/2-\kappa})$.
Therefore,
\begin{equation}\label{eq:bh_linear_main}
\widehat{b}_h - b_h
=
\widehat{\widetilde A}_{11,h}^{-1}
\left[
\frac{1}{N(T-h)}
\sum_{t=1}^{T-h}
\left(
X_t \Phi(Z_{t-1})^\top
-
\widehat A_{12,h} \widehat A_{22,h}^{-1} W_{t-1}^\top
\right)
u_{h,t+h}
\right]
+ O_p\bigl(J^{1/2-\kappa}\bigr),
\end{equation}
where the $O_p(J^{1/2-\kappa})$ is in $\ell^{2}$ norm.

Decompose the leading term into
\begin{eqnarray*}
\widehat{b}_h - b_h
&=&
\underbrace{
\widehat{\widetilde A}_{11,h}^{-1}
\left[
\frac{1}{N(T-h)}
\sum_{t=1}^{T-h}
X_t \Phi(Z_{t-1})^\top
u_{h,t+h}
\right]}_{\mathbf I_n}
\;-\;
\underbrace{
\widehat{\widetilde A}_{11,h}^{-1}
\widehat A_{12,h} \widehat A_{22,h}^{-1}
\left[
\frac{1}{N(T-h)}
\sum_{t=1}^{T-h}
W_{t-1}^\top
u_{h,t+h}
\right]}_{\mathbf S_n}
\\[6pt]
&&\;+\;O_p\bigl(J^{1/2-\kappa}\bigr).
\end{eqnarray*}

\medskip
\noindent\textbf{Step 1: Rate for $\widehat b_h-b_h$.}
We show that $\mathbf I_n=O_p\!\left(
\sqrt{\frac{J\chi_{h,J,N}}{NT}}
\right)$ and that $\mathbf S_n$ is
asymptotically dominated by $\mathbf I_n$ under the population
projection definition of $\gamma_h$.

\smallskip
\emph{Rate of $\mathbf I_n$.}
We have
\[
\frac{1}{N(T-h)}
\sum_{t=1}^{T-h}X_t\Phi(Z_{t-1})^\top u_{h,t+h}
=
\frac{1}{N(T-h)}
\sum_{i=1}^{N}\sum_{t=1}^{T-h}
X_t\phi(Z_{i,t-1})u_{i,t+h}
\in\R^{J}.
\]
Recall
$s_{h,t}
=
N^{-1/2}X_t\Phi(Z_{t-1})^\top u_{h,t+h}$.
The orthogonality condition implies $\E[s_{h,t}]=0$. 
Writing $n=T-h$, stationarity and Assumption~\ref{ass:C}(2) give
$$
\begin{aligned}
\E\left\|
\frac1{\sqrt n}\sum_{t=1}^n s_{h,t}
\right\|_2^2
&=
\sum_{|k|<n}
\left(1-\frac{|k|}{n}\right)
\operatorname{tr}\!\left(\Gamma_{h,J,N}(k)\right)
\leq
J\sum_{k\in\mathbb Z}
\left\|\Gamma_{h,J,N}(k)\right\|
\lesssim
J\chi_{h,J,N}.
\end{aligned}
$$
Therefore, by Markov's inequality,
$$
\left\|
\frac1{\sqrt{T-h}}
\sum_{t=1}^{T-h}s_{h,t}
\right\|_2
=
O_p\!\left(\sqrt{J\chi_{h,J,N}}\right).
$$

Since
$\frac{1}{N(T-h)}
\sum_{t=1}^{T-h}X_t\Phi(Z_{t-1})^\top u_{h,t+h}
=
\frac1{\sqrt{N(T-h)}}
\left(
\frac1{\sqrt{T-h}}
\sum_{t=1}^{T-h}s_{h,t}
\right)$,
it follows that
\[
\left\|
\frac{1}{N(T-h)}
\sum_{t=1}^{T-h}
X_t\Phi(Z_{t-1})^\top u_{h,t+h}
\right\|_2
=
O_p\!\left(
\sqrt{\frac{J\chi_{h,J,N}}{NT}}
\right).
\]
By Lemma~\ref{lem:gram_rates},
$\|\widehat A_{11,h}^{-1}\|=O_p(1)$.
Furthermore,
\[
\|
\widehat{\widetilde A}_{11,h}
-
\widehat A_{11,h}
\|
\leq
\|\widehat A_{12,h}\|^2
\|\widehat A_{22,h}^{-1}\|
=
O_p\!\left(\frac JT\right)
=
o_p(1).
\]
Hence, by the inverse perturbation identity,
$\|\widehat{\widetilde A}_{11,h}^{-1}\|=O_p(1)$.
Therefore,
\[
\|\mathbf I_n\|_2
=
O_p\!\left(
\sqrt{\frac{J\chi_{h,J,N}}{NT}}
\right).
\]

By Assumption~\ref{ass:C} and matrix regularity
(Lemma~\ref{lem:gram_rates} together with
$\lambda_{\min}(A_{22,h})>0$), we have
$\|\widehat A_{22,h}^{-1}\|=O_p(1)$ and
$\|\widehat{\widetilde A}_{11,h}^{-1}\|=O_p(1)$.
Under the m.d.s.\ exogeneity condition in Assumption~1,
$A_{12,h}
=
\mathbb E[X_t\phi(Z_{i,t-1})W_{i,t-1}^\top]
=0$.
Under the stated moment and time-series dependence conditions,
$\|\widehat A_{12,h}\|
=
O_p(\sqrt{J/T})$.
By Assumption~\ref{ass:C}(2),
$\mu_{Wu}\coloneqq
\E[W_{t-1}^\top u_{h,t+h}]=0$,
so
$\|
[N(T-h)]^{-1}\sum_{t=1}^{T-h}W_{t-1}^\top u_{h,t+h}
\|
=
O_p\!\left(
\sqrt{\frac{\chi_{h,J,N}}{NT}}
\right)$.
The preceding bounds give
$\|\mathbf S_n\|_2
=
T^{-1/2}
O_p\!\left(
\sqrt{\frac{J\chi_{h,J,N}}{NT}}
\right)$.
Thus, $\mathbf S_n$ is negligible relative to the stochastic rate
$\sqrt{J\chi_{h,J,N}/(NT)}$.

Adding the bias contribution $O_p(J^{1/2-\kappa})$ in $\ell^2$ from
\eqref{eq:bh_linear_main},
$\|\widehat b_h-b_h\|_2
\leq
\|\mathbf I_n\|_2+\|\mathbf S_n\|_2+O_p(J^{1/2-\kappa})
=
O_p\!\left(
\sqrt{\frac{J\chi_{h,J,N}}{NT}}
\right)+O_p(J^{1/2-\kappa})$,
where we used that $\|\mathbf S_n\|_2$ is dominated. This establishes
the stated coefficient rate, in $\ell^2$ norm.

\medskip
\noindent\textbf{Step 2: Pointwise CLT for $\widehat g_h(z)$.}
Fix $z\in\mathcal Z$. Multiply \eqref{eq:bh_linear_main} by
$\sqrt{N(T-h)}\,\phi(z)^\top$:
\begin{align*}
\sqrt{N(T-h)}\,\phi(z)^\top(\widehat b_h-b_h)
\; & =\;
\phi(z)^\top \widehat{\widetilde A}_{11,h}^{-1}
\!\left[
\frac{1}{\sqrt{N(T-h)}}\sum_{t=1}^{T-h}
\Bigl(X_t\Phi(Z_{t-1})^\top
-\widehat A_{12,h} \widehat A_{22,h}^{-1}W_{t-1}^\top\Bigr)
u_{h,t+h}\right] \\
& +\sqrt{N(T-h)}\,O_p\bigl(J^{1-\kappa}\bigr).
\end{align*}

By Step~1, $\|\widehat A_{12,h}\|=O_p(\sqrt{J/T})$ and
$\mu_{Wu}=\E[W_{t-1}^\top u_{h,t+h}]=0$ by Assumption~\ref{ass:C}(2). Hence
\[
\begin{aligned}
\Bigl\|
\widehat A_{12,h} \widehat A_{22,h}^{-1}\cdot
\frac{1}{\sqrt{N(T-h)}}\sum_t W_{t-1}^\top u_{h,t+h}
\Bigr\|
&\leq
\|\widehat A_{12,h}\|\,
\|\widehat A_{22,h}^{-1}\|\,
\Bigl\|\frac{1}{\sqrt{N(T-h)}}\sum_t W_{t-1}^\top u_{h,t+h}\Bigr\| \\
&=
O_p(\sqrt{J/T})\,O_p(1)\,O_p\!\left(\sqrt{\chi_{h,J,N}}\right).
\end{aligned}
\]

After premultiplication by
$\phi(z)^\top\widehat{\widetilde A}_{11,h}^{-1}$
and division by $\sigma_h^\star(z)$, the Frisch--Waugh correction is
$
O_p\!\left(
\frac{\|\phi(z)\|_2}{\sigma_h^\star(z)}
\sqrt{\frac{J\chi_{h,J,N}}{T}}
\right).
$
This term is controlled below by the pointwise matrix-rate condition.
Define the score process
$\widetilde S_{h,t}(z)
\coloneqq
\widetilde\phi_h(z)^\top s_{h,t}
=
\widetilde\phi_h(z)^\top
N^{-1/2}\sum_{i=1}^{N}
X_t\phi(Z_{i,t-1})u_{i,t+h}$,
where
$\widetilde\phi_h(z)\coloneqq A_{11,h}^{-1}\phi(z)$.
By Assumption~\ref{ass:C}(1)--(2),
$\{\widetilde S_{h,t}(z)\}_t$ is strictly stationary, has finite
second moments, and is mean zero because
$\E[X_t\phi(Z_{i,t-1})u_{i,t+h}]=0$.
Assumption~\ref{ass:C}(2) gives
$
\begin{aligned}
\sum_{k\in\mathbb Z}
\left|
\E\!\left[
\widetilde S_{h,t}(z)
\widetilde S_{h,t-k}(z)
\right]
\right|
&\leq
\|\widetilde\phi_h(z)\|_2^2
\sum_{k\in\mathbb Z}
\left\|\Gamma_{h,J,N}(k)\right\|
<\infty.
\end{aligned}
$
Consequently,
$
\sum_{k\in\mathbb Z}
\E\!\left[
\widetilde S_{h,t}(z)
\widetilde S_{h,t-k}(z)
\right]
=
\widetilde\phi_h(z)^\top
\Omega_{h,J,N}
\widetilde\phi_h(z).
$ 

The eigenvalue condition on $\Omega_h$ ensures that
$\sigma_h^{\star2}(z)>0$. Define
$
a_h(z)
\coloneqq
\frac{
\Omega_h^{1/2}\widetilde\phi_h(z)
}{
\sigma_h^\star(z)
}.
$
Then
$
\|a_h(z)\|_2^2
=
\frac{
\widetilde\phi_h(z)^\top
\Omega_h
\widetilde\phi_h(z)
}{
\sigma_h^{\star2}(z)
}
=1,
$
and
$
Y_{h,t,T}(z)
=
a_h(z)^\top\Omega_h^{-1/2}s_{h,t}.
$
Consequently, the standardized score-moment condition in
Assumption~\ref{ass:C}(2) gives, for some $q_0>2$,
$
\sup_{J,N}
\|Y_{h,0,T}(z)\|_{L^{q_0}}
<\infty.
$

Let $\bar\rho\coloneqq\max\{\rho,\rho_u\}<1$. A product-difference
decomposition, Hölder's inequality, the geometric functional-dependence
conditions in Assumption~\ref{ass:C}(1)--(2), and the Lipschitz
property of the spline basis give
$$
\left|
Y_{h,k,T}(z)-Y_{h,k,T,\{0\}}(z)
\right|_{L^{q_0}}
\lesssim
\min\left\{
1,L_J\bar\rho^{(k-h)_+}
\right\}.
$$
Choose
$
m_T
\coloneqq
h+
\left\lceil
\frac{(3/2+\epsilon)\log J}
{|\log\bar\rho|}
\right\rceil
$
for some $\epsilon>0$. Since $L_J\lesssim J^{3/2}$,
$
\sum_{k=m_T}^{\infty}
\left|
Y_{h,k,T}(z)-Y_{h,k,T,\{0\}}(z)
\right|_{L^{q_0}}
\lesssim
L_J\bar\rho^{m_T-h}
\lesssim
J^{-\epsilon}
\to0.
$
Moreover, $m_T=O(\log J)=O(\log T)$ because $J/T\to0$.

The long-run variance of the standardized score equals one:
$
\sum_{k\in\mathbb Z}
\E\!\left[
Y_{h,t,T}(z)Y_{h,t-k,T}(z)
\right]
=
\frac{
\widetilde\phi_h(z)^\top
\Omega_{h,J,N}
\widetilde\phi_h(z)
}{
\sigma_h^{\star2}(z)
}
=1.
$
The uniform $L^{q_0}$ bound also implies the Lindeberg condition,
since, for every $\varepsilon>0$,
$
\E\!\left[
Y_{h,0,T}(z)^2
\mathbf 1\{
|Y_{h,0,T}(z)|>\varepsilon\sqrt T
\}
\right]
\lesssim
T^{-(q_0-2)/2}
\to0.
$
The $m_T$-dependent approximation and the standard blocking argument
for triangular arrays therefore give
$
\frac{1}{\sigma_h^\star(z)}
\frac{1}{\sqrt{T-h}}
\sum_{t=1}^{T-h}
\widetilde S_{h,t}(z)
\Rightarrow
\mathcal N(0,1).
$

By Lemma~\ref{lem:gram_rates},
$\|\widehat A_{11,h}^{-1}\|=O_p(1)$ and
$\|
\widehat A_{11,h}^{-1}-A_{11,h}^{-1}
\|
=
O_p(\bar r_{A,T})$.
Moreover,
$\|
\widehat{\widetilde A}_{11,h}-\widehat A_{11,h}
\|
=
O_p(J/T)$.
Under the growth condition below, $J/T\to0$. Hence
$\widehat{\widetilde A}_{11,h}$ is nonsingular with probability
approaching one and
$\|\widehat{\widetilde A}_{11,h}^{-1}\|=O_p(1)$.
The inverse identity gives
$\widehat{\widetilde A}_{11,h}^{-1}-\widehat A_{11,h}^{-1}
=
\widehat{\widetilde A}_{11,h}^{-1}
(\widehat A_{11,h}-\widehat{\widetilde A}_{11,h})
\widehat A_{11,h}^{-1}$,
and consequently
$\|
\widehat{\widetilde A}_{11,h}^{-1}-\widehat A_{11,h}^{-1}
\|
=
O_p(J/T)$.
Combining the preceding bounds,
$\|
\widehat{\widetilde A}_{11,h}^{-1}-A_{11,h}^{-1}
\|
=
O_p(\bar a_{A,T})$.

Let
$U_{h,T}\coloneqq
(T-h)^{-1/2}\sum_{t=1}^{T-h}s_{h,t}$.
By the score bound established above,
$\|U_{h,T}\|_2
=
O_p\!\left(\sqrt{J\chi_{h,J,N}}\right)$.
The eigenvalue bounds on $A_{11,h}$ and $\Omega_h$ imply
$\|\phi(z)\|_2/\sigma_h^\star(z)=O(1)$.
It follows that
\[
\begin{aligned}
\left|
\frac{
\phi(z)^\top
\left(
\widehat{\widetilde A}_{11,h}^{-1}
-A_{11,h}^{-1}
\right)
U_{h,T}
}{
\sigma_h^\star(z)
}
\right|
&\le
\frac{\|\phi(z)\|_2}{\sigma_h^\star(z)}
\left\|
\widehat{\widetilde A}_{11,h}^{-1}
-A_{11,h}^{-1}
\right\|
\|U_{h,T}\|_2 \\
&=
O_p\!\left(
\bar a_{A,T}
\frac{
\|\phi(z)\|_2\sqrt{J\chi_{h,J,N}}
}{
\sigma_{h}^\star(z)
}
\right).
\end{aligned}
\]
Therefore, suppose that
\begin{equation}
\label{eq:pointwise_matrix_rate}
\bar a_{A,T}
\frac{
\|\phi(z)\|_2\sqrt{J\chi_{h,J,N}}
}{
\sigma_h^\star(z)
}
\to0.
\end{equation}Because
$\bar a_{A,T}\geq\sqrt{\ell_J/T}\geq T^{-1/2}$,
condition~\eqref{eq:pointwise_matrix_rate} also implies
$
\frac{\|\phi(z)\|_2}{\sigma_h^\star(z)}
\sqrt{\frac{J\chi_{h,J,N}}{T}}
\to0.
$
Hence the inverse-matrix replacement and the Frisch--Waugh 
correction are both $o_p(1)$.
 Combining these bounds with the sieve
remainder established above gives

\[
\begin{aligned}
&\frac{
\sqrt{N(T-h)}\,\phi(z)^\top(\widehat b_h-b_h)
}{
\sigma_h^\star(z)
}\\
={}&
\frac1{\sigma_h^\star(z)}
\frac1{\sqrt{T-h}}
\sum_{t=1}^{T-h}\widetilde S_{h,t}(z)+
O_p\!\left(
\bar a_{A,T}
\frac{
\|\phi(z)\|_2\sqrt{J\chi_{h,J,N}}
}{
\sigma_h^\star(z)
}
\right)+
O_p\!\left(
\frac{\|\phi(z)\|_2}{\sigma_h^\star(z)}
\sqrt{\frac{J\chi_{h,J,N}}{T}}
\right)+
O_p\!\left(
\frac{
\sqrt{NT}\,J^{1-\kappa}
}{
\sigma_h^\star(z)
}
\right).
\end{aligned}
\]
Thus, under~\eqref{eq:pointwise_matrix_rate} and the undersmoothing
condition
$
\frac{\sqrt{NT}\,J^{1-\kappa}}
{\sigma_h^\star(z)}
\to0,
$
the preceding remainder is $o_p(1)$.

Finally,
$\widehat g_h(z)-g_h(z)
=
\phi(z)^\top(\widehat b_h-b_h)-R_h(z)$, with
$\sup_{z\in\mathcal Z}|R_h(z)|=O(J^{-\kappa})$.
By pointwise nondegeneracy of $\sigma_h^\star(z)$,
$\left|
\sqrt{N(T-h)}R_h(z)/\sigma_h^\star(z)
\right|=o(1)$,
where the final conclusion follows from
$\frac{\sqrt{NT}\,J^{1-\kappa}}
{\sigma_h^\star(z)}
\to0$
. Hence
\[
\frac{\sqrt{N(T-h)}\{\widehat g_h(z)-g_h(z)\}}
{\sigma_h^\star(z)}
\Rightarrow
\mathcal N(0,1). \qedhere
\]

\end{proof}

\subsection{Proof of Theorem \ref{lem:HAC_consistency}}
\begin{lemma}[HAC bound for the infeasible effective score]
\label{lem:HAC_infeasible}
Maintain Assumption~\ref{ass:C}, and suppose that the standardized
score-moment exponent in Assumption~\ref{ass:C}(2) satisfies $q_0>4$. Recall $s_{h,t}$, $\Gamma_{h,J,N}(k)$, and
$
\Omega_{h,J,N}
=
\sum_{k\in\mathbb Z}\Gamma_{h,J,N}(k).
$
For $0\leq k\leq L<T-h$, define
$
\widehat\Gamma_h(k)
\coloneqq
\frac{1}{T-h}
\sum_{t=k+1}^{T-h}
s_{h,t}s_{h,t-k}^{\top},
$
and
$
\widehat\Omega_h^{\mathrm{inf}}
\coloneqq
\widehat\Gamma_h(0)
+
\sum_{k=1}^{L}w_k
\left\{
\widehat\Gamma_h(k)
+
\widehat\Gamma_h(k)^\top
\right\},
w_k=1-\frac{k}{L+1}.
$

For the dependence calculation, reindex and normalize the score by
outcome date:
$
\check s_{h,r}
\coloneqq
\chi_{h,J,N}^{-1/2}s_{h,r-h}.
$
Under the causal representations in Assumption~\ref{ass:C}(1)--(2),
$\{\check s_{h,r}\}_{r\in\mathbb Z}$ is a mean-zero, strictly
stationary causal process. Let
$\check s_{h,r,j}^{*(r-m)}$
denote the coupled version of its $j$th coordinate obtained by
replacing the primitive innovations at date $r-m$ with independent
copies. Define
$
d_J\coloneqq1+\log(1+L_J).
$
Then
$\sup_{J,N}
\max_{j\leq J}
\sup_{r\in\mathbb Z}
\left\|
\check s_{h,r,j}
\right\|_{L^{q_0}}
<\infty$
and
$\sup_{J,N}
\frac{1}{d_J}
\max_{j\leq J}
\sum_{m=0}^{\infty}
\sup_{r\in\mathbb Z}
\left\|
\check s_{h,r,j}
-
\check s_{h,r,j}^{*(r-m)}
\right\|_{L^{q_0}}
<\infty$.

Suppose also that
$\lim_{K\to\infty}
\sup_{J,N}
\frac{1}{\chi_{h,J,N}}
\sum_{|k|>K}
\left\|
\Gamma_{h,J,N}(k)
\right\|
=
0$.

Define
\[
\begin{aligned}
b_{\Gamma,h}(L,T)
\coloneqq{}&
\sum_{|k|>L}
\left\|
\Gamma_{h,J,N}(k)
\right\|
+
\frac{1}{L+1}
\sum_{|k|\leq L}
|k|\,
\left\|
\Gamma_{h,J,N}(k)
\right\|+
\frac{L}{T-h}
\sum_{|k|\leq L}
\left\|
\Gamma_{h,J,N}(k)
\right\|.
\end{aligned}
\]
Then
\[
\frac{1}{\chi_{h,J,N}}
\left\|
\E\!\left[
\widehat\Omega_h^{\mathrm{inf}}
\right]
-
\Omega_{h,J,N}
\right\|
\leq
\frac{b_{\Gamma,h}(L,T)}
{\chi_{h,J,N}},
\]
and
\[
\frac{1}{\chi_{h,J,N}}
\left\|
\widehat\Omega_h^{\mathrm{inf}}
-
\E\!\left[
\widehat\Omega_h^{\mathrm{inf}}
\right]
\right\|
=
O_p\!\left(
Jd_J^2\sqrt{\frac{L}{T}}
\right).
\]
Consequently,
\[
\frac{1}{\chi_{h,J,N}}
\left\|
\widehat\Omega_h^{\mathrm{inf}}
-
\Omega_{h,J,N}
\right\|
=
O_p\!\left(
Jd_J^2\sqrt{\frac{L}{T}}
\right)
+
O\!\left(
\frac{b_{\Gamma,h}(L,T)}
{\chi_{h,J,N}}
\right).
\]
In particular, if
$
L\to\infty
$
and
$
J^2d_J^4L/T\to0,
$
then
\[
\frac{
\left\|
\widehat\Omega_h^{\mathrm{inf}}
-
\Omega_{h,J,N}
\right\|
}{
\chi_{h,J,N}
}
\xrightarrow{p}0.
\]
\end{lemma}

\begin{proof}[Proof of Lemma~\ref{lem:HAC_infeasible}]
For $j\leq J$, let
$
a_{j,J,N}
\coloneqq
\chi_{h,J,N}^{-1/2}
\Omega_{h,J,N}^{1/2}e_j.
$
Because
$
\|a_{j,J,N}\|_2^2
=
e_j^\top\Omega_{h,J,N}e_j/\chi_{h,J,N}
\leq1,
$
and
$
\check s_{h,r,j}
=
a_{j,J,N}^\top
\Omega_{h,J,N}^{-1/2}s_{h,r-h},
$
the standardized score-moment condition in
Assumption~\ref{ass:C}(2) gives
$\sup_{J,N}
\max_{j\leq J}
\sup_{r\in\mathbb Z}
|\check s_{h,r,j}|_{L^{q_0}}
<\infty$.

Let $\bar\rho\coloneqq\max\{\rho,\rho_u\}<1$.
A product-difference decomposition, conditional Rosenthal's
inequality, $J/N\to0$, the sieve Lipschitz bound, and the geometric
dependence conditions in Assumption~\ref{ass:C}(1)--(2) give,
uniformly over $J,N$, $j\leq J$, and $r\in\mathbb Z$,
$\left\|
\check s_{h,r,j}
-
\check s_{h,r,j}^{*(r-m)}
\right\|_{L^{q_0}}
\lesssim
\min\left\{
1,\,
L_J\bar\rho^{(m-h)_+}
\right\}$. Since $h$ is fixed,
$\begin{aligned}
\max_{j\leq J}
\sum_{m=0}^{\infty}
\sup_{r\in\mathbb Z}
\left\|
\check s_{h,r,j}
-
\check s_{h,r,j}^{*(r-m)}
\right\|_{L^{q_0}}
&\lesssim
1+\log(1+L_J)=
d_J.
\end{aligned}$
This proves the two score-dependence bounds stated in the lemma.

For the lag-product process, define
$
\check Y_{t,k}^{j\ell}
\coloneqq
\check s_{h,t+h,j}\check s_{h,t-k+h,\ell}
-
\E\!\left[
\check s_{h,t+h,j}\check s_{h,t-k+h,\ell}
\right].
$
If $\check Y_{t,k}^{j\ell,*}$ denotes its coupled version, then
$ab-a^*b^*=(a-a^*)b+a^*(b-b^*)$, so H\"older's inequality gives
\[
\begin{aligned}
\left\|
\check Y_{t,k}^{j\ell}
-
\check Y_{t,k}^{j\ell,*}
\right\|_{L^{q_0/2}}
\lesssim{}&
\left\|
\check s_{h,t+h,j}
-
\check s_{h,t+h,j}^*
\right\|_{L^{q_0}}
\left\|
\check s_{h,t-k+h,\ell}
\right\|_{L^{q_0}}\\
&+
\left\|
\check s_{h,t+h,j}^*
\right\|_{L^{q_0}}
\left\|
\check s_{h,t-k+h,\ell}
-
\check s_{h,t-k+h,\ell}^*
\right\|_{L^{q_0}}.
\end{aligned}
\]
Thus, the lag-product process inherits summable
functional-dependence coefficients from the score process. The
relevant products of the score dependence norms are of order $d_J^2$.

Stationarity gives, for $k\geq0$,
$
\E[\widehat\Gamma_h(k)]
=
(1-k/(T-h))\Gamma_{h,J,N}(k).
$
Therefore,
$
\E[\widehat\Omega_h^{\mathrm{inf}}]
=
\sum_{|k|\leq L}
w_{|k|}(1-|k|/(T-h))\Gamma_{h,J,N}(k),
$
where $w_0=1$. Hence
\[
\begin{aligned}
\E[\widehat\Omega_h^{\mathrm{inf}}]-\Omega_{h,J,N}
={}&
-\sum_{|k|>L}\Gamma_{h,J,N}(k)
+
\sum_{|k|\leq L}
(w_{|k|}-1)\Gamma_{h,J,N}(k)\\
&-
\sum_{|k|\leq L}
w_{|k|}
\frac{|k|}{T-h}
\Gamma_{h,J,N}(k).
\end{aligned}
\]
Since
$|w_{|k|}-1|=|k|/(L+1)$ and
$0\leq w_{|k|}\leq1$,
the triangle inequality yields
\[
\left\|
\E[\widehat\Omega_h^{\mathrm{inf}}]-\Omega_{h,J,N}
\right\|
\leq
b_{\Gamma,h}(L,T).
\]
The normalized absolute-summability condition in
Assumption~\ref{ass:C}(2), the uniform normalized covariance-tail
condition, $L\to\infty$, and $L/T\to0$ imply
$\frac{b_{\Gamma,h}(L,T)}
{\chi_{h,J,N}}
\to0$.

Next, let
$
D_h
\coloneqq
\widehat\Omega_h^{\mathrm{inf}}
-
\E[\widehat\Omega_h^{\mathrm{inf}}],
$
and denote its $(j,\ell)$th entry by $D_{h,j\ell}$. Let
$
\check D_h
\coloneqq
\chi_{h,J,N}^{-1}D_h.
$
Define
$
a_{tr}^{(L)}
\coloneqq
w_{|t-r|}
\mathbf 1
\left\{
1\leq t,r\leq T-h,\ |t-r|\leq L
\right\},
$
with $w_0=1$. Then
\begin{equation}
\label{eq:HAC_quadratic_form}
\check D_{h,j\ell}
=
\frac1{T-h}
\sum_{t=1}^{T-h}
\sum_{r=1}^{T-h}
a_{tr}^{(L)}
\left\{
\check s_{h,t+h,j}\check s_{h,r+h,\ell}
-
\E\!\left[
\check s_{h,t+h,j}\check s_{h,r+h,\ell}
\right]
\right\}.
\end{equation}

Because $q_0>4$, the preceding bound and the monotonicity of
$L^p$ norms imply\\
$\max_{j\leq J}
\sum_{m=0}^{\infty}
\sup_{r\in\mathbb Z}
\left\|
\check s_{h,r,j}
-
\check s_{h,r,j}^{*(r-m)}
\right\|_{L^4}
\lesssim
d_J$.
The bilinear quadratic-form moment inequality
(see \citealt{LiuWu2010} and \citealt{WuZaffaroni2018}) for
functionally dependent processes therefore gives, uniformly over
$j,\ell\leq J$,
\begin{equation}
\label{eq:HAC_quadratic_inequality}
\left\|
\sum_{t=1}^{T-h}
\sum_{r=1}^{T-h}
a_{tr}^{(L)}
\left\{
\check s_{h,t+h,j}\check s_{h,r+h,\ell}
-
\E\!\left[
\check s_{h,t+h,j}\check s_{h,r+h,\ell}
\right]
\right\}
\right\|_{L^2}
\leq
C d_J^2
\left\{
\sum_{t=1}^{T-h}
\sum_{r=1}^{T-h}
\bigl(a_{tr}^{(L)}\bigr)^2
\right\}^{1/2}.
\end{equation}
The bilinear version follows from the corresponding quadratic-form
inequality by polarization. For the Bartlett weights,
\[
\begin{aligned}
\sum_{t=1}^{T-h}
\sum_{r=1}^{T-h}
\bigl(a_{tr}^{(L)}\bigr)^2
&=
T-h
+
2\sum_{k=1}^{L}(T-h-k)w_k^2\leq
T(1+2L)
\lesssim
TL.
\end{aligned}
\]
Combining this with
\eqref{eq:HAC_quadratic_form}--%
\eqref{eq:HAC_quadratic_inequality} gives
$\sup_{j,\ell\leq J}
|\check D_{h,j\ell}|_{L^2}
\lesssim
d_J^2\sqrt{\frac LT}$,
and hence
$\sup_{j,\ell\leq J}
\E[\check D_{h,j\ell}^2]
\lesssim
d_J^4\frac LT$.
Consequently,
\[
\E\|\check D_h\|_F^2
=
\sum_{j=1}^J
\sum_{\ell=1}^J
\E[\check D_{h,j\ell}^2]
\lesssim
J^2d_J^4\frac LT.
\]
Since $\|\check D_h\|\leq\|\check D_h\|_F$, Markov's inequality yields
$\|\check D_h\|
=
O_p\!\left(
Jd_J^2\sqrt{\frac LT}
\right)$.
\end{proof}

\begin{lemma}[Approximation-error and control-score bounds]
\label{lem:feasible_score_aux}

Maintain Assumption~\ref{ass:C}. Write
$R_{h,t}
\coloneqq
R_h(Z_{t-1})
=
\bigl(
R_h(Z_{1,t-1}),\ldots,R_h(Z_{N,t-1})
\bigr)^\top$.
Suppose that, uniformly over
$j\leq J$,
\[
\E\!\left[
\phi_j(Z_{i,t-1})^2R_h(Z_{i,t-1})^2
\right]
\lesssim J^{-2\kappa},
\qquad
\left\|
\E\!\left[
\phi_j(Z_{i,t-1})R_h(Z_{i,t-1})
\,\middle|\,
\mathcal A
\right]
\right\|_{L^2}
\lesssim J^{-\kappa-1}.
\tag{F-Loc}\label{eq:feasible_localization}
\]
Then the following conclusions hold.

\begin{enumerate}[label=\upshape(\alph*)]

\item Define
$\mathbf A_t^{\mathrm{(bias)}}
\coloneqq
X_t^2\Phi(Z_{t-1})^\top R_{h,t}$.
Then
\[
\frac1{N (T-h)}
\sum_{t=1}^{T-h}
\left\|
\mathbf A_t^{\mathrm{(bias)}}
\right\|_2^2
=
O_p\!\left(
J^{1-2\kappa}
+
NJ^{-2\kappa-1}
\right).
\]

\item Suppose the second block of the sample normal equations is
\[
\widehat A_{21,h}(\widehat b_h-b_h)
+
\widehat A_{22,h}(\widehat\gamma_h-\gamma_h)
=
\frac1{N(T-h)}
\sum_{t=1}^{T-h}
W_{t-1}^\top v_{h,t+h},
\]
and assume
$\|
[N(T-h)]^{-1}\sum_{t=1}^{T-h}
W_{t-1}^\top v_{h,t+h}
\|
=
O_p\!\left(
\sqrt{\frac{\chi_{h,J,N}}{NT}}
+
\frac{J^{-\kappa}}{\sqrt T}
\right)$,
By Lemma~\ref{lem:gram_rates} and
Theorem~\ref{thm:main_final},
$\|\widehat A_{21,h}\|
=
O_p\!\left(\sqrt{\frac JT}\right),
\|\widehat A_{22,h}^{-1}\|
=
O_p(1)$,

and
$\|\widehat b_h-b_h\|_2
=
O_p\!\left(
J^{1/2-\kappa}
+
\sqrt{\frac{J\chi_{h,J,N}}{NT}}
\right)$.
Then
\begin{equation} \label{eq:gamma_corrected_rate}
\|\widehat\gamma_h-\gamma_h\|
=
O_p(r_{\gamma,T}),
\end{equation}
where
$r_{\gamma,T}
\coloneqq
\frac{\sqrt{\chi_{h,J,N}}}{\sqrt{NT}}
+
\frac{J \sqrt{\chi_{h,J,N}}}{T\sqrt N}
+
\frac{J^{1-\kappa}}{\sqrt T}$.
The term $J^{1-\kappa}/\sqrt T$ accounts for the
approximation-bias component of $\widehat b_h-b_h$.

\item Recall that
$\widehat v_{h,t+h}
=
u_{h,t+h}
+
X_tR_{h,t}
-
X_t\Phi(Z_{t-1})(\widehat b_h-b_h)
-
W_{t-1}(\widehat\gamma_h-\gamma_h)$.
Suppose, in addition, that
$\frac1{T-h}\sum_{t=1}^{T-h}
X_t^2
\left\|
\frac1N
W_{t-1}^\top\Phi(Z_{t-1})
\right\|^2
=
O_p(1)$.
Then
\begin{equation} \label{eq:feasible_control_bound}
\frac{1}{T-h}\sum_{t=1}^{T-h}
\left\|
\frac1{\sqrt N}
W_{t-1}^{\top}\widehat v_{h,t+h}
\right\|_2^2
=
O_p\!\Bigl(
\chi_{h,J,N}
+
NJ^{-2\kappa}
+
NJ^{1-2\kappa}
+
\frac{J\chi_{h,J,N}}{T}
+
Nr_{\gamma,T}^2
\Bigr).
\end{equation}

Consequently, if
\[
NJ^{-2\kappa}
+
NJ^{1-2\kappa}
+
\frac{J\chi_{h,J,N}}{T}
+
Nr_{\gamma,T}^2
=
O(\chi_{h,J,N}),
\tag{F-Control-Growth}
\label{eq:feasible_control_growth}
\]
then
\[
\frac1{T-h}
\sum_{t=1}^{T-h}
\left\|
N^{-1/2}
W_{t-1}^\top\widehat v_{h,t+h}
\right\|_2^2
=
O_p(\chi_{h,J,N}).
\]
\end{enumerate}
\end{lemma}

\begin{proof}[Proof of Lemma~\ref{lem:feasible_score_aux}]
For part~(a), let
$A_{t,j}^{\mathrm{(bias)}}
=
X_t^2\sum_{i=1}^N
\phi_j(Z_{i,t-1})R_h(Z_{i,t-1})$.
The cross-sectional sum is $\mathcal F_{t-1}$-measurable. Therefore, by Assumption~\ref{ass:C}(1), together with $q>4$, 
\[
\begin{aligned}
\E[(A_{t,j}^{\mathrm{(bias)}})^2]
&=
\E\!\left[
\E[X_t^4\mid\mathcal F_{t-1}]
\left(
\sum_{i=1}^N
\phi_j(Z_{i,t-1})R_h(Z_{i,t-1})
\right)^2
\right]\\
&\leq
C\E\!\left[
\left(
\sum_{i=1}^N
\phi_j(Z_{i,t-1})R_h(Z_{i,t-1})
\right)^2
\right].
\end{aligned}
\]
Thus, no independence between $X_t$ and $\mathcal F_{t-1}$ is
required. Conditional cross-sectional independence and
\eqref{eq:feasible_localization} give
$\E[(A_{t,j}^{\mathrm{(bias)}})^2]
=
O(NJ^{-2\kappa}+N^2J^{-2\kappa-2})$.
Summing over $j\leq J$, averaging over $t$, and applying Markov's
inequality yields
$\frac1{N(T-h)}
\sum_{t=1}^{T-h}
\|\mathbf A_t^{\mathrm{(bias)}}\|_2^2
=
O_p(J^{1-2\kappa}+NJ^{-2\kappa-1})$.

For part~(b), the second block of the normal equations gives
\[
\widehat\gamma_h-\gamma_h
=
\widehat A_{22,h}^{-1}
\left[
\frac1{N(T-h)}
\sum_{t=1}^{T-h}
W_{t-1}^\top v_{h,t+h}
-
\widehat A_{21,h}(\widehat b_h-b_h)
\right].
\]
Hence
\[
\begin{aligned}
\|\widehat\gamma_h-\gamma_h\|
={}&&
O_p\!\left(
\sqrt{\frac{\chi_{h,J,N}}{NT}}
+
\frac{J^{-\kappa}}{\sqrt T}
\right)+
O_p\!\left(\sqrt{\frac JT}\right)
O_p\!\left(
\sqrt{\frac{J\chi_{h,J,N}}{NT}}
+
J^{1/2-\kappa}
\right)
\\&&=
O_p\!\left(
\sqrt{\frac{\chi_{h,J,N}}{NT}}
+
\frac{J\sqrt{\chi_{h,J,N}}}{T\sqrt N}
+
\frac{J^{1-\kappa}}{\sqrt T}
\right)=
O_p(r_{\gamma,T}),
\end{aligned}
\]
which proves~\eqref{eq:gamma_corrected_rate}.

For part~(c), expand
\[
\begin{aligned}
\frac1{\sqrt N}W_{t-1}^\top\widehat v_{h,t+h}
={}&
\frac1{\sqrt N}W_{t-1}^\top u_{h,t+h}
+
\frac{X_t}{\sqrt N}W_{t-1}^\top R_{h,t}\\
&-
\frac{X_t}{\sqrt N}
W_{t-1}^\top\Phi(Z_{t-1})(\widehat b_h-b_h)
-
\frac1{\sqrt N}
W_{t-1}^\top W_{t-1}
(\widehat\gamma_h-\gamma_h).
\end{aligned}
\]
Let
$
q_{h,t}^{W}
\coloneqq
N^{-1/2}W_{t-1}^\top u_{h,t+h}.
$
By stationarity and Assumption~\ref{ass:C}(2),
\[
\begin{aligned}
\E\!\left[
\frac{1}{T-h}
\sum_{t=1}^{T-h}
\|q_{h,t}^{W}\|_2^2
\right]
&=
\operatorname{tr}\!\left(
\E[q_{h,0}^{W}q_{h,0}^{W\top}]
\right) \lesssim
\sum_{k\in\mathbb Z}
\left\|
\E[q_{h,t}^{W}q_{h,t-k}^{W\top}]
\right\|
\lesssim
\chi_{h,J,N}.
\end{aligned}
\]
Therefore, Markov's inequality gives
$\frac{1}{T-h}
\sum_{t=1}^{T-h}
\left\|
\frac1{\sqrt N}
W_{t-1}^\top u_{h,t+h}
\right\|_2^2
=
O_p(\chi_{h,J,N})$.
For the second component, Cauchy--Schwarz gives
$\|
X_tN^{-1/2}W_{t-1}^\top R_{h,t}
\|_2^2
\leq
NX_t^2
(N^{-1}\|W_{t-1}\|_F^2)
(N^{-1}\|R_{h,t}\|_2^2)$,
and therefore
$\frac{1}{T-h}\sum_{t=1}^{T-h}
\|
X_tN^{-1/2}W_{t-1}^\top R_{h,t}
\|_2^2
=
O_p(NJ^{-2\kappa})$.

For the third component,
\[
\begin{aligned}
\frac{1}{T-h}\sum_{t=1}^{T-h}
\left\|
\frac{X_t}{\sqrt N}
W_{t-1}^\top\Phi(Z_{t-1})
(\widehat b_h-b_h)
\right\|_2^2
&\leq
N\|\widehat b_h-b_h\|_2^2
\frac{1}{T-h}\sum_{t=1}^{T-h}
X_t^2
\left\|
\frac1N
W_{t-1}^\top\Phi(Z_{t-1})
\right\|^2\\
&=
O_p\!\left(
NJ^{1-2\kappa}
+
\frac {J\chi_{h,J,N}}{T}
\right).
\end{aligned}
\]
Finally,
\[
\begin{aligned}
\frac{1}{T-h}\sum_{t=1}^{T-h}
\left\|
\frac1{\sqrt N}
W_{t-1}^\top W_{t-1}
(\widehat\gamma_h-\gamma_h)
\right\|_2^2
&\leq
N\|\widehat\gamma_h-\gamma_h\|_2^2
\frac{1}{T-h}\sum_{t=1}^{T-h}
\left\|
\frac1N
W_{t-1}^\top W_{t-1}
\right\|^2\\
&=
O_p(Nr_{\gamma,T}^2).
\end{aligned}
\]
Combining the four bounds proves~\eqref{eq:feasible_control_bound}.
\end{proof}

\begin{proof}[Proof of Theorem~\ref{lem:HAC_consistency}]
We split the argument into three steps.

\noindent\textbf{Step 1: Preliminary limits and matrix regularity.}

By Assumption~1, $\E[X_t\mid\mathcal F_{t-1}]=0$. Because
$\phi(Z_{i,t-1})$ and $W_{i,t-1}$ are $\mathcal F_{t-1}$-measurable,
$A_{12,h}
=
\E[X_t\phi(Z_{i,t-1})W_{i,t-1}^{\top}]
=
\E[\E(X_t\mid\mathcal F_{t-1})
\phi(Z_{i,t-1})W_{i,t-1}^{\top}]
=0$.
Consequently,
$\widetilde A_{11,h}
=
A_{11,h}
-
A_{12,h}A_{22,h}^{-1}A_{21,h}
=
A_{11,h}$.
Thus, under the maintained martingale-difference restriction, the
population covariance matrix $V_h^\star$ may equivalently be written
using either $\widetilde A_{11,h}^{-1}$ or $A_{11,h}^{-1}$.

By Assumption~\ref{ass:C}(4), there exist constants $c,C>0$ such that
$c\leq\lambda_{\min}(A_{11,h})
\leq\lambda_{\max}(A_{11,h})\leq C$.
Moreover, $A_{22,h}$ is positive definite. Lemma~\ref{lem:gram_rates}
therefore implies
$\|\widehat A_{11,h}^{-1}\|=O_p(1)$ and
$\widehat A_{11,h}^{-1}\to_p A_{11,h}^{-1}$.
Under the maintained ergodicity or weak-dependence condition for
$\{W_{i,t-1}\}$, the relevant law of large number  gives
$\widehat A_{22,h}\to_p A_{22,h}$.
It follows from continuity of matrix inversion that
$\widehat A_{22,h}^{-1}\to_p A_{22,h}^{-1}$ and
$\|\widehat A_{22,h}^{-1}\|=O_p(1)$.

To determine the rate of the sample cross-moment, write
$\widehat A_{12,h}
=
(T-h)^{-1}\sum_{t=1}^{T-h}X_tM_t$,
where
$M_t
=
N^{-1}\sum_{i=1}^N
\phi(Z_{i,t-1})W_{i,t-1}^{\top}$.
Although $A_{12,h}=0$, the aggregate shock $X_t$ is common across
units. Consequently, cross-sectional averaging does not generally
reduce the component of $M_t$ corresponding to its conditional
cross-sectional mean. Under the stated moment and time-series
dependence conditions,
$\|\widehat A_{12,h}\|
=
O_p(\sqrt{J/T})$.
The faster rate $O_p(\sqrt{J/(NT)})$ would require the additional
cross-sectional centering condition
\[
\frac1N\sum_{i=1}^N
\E\!\left[
\phi(Z_{i,t-1})W_{i,t-1}^{\top}
\,\middle|\,
\mathcal A_{t-1}
\right]
=0.
\]

Using the general rate and
$\|\widehat A_{22,h}^{-1}\|=O_p(1)$,
\[
\left\|
\widehat{\widetilde A}_{11,h}
-
\widehat A_{11,h}
\right\|
\leq
\|\widehat A_{12,h}\|^2
\|\widehat A_{22,h}^{-1}\|
=
O_p\!\left(\frac JT\right).
\]
Therefore, provided $J/T\to0$,
$\widehat{\widetilde A}_{11,h}\to_p A_{11,h}$
in operator norm. The inverse identity gives
\[
\widehat{\widetilde A}_{11,h}^{-1}
-
\widehat A_{11,h}^{-1}
=
\widehat{\widetilde A}_{11,h}^{-1}
\left(
\widehat A_{11,h}
-
\widehat{\widetilde A}_{11,h}
\right)
\widehat A_{11,h}^{-1}
=
O_p\!\left(\frac JT\right).
\]
Hence
$\widehat{\widetilde A}_{11,h}^{-1}\to_p A_{11,h}^{-1}$ and
$\|\widehat{\widetilde A}_{11,h}^{-1}\|=O_p(1)$.

\medskip
\noindent\textbf{Step 2: Consistency of the long-run covariance
estimator for the infeasible score.}

Define the infeasible effective score
\[
s_{h,t}
=
\frac1{\sqrt N}
\left\{
X_t\Phi(Z_{t-1})^\top
-
A_{12,h}A_{22,h}^{-1}W_{t-1}^\top
\right\}
u_{h,t+h}.
\]
As shown in Step~1, the martingale-difference restriction on $X_t$
and the predeterminedness of $Z_{i,t-1}$ and $W_{i,t-1}$ imply
$A_{12,h}=0$. Therefore,
$s_{h,t}
=
N^{-1/2}X_t\Phi(Z_{t-1})^\top u_{h,t+h}$.

By Lemma~\ref{lem:HAC_infeasible},
$\frac{
\left\|
\widehat\Omega_h^{\mathrm{inf}}
-
\Omega_{h,J,N}
\right\|
}{
\chi_{h,J,N}
}
=
O_p\!\left(
Jd_J^2\sqrt{\frac{L}{T}}
\right)
+
O\!\left(
\frac{
b_{\Gamma,h}(L,T)
}{
\chi_{h,J,N}
}
\right)$.

Thus, if $L\to\infty$ and
$J^2d_J^4L/T\to0$, then
$\frac{
\left\|
\widehat\Omega_h^{\mathrm{inf}}
-
\Omega_{h,J,N}
\right\|
}{
\chi_{h,J,N}
}
\xrightarrow{p}0$.

\medskip
\noindent\textbf{Step 3: Feasible scores and the sandwich mapping.}

Recall
$s_{h,t}
=
N^{-1/2}X_t\Phi(Z_{t-1})^\top u_{h,t+h}$
and
\[
\widehat s_{h,t}
=
\frac1{\sqrt N}
\left\{
X_t\Phi(Z_{t-1})^\top
-
\widehat A_{12,h}
\widehat A_{22,h}^{-1}
W_{t-1}^\top
\right\}
\widehat v_{h,t+h}.
\]
Let
$e_{h,t}\coloneqq\widehat s_{h,t}-s_{h,t}$.
Using
$\widehat v_{h,t+h}-u_{h,t+h}
=
X_tR_{h,t}
-
D_t(\widehat\theta_h-\theta_h)$,
we obtain
\[
\begin{aligned}
e_{h,t}
=
\frac1{\sqrt N}\Bigl\{
X_t^2\Phi(Z_{t-1})^\top R_{h,t}
-
X_t\Phi(Z_{t-1})^\top
D_t(\widehat\theta_h-\theta_h)
-
\widehat A_{12,h}
\widehat A_{22,h}^{-1}
W_{t-1}^\top\widehat v_{h,t+h}
\Bigr\}.
\end{aligned}
\]

By Lemma~\ref{lem:feasible_score_aux},
\begin{equation}
\label{eq:F1}
\frac{1}{
\chi_{h,J,N}N(T-h)
}
\sum_{t=1}^{T-h}
\left\|
X_t^2\Phi(Z_{t-1})^\top R_{h,t}
\right\|_2^2
=
O_p\!\left(
\frac{
J^{1-2\kappa}
+
NJ^{-2\kappa-1}
}{
\chi_{h,J,N}
}
\right).
\end{equation}
Under
\[
\sqrt{\frac{NT}{\chi_{h,J,N}}}\,J^{-\kappa}\to0
\quad\text{and}\quad
\frac JT\to0,
\]
the right-hand side is $o_p(J/T)$.

Next, set $v_h=\widehat\theta_h-\theta_h$. Under the maintained
design-moment law of large numbers,
\[
\left\|
\frac1{N(T-h)}
\sum_{t=1}^{T-h}
X_t^2
D_t^\top
\Phi(Z_{t-1})\Phi(Z_{t-1})^\top
D_t
\right\|
=
O_p(N).
\]

Consequently,
\begin{align}
&
\frac{1}{
\chi_{h,J,N}N(T-h)
}
\sum_{t=1}^{T-h}
\left\|
X_t\Phi(Z_{t-1})^\top D_t
(\widehat\theta_h-\theta_h)
\right\|_2^2
\nonumber\\
&\qquad =
O_p\!\left(
\frac{N}{\chi_{h,J,N}}
\|\widehat b_h-b_h\|_2^2
+
\frac{N}{\chi_{h,J,N}}
\|\widehat\gamma_h-\gamma_h\|_2^2
\right)
\nonumber =
O_p\!\left(
\frac{
NJ^{1-2\kappa}
}{
\chi_{h,J,N}
}
+
\frac JT
+
\frac{
Nr_{\gamma,T}^2
}{
\chi_{h,J,N}
}
\right)
\nonumber =
O_p\!\left(\frac JT\right),
\end{align}
where Theorem~\ref{thm:main_final} and
Lemma~\ref{lem:feasible_score_aux}\textup{(b)} have been used.
 Since
$\frac{NJ^{1-2\kappa}}{\chi_{h,J,N}}
=o(J/T),
\frac{Nr_{\gamma,T}^2}{\chi_{h,J,N}}
=O(J/T),$
condition~\eqref{eq:feasible_control_growth} also holds, so
Lemma~\ref{lem:feasible_score_aux}(c) applies.
For the partialling term, submultiplicativity and
Lemma~\ref{lem:feasible_score_aux}(c) give
\begin{align}
\frac{1}{\chi_{h,J,N}(T-h)}
\sum_{t=1}^{T-h}
\left\|
\frac1{\sqrt N}
\widehat A_{12,h}
\widehat A_{22,h}^{-1}
W_{t-1}^\top\widehat v_{h,t+h}
\right\|_2^2
&\leq
\|\widehat A_{12,h}\|^2
\|\widehat A_{22,h}^{-1}\|^2
\frac{1}{\chi_{h,J,N}(T-h)}
\sum_{t=1}^{T-h}
\left\|
\frac1{\sqrt N}
W_{t-1}^\top\widehat v_{h,t+h}
\right\|_2^2
\nonumber\\
&=
O_p\!\left(\frac JT\right).
\label{eq:F3}
\end{align}
because
$\|\widehat A_{12,h}\|
=
O_p(\sqrt{J/T})$ and
$\|\widehat A_{22,h}^{-1}\|=O_p(1)$.

Combining~\eqref{eq:F1}--\eqref{eq:F3} yields
$\frac{1}{\chi_{h,J,N}(T-h)}
\sum_{t=1}^{T-h}
\|e_{h,t}\|_2^2
=
O_p\!\left(
\frac{J^{1-2\kappa}+NJ^{-2\kappa-1}}{\chi_{h,J,N}}
+\frac JT
\right)$.
Under the undersmoothing condition
$\sqrt{NT/\chi_{h,J,N}}\,J^{-\kappa}\to0$ and $J/T\to0$,
$\{J^{1-2\kappa}+NJ^{-2\kappa-1}\}/\chi_{h,J,N}
=o(J/T)$.
Therefore,
\[
\left[
\frac{1}{\chi_{h,J,N}(T-h)}
\sum_{t=1}^{T-h}
\|e_{h,t}\|_2^2
\right]^{1/2}
=
O_p\!\left(
\sqrt{\frac JT}
\right).
\]
Moreover, because the dimension of $s_{h,t}$ increases with $J$, the
appropriate normalized score bound is
\[
\frac{1}{\chi_{h,J,N}(T-h)}
\sum_{t=1}^{T-h}\|s_{h,t}\|_2^2
=
O_p(J),
\]
rather than $O_p(1)$. At every lag $k\leq L$, Cauchy--Schwarz gives
\[
\begin{aligned}
\frac{1}{\chi_{h,J,N}}
\left\|
\frac{1}{T-h}\sum_t e_{h,t}s_{h,t-k}^\top
\right\|
&\leq
\left(
\frac{1}{\chi_{h,J,N}(T-h)}
\sum_t\|e_{h,t}\|_2^2
\right)^{1/2}
\left(
\frac{1}{\chi_{h,J,N}(T-h)}
\sum_t\|s_{h,t-k}\|_2^2
\right)^{1/2},\\
\frac{1}{\chi_{h,J,N}}
\left\|
\frac{1}{T-h}\sum_t e_{h,t}e_{h,t-k}^\top
\right\|
&\leq
\frac{1}{\chi_{h,J,N}(T-h)}
\sum_t\|e_{h,t}\|_2^2.
\end{aligned}
\]
Expanding
$\widehat s_{h,t}\widehat s_{h,t-k}^\top
-
s_{h,t}s_{h,t-k}^\top
=
e_{h,t}s_{h,t-k}^\top
+
s_{h,t}e_{h,t-k}^\top
+
e_{h,t}e_{h,t-k}^\top$
and summing the $L+1$ Bartlett-weighted lags therefore gives
\begin{equation} \label{eq:feasible_HAC_difference}
\frac{
\|\widehat\Omega_h-\widehat\Omega_h^{\mathrm{inf}}\|
}{
\chi_{h,J,N}
}
=
O_p\!\left[
L\sqrt J
\sqrt{\frac JT}
+
L\frac JT
\right]
=
O_p\!\left(
\frac{LJ}{\sqrt T}
\right).
\end{equation}
Consequently,
$\|\widehat\Omega_h-\widehat\Omega_h^{\mathrm{inf}}\|
/\chi_{h,J,N}
\xrightarrow{p}0$
provided
\[
\frac{LJ}{\sqrt T}\to0.
\tag{F-HAC-Growth}\label{eq:feasible_HAC_growth}
\]

Combining~\eqref{eq:feasible_HAC_difference} with Step~2 gives
$\|\widehat\Omega_h-\Omega_{h,J,N}\|
/\chi_{h,J,N}
\xrightarrow{p}0$.
Since $\Omega_h\equiv\Omega_{h,J,N}$, it follows that
$\|\widehat\Omega_h-\Omega_h\|
/\chi_{h,J,N}
\xrightarrow{p}0$.
Finally, define
$\widehat V_h^\star
=
\widehat{\widetilde A}_{11,h}^{-1}
\widehat\Omega_h
\widehat{\widetilde A}_{11,h}^{-1}$
and
$V_h^\star
=
A_{11,h}^{-1}
\Omega_h
A_{11,h}^{-1}$.
Let
$\widehat B_h
=
\widehat{\widetilde A}_{11,h}^{-1}$
and
$B_h=A_{11,h}^{-1}$.
Then
$\begin{aligned}
\widehat V_h^\star-V_h^\star
={}&
(\widehat B_h-B_h)\widehat\Omega_h\widehat B_h
+
B_h(\widehat\Omega_h-\Omega_h)\widehat B_h
+
B_h\Omega_h(\widehat B_h-B_h).
\end{aligned}$
Because
$\|\widehat B_h-B_h\|=o_p(1)$,
$\|\widehat B_h\|+\|B_h\|=O_p(1)$, and
$\|\widehat\Omega_h\|/\chi_{h,J,N}
+\|\Omega_h\|/\chi_{h,J,N}=O_p(1)$,
we conclude that
$\|\widehat V_h^\star-V_h^\star\|
/\chi_{h,J,N}
\xrightarrow{p}0$.

If, in addition,
$\phi(z)^\top V_h^\star\phi(z)
\gtrsim
\chi_{h,J,N}\|\phi(z)\|_2^2$, then
\[
\left|
\frac{
\phi(z)^\top\widehat V_h^\star\phi(z)
}{
\phi(z)^\top V_h^\star\phi(z)
}
-1
\right|
=o_p(1). \qedhere
\] 
\end{proof}

\subsection{Proof of Theorem \ref{thm:uniform_GA}}

\noindent\emph{Notation for this subsection.}
The symbols $X_t,n,p,m,M,$ and $L$ below are local to this theorem and its
proof: they denote, respectively, a generic vector process, the sample size
and dimension, the dependence-truncation and big-block lengths, and the
number of big blocks.  They are unrelated to similarly named aggregate-shock,
HAC-lag, or sieve quantities used elsewhere in the paper.

\begin{theorem}[Gaussian Approximation for High-Dimensional Stationary Time
Series]\label{strong}
Let $\{X_t\}_{t=1}^{n}$ be a $p$-dimensional mean-zero strictly stationary
Bernoulli shift, $X_t=g(\varepsilon_t,\varepsilon_{t-1},\ldots)$,
where $\{\varepsilon_t\}_{t\in\mathbb Z}$ are i.i.d. Let $m$ and $M$ be
positive integers with $m\ll M$, set
$L\coloneqq\lfloor n/(m+M)\rfloor$, and write
$\ell_p\coloneqq1\vee\log(2p)$. Suppose $p=p_n\to\infty$ and that the
following hold for some $q>4$ and $\alpha>1/2-1/q$.

\begin{enumerate}[label=\upshape(\roman*)]

\item \textbf{(Moments.)}
$\sup_n\sup_{t\in\mathbb Z}\max_{1\leq j\leq p_n}
\E|X_{t,j}|^q<\infty$.

\item \textbf{(Functional dependence.)}
Let $X_{t,\{0\}}$ denote the coupled version of $X_t$ obtained by
replacing $\varepsilon_0$ with an independent copy $\varepsilon_0^*$,
and let $X_{t,j,\{0\}}$ denote its $j$th coordinate. For $i\geq0$,
$r\geq2$, and $j\leq p$, define
$\delta_{i,r,j}\coloneqq
\|X_{i,j}-X_{i,j,\{0\}}\|_{L^r}$ and
$\Delta_{s,r,j}\coloneqq\sum_{i=s}^{\infty}\delta_{i,r,j}$.
Define
$\|X_{\cdot j}\|_{r,\alpha}
\coloneqq
\sup_{s\geq0}(s+1)^\alpha\Delta_{s,r,j}$
and
$\Psi_{r,\alpha}
\coloneqq
\max_{1\leq j\leq p}\|X_{\cdot j}\|_{r,\alpha}$,
and assume $\Psi_{r,\alpha}$ is uniformly bounded along the asymptotic
sequence for each $r\in\{2,q\}$.

For the vectorwise dependence measure, define
$\omega_{i,q}\coloneqq
\|\|X_i-X_{i,\{0\}}\|_\infty\|_{L^q}$,
$\Omega_{s,q}\coloneqq\sum_{i=s}^{\infty}\omega_{i,q}$, and
$\Pi_{q,\alpha}\coloneqq
\sup_{s\geq0}(s+1)^\alpha\Omega_{s,q}$.
Finally, define
$\Upsilon_{q,\alpha}\coloneqq
(\sum_{j=1}^{p}\|X_{\cdot j}\|_{q,\alpha}^{q})^{1/q}$
and
$\Theta_{q,\alpha}
\coloneqq
\Upsilon_{q,\alpha}\wedge(\ell_p\Pi_{q,\alpha})$.
Here $\ell_p=1\vee\log(2p)$ is used in place of the asymptotically
equivalent $\log p$ convention. By construction,
$\Theta_{q,\alpha}\leq
\Upsilon_{q,\alpha}\leq
p^{1/q}\Psi_{q,\alpha}$.

\item \textbf{(Long-run covariance.)}
The long-run covariance matrix
$\Sigma\coloneqq\sum_{l\in\mathbb Z}\E[X_0X_l^\top]$
exists, with the series converging absolutely entrywise. Moreover, for
some constants $0<c<C<\infty$,
\[
c
\leq
\min_{1\leq j\leq p}\Sigma_{jj}
\leq
\max_{1\leq j\leq p}\Sigma_{jj}
\leq
C
\]
uniformly along the asymptotic sequence.

\item \textbf{(Block-size conditions.)}
$\ell_p^{9}\ll L$,
$M\gg m\ell_p^3$, and
$\{m^{-\alpha}+v(n)\}\ell_p^2\to0$,
where $v(x)=x^{-1}$ if $\alpha>1$,
$v(x)=(\log x)/x$ if $\alpha=1$, and
$v(x)=x^{-\alpha}$ if $0<\alpha<1$.

\item \textbf{(Moment exponent condition.)}
$p^{1/q}\ell_p=o(L^{1/2-1/q})$ and
$p^{1/q}\ell_p
=o(m^{\alpha-1/2+1/q}n^{1/2-1/q})$.

\end{enumerate}

Let $Z\sim\mathcal N(0,\Sigma)$. Then, on a suitably enriched
probability space, there exists a real-valued random variable
$\widetilde Z_n$ satisfying
$\widetilde Z_n\overset{d}{=}\|Z\|_\infty$ such that
\[
\mathbb P\!\left(
\left|
\Bigl\|\frac{1}{\sqrt n}\sum_{t=1}^{n}X_t\Bigr\|_{\infty}
-\widetilde Z_n
\right|
\ge
\delta_n
\right)
\to0,
\]
where
\begin{equation}\label{eq:delta-rate}
\begin{aligned}
\delta_n\coloneqq C\Bigl[\,&
L^{-1/6}\ell_p^{7/6}
\vee p^{1/q}\ell_p L^{-1/2+1/q}
{}\vee p^{1/q}\ell_p
m^{1/2-1/q-\alpha}n^{1/q-1/2}\\
&{}\vee\sqrt{(m/M+L^{-1})\ell_p}
\vee\{m^{-\alpha}+v(n)\}^{1/3}\ell_p^{2/3}
\Bigr]
\end{aligned}
\end{equation}
for a sufficiently large constant $C$ depending only on the constants
in the assumptions, $q$, and $\alpha$.
\end{theorem}

Proof of Theorem~\ref{strong} is contained in the accompanying Technical Note, which is provided separately from the Online Appendix.

\begin{remark}[Discussion of rates]\label{rem:rates}
The first term in $\delta_n$ \eqref{eq:delta-rate} is a convenient conservative
choice for the CCK coupling of the $L$ independent big-block sums.  The
power $7/6$ makes the remainder in CCK Theorem~3.1 tend to zero; a
$(\log p)^{2/3}$ radius would require an additional unspecified slowly
diverging factor.  The second term is the finite-$q$ part of the same
coupling, and the third is the finite-$q$ cost of the $m$-dependent
approximation.  These two terms cannot in general be deleted merely
because they are $o(1)$.

The fourth term is the Gaussian cost of restoring the small and terminal
blocks.  The last term is the Gaussian covariance-comparison radius.  Its
$m^{-\alpha}$ component is essential: the Gaussian process used in the
intermediate coupling matches the covariance of the truncated process
$\{X_{t,m}\}$, not that of $\{X_t\}$.  The term $v(n)$ is the Ces\`aro
error between the full-sample covariance and the long-run covariance.
Thus $v(M)$ is not, by itself, the covariance gap in this construction.

The restrictions also have distinct roles.  The condition
$\ell_p^9\ll L$ controls the CCK remainder, and the two conditions in
(v) make the finite-$q$ truncation remainders vanish.  Zhang--Wu's exact
small-block coefficient is $Lm\Theta_{q,\alpha}^q$, so its polynomial
radius is already absorbed into the displayed CCK finite-$q$ term under
$m\ll M$; no condition $M\gg m^2$ is needed.  The stronger convenient
requirement $M\gg m\ell_p^3$ makes the Gaussian block-restoration radius
vanish (for that purpose $M\gg m\ell_p$ would suffice), while
$\{m^{-\alpha}+v(n)\}\ell_p^2\to0$ makes the covariance-comparison
radius vanish.

\smallskip
\emph{Worked example: $L=\lfloor n^{2/3}\rfloor$.}
Set $L=\lfloor n^{2/3}\rfloor$, implying $m+M\asymp n^{1/3}$ and hence
$M\asymp n^{1/3}$ when $m\ll M$.  If the two displayed finite-$q$ terms
in $\delta_n$ \eqref{eq:delta-rate} are additionally dominated, the remaining
verified rate is
\[
\delta_n\;\lesssim\;
n^{-1/9}(\log p)^{7/6}
\;\vee\;
\sqrt{m/n^{1/3}}\,(\log p)^{1/2}
\;\vee\;
\{m^{-\alpha}+v(n)\}^{1/3}(\log p)^{2/3}.
\]
In particular, making the $m^{-\alpha}$ covariance component no larger
than the CCK term requires
\[
m\gg n^{1/(3\alpha)}(\log p)^{-3/(2\alpha)}.
\]
Making the Gaussian small-block term no larger than the CCK term also
requires
\[
m\lesssim n^{1/9}(\log p)^{4/3}.
\]
These lower and upper choices are compatible at the polynomial level
only when $\alpha\ge3$ (with the boundary $\alpha=3$ permitted by the
displayed logarithmic factors), subject also to the block-size and
finite-$q$ restrictions.  For $\alpha<3$, the displayed maximum, rather
than the clean CCK term alone, is the valid conclusion.
\end{remark}

\begin{proof}[Proof of Theorem~\ref{thm:uniform_GA}]
To apply Theorem~\ref{strong}, identify its generic sample-size index
$n$ with the effective time dimension $T-h$. The rate conditions stated
in terms of $n=T-h$ are therefore asymptotically equivalent to the
corresponding conditions stated in terms of $T$.

Define
$U_{h,T}\coloneqq
(T-h)^{-1/2}\sum_{t=1}^{T-h}s_{h,t}$,
$W_{h,T}\coloneqq
[N(T-h)]^{-1/2}\sum_{t=1}^{T-h}W_{t-1}^\top u_{h,t+h}$, and
$\widetilde B_{h,T}\coloneqq
[N(T-h)]^{-1}\sum_{t=1}^{T-h}
\widetilde P_{h,t}X_tR_h(Z_{t-1})$.
We prove the two claims in four steps.

\medskip
\noindent\textbf{Step 1: Exact normal equations and the uniform rate.}

The two block normal equations for
$(\widehat b_h,\widehat\gamma_h)$, followed by exact block
elimination of $\widehat\gamma_h$, give
\begin{equation}
\label{eq:U_exact_normal_equation}
\sqrt{N(T-h)}(\widehat b_h-b_h)
=
\widehat{\widetilde A}_{11,h}^{-1}
\left[
U_{h,T}
-
\widehat A_{12,h}
\widehat A_{22,h}^{-1}W_{h,T}
+
\sqrt{N(T-h)}\,\widetilde B_{h,T}
\right].
\end{equation}
Thus, no sample control term is discarded.

Because
$\lambda_{\max}(\Omega_h)\leq\chi_{h,J,N}$ and the score
autocovariances are absolutely summable,
$
\E\|U_{h,T}\|_2^2
\lesssim
\operatorname{tr}(\Omega_h)
\lesssim
J\chi_{h,J,N}.
$
Therefore,
$
\|U_{h,T}\|_2
=
O_p\!\left(
\sqrt{J\chi_{h,J,N}}
\right).
$
Conditions and the preceding Gram-matrix
bounds give
$\|W_{h,T}\|_2=O_p(\sqrt{\chi_{h,J,N}})$,
$\|\widehat{\widetilde A}_{11,h}^{-1}\|=O_p(1)$,
$\|\widehat A_{12,h}\|=O_p(\sqrt{J/T})$, and
$\|\widehat A_{22,h}^{-1}\|=O_p(1)$.

The direct component of $\widetilde B_{h,T}$ is
$O_p(J^{1/2-\kappa})$ by the sieve-envelope condition and
Cauchy--Schwarz. Its partialling-out component is
$O_p(\sqrt{J/T}\,J^{-\kappa})$ and is therefore dominated. Hence
$\|\widetilde B_{h,T}\|_2=O_p(J^{1/2-\kappa})$.
Using
$\sup_z\|\phi(z)\|_2\lesssim\sqrt J$ in
\eqref{eq:U_exact_normal_equation}, we obtain
$$
\begin{aligned}
\sup_{z\in\mathcal Z}
\left|
\phi(z)^\top(\widehat b_h-b_h)
\right|
={}&
O_p\!\left(
J\sqrt{\frac{\chi_{h,J,N}}{NT}}
\right)
+
O_p\!\left(
\frac{J\sqrt{\chi_{h,J,N}}}{\sqrt N\,T}
\right)
+
O_p(J^{1-\kappa})\\
={}&
O_p\!\left(
J\sqrt{\frac{\chi_{h,J,N}}{NT}}
\right)
+
O_p(J^{1-\kappa}).
\end{aligned}
$$
Because
$\sup_{z\in\mathcal Z}|R_h(z)|
=
O(J^{-\kappa})
=
O(J^{1-\kappa})$,
it follows that
$$
\sup_{z\in\mathcal Z}
|\widehat g_h(z)-g_h(z)|
=
O_p\!\left(
J\sqrt{\frac{\chi_{h,J,N}}{NT}}
\right)
+
O_p(J^{1-\kappa}),
$$
which proves part~\textbf{(a)}. By stationarity and Assumption~\ref{ass:C}(2),
$$
\begin{aligned}
\E\|U_{h,T}\|_2^2
&=
\operatorname{tr}\!\left[
\sum_{|k|<T-h}
\left(1-\frac{|k|}{T-h}\right)
\Gamma_{h,J,N}(k)
\right]\leq
J\sum_{k\in\mathbb Z}
\|\Gamma_{h,J,N}(k)\|
\lesssim
J\chi_{h,J,N}.
\end{aligned}
$$

\medskip
\noindent\textbf{Step 2: Studentized linear representation.}

Because
$A_{12,h}=A_{21,h}^\top=0$,
the population Schur complement equals $A_{11,h}$. Write
$$
\sqrt{N(T-h)}(\widehat b_h-b_h)
=
A_{11,h}^{-1}U_{h,T}+r_{h,T}.
$$
Equation~\eqref{eq:U_exact_normal_equation} gives
$$
\begin{aligned}
r_{h,T}
={}&
\left(
\widehat{\widetilde A}_{11,h}^{-1}
-
A_{11,h}^{-1}
\right)U_{h,T}-
\widehat{\widetilde A}_{11,h}^{-1}
\widehat A_{12,h}
\widehat A_{22,h}^{-1}W_{h,T}
+
\sqrt{N(T-h)}\,
\widehat{\widetilde A}_{11,h}^{-1}
\widetilde B_{h,T}.
\end{aligned}
$$

Since
$$
\frac{\phi(z)}{\sigma_h^{\star}(z)}
=
A_{11,h}\psi_h(z)
$$
and
$\lambda_{\min}(\Omega_h)\geq c_\Omega>0$,
$$
\sup_{z\in\mathcal Z}
\|\psi_h(z)\|_2
\leq
c_\Omega^{-1/2}.
$$
Therefore, for the inverse-replacement term,
$$
\begin{aligned}
&\sup_{z\in\mathcal Z}
\left|
\frac{
\phi(z)^\top
\left(
\widehat{\widetilde A}_{11,h}^{-1}
-
A_{11,h}^{-1}
\right)U_{h,T}
}{
\sigma_h^{\star}(z)
}
\right|=
\sup_{z\in\mathcal Z}
\left|
\psi_h(z)^\top
A_{11,h}
\left(
\widehat{\widetilde A}_{11,h}^{-1}
-
A_{11,h}^{-1}
\right)U_{h,T}
\right|\\
&\qquad=
O_p\!\left(
\bar a_{A,T}\sqrt{J\chi_{h,J,N}}
\right),
\end{aligned}
$$
by Lemma~\ref{lem:gram_rates}. Similarly,
$$
\begin{aligned}
&\sup_{z\in\mathcal Z}
\left|
\frac{
\phi(z)^\top
\widehat{\widetilde A}_{11,h}^{-1}
\widehat A_{12,h}
\widehat A_{22,h}^{-1}W_{h,T}
}{
\sigma_h^{\star}(z)
}
\right|=
O_p\!\left(
\sqrt{\frac{J\chi_{h,J,N}}{T}}
\right)
=
O_p\!\left(
\bar a_{A,T}\sqrt{J\chi_{h,J,N}}
\right).
\end{aligned}
$$

For the approximation-bias term,
$$
\begin{aligned}
&\sqrt{N(T-h)}
\sup_{z\in\mathcal Z}
\left|
\frac{
\phi(z)^\top
\widehat{\widetilde A}_{11,h}^{-1}
\widetilde B_{h,T}
}{
\sigma_h^{\star}(z)
}
\right|=
\sqrt{N(T-h)}
\sup_{z\in\mathcal Z}
\left|
\psi_h(z)^\top
A_{11,h}
\widehat{\widetilde A}_{11,h}^{-1}
\widetilde B_{h,T}
\right|=
O_p\!\left(
\sqrt{NT}\,J^{1/2-\kappa}
\right).
\end{aligned}
$$

Moreover, Assumption~\ref{ass:C}(2) and
\eqref{eq:U_normalization} imply
$
\inf_{z\in\mathcal Z}\sigma_h^{\star}(z)
\gtrsim
\tau_J
\geq1.
$\\
Hence
$
\sqrt{N(T-h)}
\sup_{z\in\mathcal Z}
\frac{|R_h(z)|}{\sigma_h^{\star}(z)}
=
O\!\left(
\sqrt{NT}\,J^{-\kappa}
\right),
$
which is dominated by
$O(\sqrt{NT}\,J^{1/2-\kappa})$.

Because
$$
\widehat g_h(z)-g_h(z)
=
\phi(z)^\top(\widehat b_h-b_h)-R_h(z),
$$
we conclude that
\begin{equation}
\label{eq:U_studentized_remainder}
\sup_{z\in\mathcal Z}
\left|
\frac{
\sqrt{N(T-h)}
\{\widehat g_h(z)-g_h(z)\}
}{
\sigma_h^{\star}(z)
}
-
\psi_h(z)^\top U_{h,T}
\right|
=
O_p\!\left(
\bar a_{A,T}\sqrt{J\chi_{h,J,N}}
+
\sqrt{NT}\,J^{1/2-\kappa}
\right).
\end{equation}

\medskip
\noindent\textbf{Step 3: Gaussian coupling on a deterministic net.}

For $\ell=1,\ldots,K_T$, define
$
Y_{h,r,\ell}
\coloneqq
y_{h,r}(z_\ell)
=
\psi_h(z_\ell)^\top s_{h,r-h},
$
and
$
Y_{h,r}
\coloneqq
(Y_{h,r,1},\ldots,Y_{h,r,K_T})^\top.
$
Because
$
\frac{1}{\sqrt{T-h}}
\sum_{r=h+1}^{T}Y_{h,r,\ell}
=
\psi_h(z_\ell)^\top U_{h,T},
$
we may relabel this length-$(T-h)$ stationary block by
$t=1,\ldots,T-h$.
The standardized score-moment condition in
Assumption~\ref{ass:C}(2) implies the required moment condition,
while condition~\eqref{eq:U_score_dependence} implies the
functional-dependence condition of Theorem~\ref{strong} for
$\{Y_{h,t}\}$.
The long-run covariance matrix of $Y_{h,t}$ has entries
$
\Sigma_{Y,\ell k}
=
\psi_h(z_\ell)^\top
\Omega_h
\psi_h(z_k).
$
By the definition of $\psi_h(z)$,
$
\Sigma_{Y,\ell\ell}
=
\psi_h(z_\ell)^\top
\Omega_h
\psi_h(z_\ell)
=
1.
$
Thus, the coordinate variances are uniformly bounded above and away
from zero, including at rate-switching points. The covariance matrix
may nevertheless be singular when $K_T>J$, which does not affect the
coordinatewise variance condition in Theorem~\ref{strong}.

Recall that
$
G_h(z)
=
\psi_h(z)^\top
\Omega_h^{1/2}\mathcal N_J.
$
The Gaussian vector in Theorem~\ref{strong} therefore has the same
distribution as
$(G_h(z_\ell))_{\ell=1}^{K_T}$.
Theorem~\ref{strong} gives a coupling such that
\begin{equation}
\label{eq:U_studentized_net_coupling}
\left|
\max_{\ell\leq K_T}
\left|
\psi_h(z_\ell)^\top U_{h,T}
\right|
-
\max_{\ell\leq K_T}|G_h(z_\ell)|
\right|
=
O_p(\delta_T).
\end{equation}

\medskip
\noindent\textbf{Step 4: From the net to the continuum and then to
Kolmogorov distance.}

For each $z\in\mathcal Z$, choose $z_\ell\in\mathcal Z_T$ such that
$|z-z_\ell|\leq\varepsilon_T$. Since
$\|U_{h,T}\|_2
=
O_p(\sqrt{J\chi_{h,J,N}})$,
\begin{equation}
\label{eq:U_studentized_empirical_modulus}
\begin{aligned}
&\sup_{z\in\mathcal Z}
\min_{\ell\leq K_T}
\left|
\psi_h(z)^\top U_{h,T}
-
\psi_h(z_\ell)^\top U_{h,T}
\right|\leq
L_{\psi,J,N}\varepsilon_T
\|U_{h,T}\|_2=
O_p\!\left(
L_{\psi,J,N}\varepsilon_T
\sqrt{J\chi_{h,J,N}}
\right).
\end{aligned}
\end{equation}

Combining
\eqref{eq:U_studentized_remainder},
\eqref{eq:U_studentized_net_coupling}, and
\eqref{eq:U_studentized_empirical_modulus} gives
\begin{equation}
\label{eq:U_studentized_prelimit}
\begin{aligned}
&\left|
\sup_{z\in\mathcal Z}
\left|
\frac{
\sqrt{N(T-h)}
\{\widehat g_h(z)-g_h(z)\}
}{
\sigma_h^{\star}(z)
}
\right|
-
\max_{\ell\leq K_T}|G_h(z_\ell)|
\right|\\
&\quad=
O_p\!\left(
\delta_T
+
\bar a_{A,T}\sqrt{J\chi_{h,J,N}}
+
\sqrt{NT}\,J^{1/2-\kappa}
+
L_{\psi,J,N}\varepsilon_T
\sqrt{J\chi_{h,J,N}}
\right).
\end{aligned}
\end{equation}

Moreover,
$$
\begin{aligned}
\left\{
\E|G_h(z)-G_h(z')|^2
\right\}^{1/2}
&\leq
\|\Omega_h\|^{1/2}
\|\psi_h(z)-\psi_h(z')\|_2\\
&\lesssim
\sqrt{\chi_{h,J,N}}\,
L_{\psi,J,N}|z-z'|.
\end{aligned}
$$
Dudley's entropy bound and Gaussian concentration therefore give
\begin{equation}
\label{eq:U_studentized_gaussian_modulus}
\sup_{z\in\mathcal Z}
\min_{\ell\leq K_T}
|G_h(z)-G_h(z_\ell)|
=
O_p\!\left(
\sqrt{\chi_{h,J,N}}\,
L_{\psi,J,N}\varepsilon_T
\sqrt{\log(1/\varepsilon_T)}
\right).
\end{equation}

Combining
\eqref{eq:U_studentized_prelimit} and
\eqref{eq:U_studentized_gaussian_modulus}, condition
\eqref{eq:U_growth} implies
$$
\left|
\sup_{z\in\mathcal Z}
\left|
\frac{
\sqrt{N(T-h)}
\{\widehat g_h(z)-g_h(z)\}
}{
\sigma_h^{\star}(z)
}
\right|
-
\sup_{z\in\mathcal Z}|G_h(z)|
\right|
=
o_p\!\left(
\frac1{\sqrt{\ell_T}}
\right).
$$

Since
$\Var(G_h(z_\ell))=1$ for every $\ell\leq K_T$,
Gaussian anti-concentration gives, for $0<r<1$,
$$
\sup_{x\in\mathbb R}
\P\!\left(
\left|
\max_{\ell\leq K_T}|G_h(z_\ell)|-x
\right|
\leq r
\right)
\lesssim
r\sqrt{\ell_T+\log(1/r)}.
$$
The Gaussian modulus bound transfers this anti-concentration result
from the net to the continuum.
Indeed, if $E_T$ denotes the preceding coupling error, then
$\sqrt{\ell_T}E_T=o_p(1)$. A diagonal argument therefore gives a
deterministic sequence $c_T\downarrow0$ such that
$
\P\!\left(\sqrt{\ell_T}E_T>c_T\right)\to0.
$
Replacing $c_T$ by $\max\{c_T,\ell_T^{-1}\}$ if necessary preserves
this property. For $r_T=c_T/\sqrt{\ell_T}$,
$
r_T\sqrt{\ell_T+\log(1/r_T)}
\lesssim c_T\to0.
$
The coupling-to-Kolmogorov inequality therefore gives
$$
\sup_{x\in\mathbb R}
\left|
\P\!\left(
\sup_{z\in\mathcal Z}
\left|
\frac{
\sqrt{N(T-h)}
\{\widehat g_h(z)-g_h(z)\}
}{
\sigma_h^{\star}(z)
}
\right|
\leq x
\right)
-
\P\!\left(
\sup_{z\in\mathcal Z}|G_h(z)|
\leq x
\right)
\right|
\to0.
$$
This proves part~\textbf{(b)}.
\end{proof}

\newpage

\section{Additional Simulation Results} \label{sec: additional simulation resutls}

Figure \ref{fig:sieve-vs-linear under Fourier DGP} provides a robustness check
for the Monte Carlo results in Section 5.1 of the paper by repeating the comparison
under a Fourier DGP,
\[
g(z)=0.8\sin(2\pi z)+2\cos(2\pi z)-0.5\sin(4\pi z)+\cos(4\pi z),
\]
with $z$ normalized to $[0,1]$ by $(z+4.65)/9.3$. The results are similar to
those for the cubic DGP: the sieve estimator captures the nonlinear IRF, while
the linear specification misses its curvature.

Tables \ref{tab:mc_dynamic_cubic_spline_h0_h4_all_selectors} and
\ref{tab:mc_dynamic_cubic_spline_h8_h12_all_selectors} report the full results
for the uniform-band simulation in Section 5.2, including coverage and average
width across oracle, AIC, GCV, and LASSO sieve-dimension choices. Coverage
improves with sample size; AIC and GCV mildly undercover, while LASSO typically
achieves higher coverage with wider bands.

\begin{figure}[htbp]
\centering
\includegraphics[width=\textwidth]{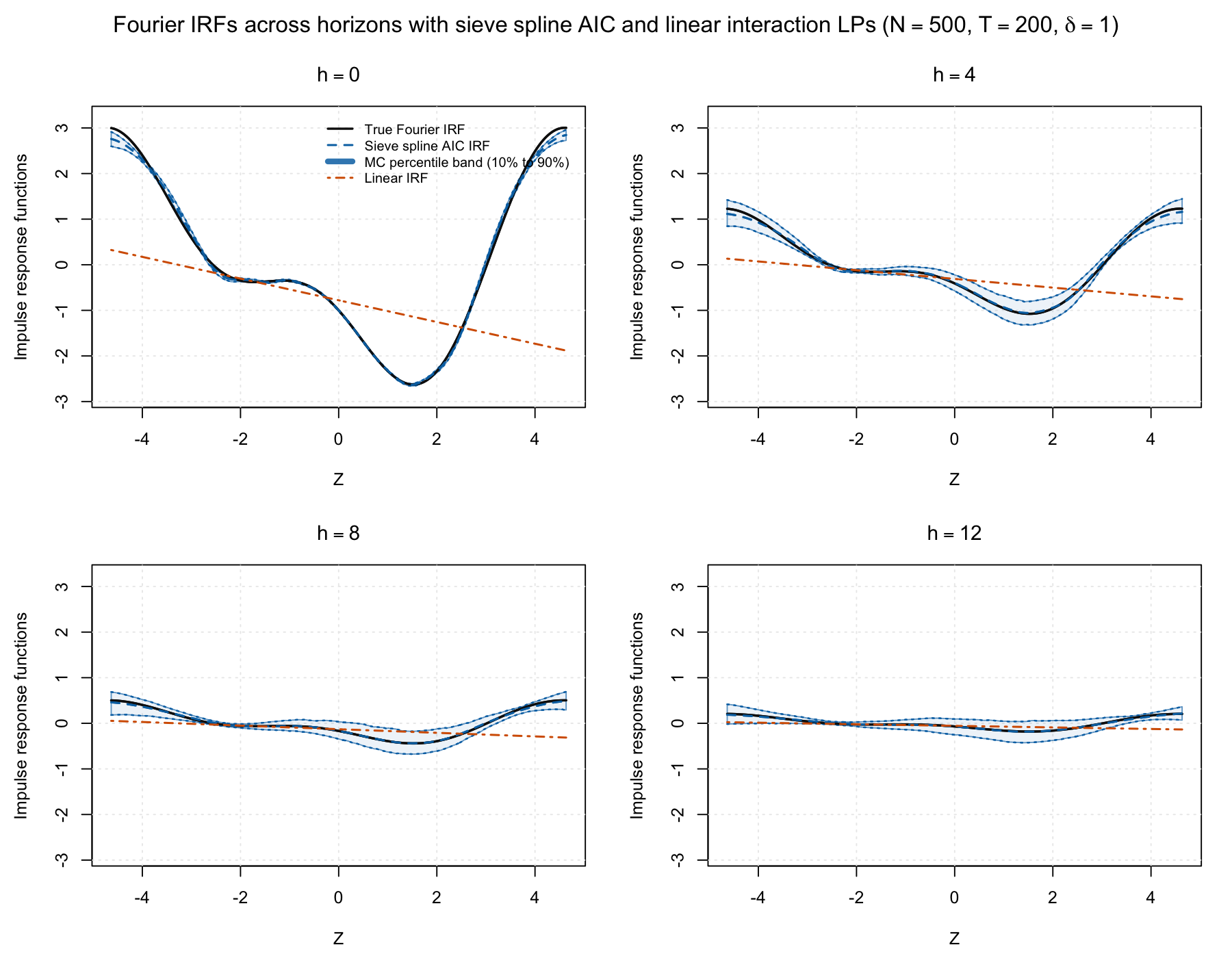}
\caption{Monte Carlo comparison of sieve and linear IRFs under the Fourier DGP}
\label{fig:sieve-vs-linear under Fourier DGP}
\end{figure}

\begin{table}[htbp]
\centering
\caption{Uniform-band performance of sieve across selector choices at horizons \(h \in \{0,4\}\) under nominal coverage \(1-\alpha = 0.95\)}
\label{tab:mc_dynamic_cubic_spline_h0_h4_all_selectors}
\small
\setlength{\tabcolsep}{3.2pt}
\renewcommand{\arraystretch}{1.06}
\begin{adjustbox}{max width=\textwidth}
\begin{tabular}{c c *{4}{S[table-format=1.2] S[table-format=2.2]}}
\toprule
& & \multicolumn{2}{c}{Oracle} & \multicolumn{2}{c}{AIC} & \multicolumn{2}{c}{GCV} & \multicolumn{2}{c}{LASSO} \\
\cmidrule(lr){3-4}\cmidrule(lr){5-6}\cmidrule(lr){7-8}\cmidrule(lr){9-10}
$N$ & $T$ & {Coverage} & {Width} & {Coverage} & {Width} & {Coverage} & {Width} & {Coverage} & {Width} \\
\midrule
\multicolumn{10}{l}{\textit{Panel A: Horizon $h=0$}} \\
\midrule
\multirow{3}{*}{1} & 40
 & 0.55 & 1.25
 & 0.44 & 1.27
 & 0.45 & 1.26
 & 0.50 & 1.33 \\
 & 120
 & 0.76 & 0.74
 & 0.63 & 0.82
 & 0.64 & 0.81
 & 0.68 & 1.07 \\
 & 200
 & 0.85 & 0.57
 & 0.73 & 0.64
 & 0.73 & 0.64
 & 0.74 & 1.06 \\
\addlinespace[0.2em]
\multirow{3}{*}{100} & 40
 & 0.84 & 0.12
 & 0.73 & 0.14
 & 0.73 & 0.14
 & 0.39 & 0.57 \\
 & 120
 & 0.91 & 0.07
 & 0.83 & 0.08
 & 0.83 & 0.08
 & 0.78 & 0.36 \\
 & 200
 & 0.93 & 0.06
 & 0.88 & 0.06
 & 0.88 & 0.06
 & 0.87 & 0.28 \\
\addlinespace[0.2em]
\multirow{3}{*}{500} & 40
 & 0.82 & 0.05
 & 0.72 & 0.06
 & 0.72 & 0.06
 & 0.41 & 0.26 \\
 & 120
 & 0.91 & 0.03
 & 0.84 & 0.04
 & 0.84 & 0.04
 & 0.76 & 0.16 \\
 & 200
 & 0.92 & 0.03
 & 0.86 & 0.03
 & 0.86 & 0.03
 & 0.86 & 0.13 \\
\midrule
\multicolumn{10}{l}{\textit{Panel B: Horizon $h=4$}} \\
\midrule
\multirow{3}{*}{1} & 40
 & 0.61 & 7.80
 & 0.52 & 8.15
 & 0.52 & 7.98
 & 0.61 & 7.84 \\
 & 120
 & 0.82 & 4.80
 & 0.72 & 5.49
 & 0.72 & 5.34
 & 0.80 & 5.10 \\
 & 200
 & 0.89 & 3.51
 & 0.78 & 4.18
 & 0.79 & 4.15
 & 0.88 & 3.92 \\
\addlinespace[0.2em]
\multirow{3}{*}{100} & 40
 & 0.82 & 2.79
 & 0.73 & 3.13
 & 0.74 & 3.13
 & 0.78 & 3.03 \\
 & 120
 & 0.90 & 1.76
 & 0.82 & 2.01
 & 0.82 & 2.01
 & 0.83 & 2.07 \\
 & 200
 & 0.90 & 1.47
 & 0.92 & 1.66
 & 0.92 & 1.66
 & 0.88 & 1.80 \\
\addlinespace[0.2em]
\multirow{3}{*}{500} & 40
 & 0.79 & 2.35
 & 0.75 & 2.56
 & 0.75 & 2.56
 & 0.73 & 2.54 \\
 & 120
 & 0.87 & 1.48
 & 0.87 & 1.64
 & 0.87 & 1.64
 & 0.84 & 1.73 \\
 & 200
 & 0.92 & 1.17
 & 0.89 & 1.30
 & 0.89 & 1.30
 & 0.90 & 1.42 \\
\bottomrule
\end{tabular}
\end{adjustbox}
\end{table}

\begin{table}[htbp]
\centering
\caption{Uniform-band performance of sieve across selector choices at horizons \(h \in \{8, 12\}\) under nominal coverage \(1-\alpha = 0.95\)}
\label{tab:mc_dynamic_cubic_spline_h8_h12_all_selectors}
\small
\setlength{\tabcolsep}{3.2pt}
\renewcommand{\arraystretch}{1.06}
\begin{adjustbox}{max width=\textwidth}
\begin{tabular}{c c *{4}{S[table-format=1.2] S[table-format=2.2]}}
\toprule
& & \multicolumn{2}{c}{Oracle} & \multicolumn{2}{c}{AIC} & \multicolumn{2}{c}{GCV} & \multicolumn{2}{c}{LASSO} \\
\cmidrule(lr){3-4}\cmidrule(lr){5-6}\cmidrule(lr){7-8}\cmidrule(lr){9-10}
$N$ & $T$ & {Coverage} & {Width} & {Coverage} & {Width} & {Coverage} & {Width} & {Coverage} & {Width} \\
\midrule
\multicolumn{10}{l}{\textit{Panel A: Horizon $h=8$}} \\
\midrule
\multirow{3}{*}{1} & 40
 & 0.61 & 9.85
 & 0.53 & 9.83
 & 0.55 & 9.83
 & 0.60 & 9.87 \\
 & 120
 & 0.83 & 5.18
 & 0.72 & 5.66
 & 0.73 & 5.57
 & 0.81 & 5.31 \\
 & 200
 & 0.91 & 3.77
 & 0.85 & 4.31
 & 0.85 & 4.21
 & 0.91 & 3.91 \\
\addlinespace[0.2em]
\multirow{3}{*}{100} & 40
 & 0.80 & 3.08
 & 0.71 & 3.46
 & 0.71 & 3.47
 & 0.76 & 3.25 \\
 & 120
 & 0.90 & 1.89
 & 0.88 & 2.17
 & 0.88 & 2.17
 & 0.89 & 2.03 \\
 & 200
 & 0.93 & 1.58
 & 0.92 & 1.82
 & 0.92 & 1.82
 & 0.89 & 1.72 \\
\addlinespace[0.2em]
\multirow{3}{*}{500} & 40
 & 0.83 & 2.67
 & 0.80 & 2.88
 & 0.80 & 2.88
 & 0.80 & 2.77 \\
 & 120
 & 0.88 & 1.62
 & 0.87 & 1.81
 & 0.87 & 1.81
 & 0.86 & 1.75 \\
 & 200
 & 0.93 & 1.26
 & 0.92 & 1.42
 & 0.92 & 1.42
 & 0.92 & 1.40 \\
\midrule
\multicolumn{10}{l}{\textit{Panel B: Horizon $h=12$}} \\
\midrule
\multirow{3}{*}{1} & 40
 & 0.54 & 12.18
 & 0.46 & 12.35
 & 0.47 & 12.11
 & 0.52 & 12.14 \\
 & 120
 & 0.84 & 5.33
 & 0.76 & 5.73
 & 0.77 & 5.68
 & 0.83 & 5.44 \\
 & 200
 & 0.91 & 3.80
 & 0.82 & 4.27
 & 0.82 & 4.25
 & 0.91 & 3.85 \\
\addlinespace[0.2em]
\multirow{3}{*}{100} & 40
 & 0.83 & 3.32
 & 0.72 & 3.75
 & 0.72 & 3.75
 & 0.81 & 3.41 \\
 & 120
 & 0.91 & 1.93
 & 0.90 & 2.24
 & 0.90 & 2.24
 & 0.89 & 2.00 \\
 & 200
 & 0.93 & 1.58
 & 0.93 & 1.82
 & 0.93 & 1.82
 & 0.93 & 1.66 \\
\addlinespace[0.2em]
\multirow{3}{*}{500} & 40
 & 0.81 & 2.85
 & 0.75 & 3.09
 & 0.75 & 3.09
 & 0.78 & 2.97 \\
 & 120
 & 0.95 & 1.64
 & 0.91 & 1.82
 & 0.91 & 1.82
 & 0.92 & 1.71 \\
 & 200
 & 0.95 & 1.30
 & 0.94 & 1.46
 & 0.94 & 1.46
 & 0.93 & 1.37 \\
\bottomrule
\end{tabular}
\end{adjustbox}
\end{table}


\newpage
\bibliographystyle{apalike}
\bibliography{nonlinearLP_bib}


\title{Technical Note to ``Causal State-Dependent Local Projections''}

\author{
Joel M. David, Raffaella Giacomini, Xiyu Jiao, and Weining Wang%
\thanks{David: Federal Reserve Bank of Chicago
(joel.david@chi.frb.org);
Giacomini: Department of Economics, University College London
(r.giacomini@ucl.ac.uk);
Jiao: Department of Economics, University of Gothenburg
(xiyu.jiao@economics.gu.se);
Wang: Department of Economics, University of Bristol
(weining.wang@bristol.ac.uk).}
}

\date{\today}
\maketitle

This Technical Note provides additional theoretical material that complements the main paper and the Online Appendix. It contains 
the proof of Theorem~A.4, which establishes the high-dimensional Gaussian approximation result used for uniform inference. 
It is presented separately from the Online Appendix for completeness and ease of reference.

\begin{proof}[Proof of Online Appendix Theorem~A.4]
The proof proceeds in three steps. Let
$\mathcal{B}_{b}\coloneqq[(b-1)(M+m)+1,\;bM+(b-1)m]$ and
$\mathcal{S}_{b}\coloneqq[bM+(b-1)m+1,\;b(M+m)]$ denote the $b$th big
block and small block respectively, for $b=1,\ldots,L$.  Set
$n_0\coloneqq L(M+m)$ and
$\mathcal R_n\coloneqq\{n_0+1,\ldots,n\}$.  Define the
$m$-dependent approximation of $X_{t}$ by
$X_{t,m}\coloneqq\mathbb{E}(X_{t}\mid\varepsilon_{t-m},\ldots,\varepsilon_{t})$,
and set
\[
Y_{b}\;\coloneqq\;\sum_{t\in\mathcal{B}_{b}}X_{t},
\qquad
\widetilde{Y}_{b}\;\coloneqq\;\sum_{t\in\mathcal{B}_{b}}X_{t,m},
\qquad b=1,\ldots,L.
\]

\noindent\textbf{Step 1: $m$-dependence approximation and small-block
removal.}
We first control the gap $T_X^\circ-T_{Y,m}$, where
\[
T_X^\circ \;\coloneqq\; \sum_{t=1}^{n_0} X_t,
\qquad
T_{Y,m} \;\coloneqq\; \sum_{b=1}^{L}\widetilde Y_b
\;=\; \sum_{b=1}^{L}\sum_{t\in\mathcal{B}_b} X_{t,m}.
\]
This gap encodes both effects at once: replacing big-block summands by
their $m$-dependent versions, and discarding the small-block
observations.  The terminal sum over $\mathcal R_n$ is treated below.

Under the standing hypotheses $q>4>2$ and $\alpha>1/2-1/q$, Lemma~7.1(i)
of \citet{ZW15gaussian}  yields, for every $y>0$,
\[
\P\!\bigl(\|T_X^\circ - T_{Y,m}\|_\infty \ge y\bigr)
\;\lesssim\;
f^{*}_{1}(y) + f^{*}_{2}(y),
\]
with
\[
\begin{aligned}
f^{*}_{1}(y)
&\;=\;
\frac{n\,m^{q/2-1-q\alpha}\,\Theta_{q,\alpha}^{q}}{y^{q}}
\;+\;
p\exp\!\Bigl(-\frac{c\,y^{2}\,m^{2\alpha}}{n\,\Psi_{2,\alpha}^{2}}\Bigr),
\\[2pt]
f^{*}_{2}(y)
&\;=\;
\frac{L\,m\,\Theta_{q,\alpha}^{q}}{y^{q}}
\;+\;
p\exp\!\Bigl(-\frac{c\,y^{2}}{L\,m\,\Psi_{2,\alpha}^{2}}\Bigr),
\end{aligned}
\]
where $\Theta_{q,\alpha}$ is defined in condition~(ii).  That condition gives
$\Theta_{q,\alpha}^{q}\lesssim p\Psi_{q,\alpha}^{q}$, so the two polynomial
coefficients are bounded respectively by
\[
p n m^{q/2-1-q\alpha}\Psi_{q,\alpha}^{q},
\qquad
p Lm\Psi_{q,\alpha}^{q}.
\]
The pair $f^{*}_{1}$ captures the
truncation gap, and the pair $f^{*}_{2}$ captures the small-block sum.

Pick
\[
\begin{aligned}
y \;=\;{} C\Psi_{2,\alpha}
\Bigl[m^{-\alpha}\sqrt{n\ell_p}\vee\sqrt{Lm\ell_p}\Bigr]
{}\vee C p^{1/q}\ell_p\Psi_{q,\alpha}
\Bigl[n^{1/q}m^{1/2-1/q-\alpha}\vee (Lm)^{1/q}\Bigr]
\end{aligned}
\]
for a sufficiently large constant $C>0$. With this choice, each of the
four terms above is $o(1)$:
\begin{itemize}
\item[i)] \emph{$f^{*}_{1}$, Gaussian:}
$y\gtrsim \Psi_{2,\alpha}\,m^{-\alpha}\sqrt{n\ell_p}$ forces
$y^{2}m^{2\alpha}/(n\Psi_{2,\alpha}^{2})\gtrsim \ell_p$, so the
exponential term tends to zero for a sufficiently large $C$.
\item[ii)] \emph{$f^{*}_{1}$, polynomial:}
$y\gtrsim p^{1/q}\ell_p n^{1/q}m^{1/2-1/q-\alpha}
\Psi_{q,\alpha}$ ensures that the polynomial ratio is
$O(\ell_p^{-q})=o(1)$.
\item[iii)] \emph{$f^{*}_{2}$, Gaussian:}
$y\gtrsim \Psi_{2,\alpha}\sqrt{L\,m\ell_p}$ forces
$y^{2}/(L\,m\,\Psi_{2,\alpha}^{2})\gtrsim \ell_p$, so the Gaussian tail
of $f^{*}_{2}$ vanishes.
\item[iv)] \emph{$f^{*}_{2}$, polynomial:}
$y\gtrsim p^{1/q}\ell_p (Lm)^{1/q}\Psi_{q,\alpha}$ again gives a ratio
$O(\ell_p^{-q})=o(1)$.
\end{itemize}
Consequently
$\P\bigl(\|T_X^\circ-T_{Y,m}\|_\infty\ge y\bigr)\to0$.

Dividing $y$ by $\sqrt n$ and
using $L\asymp n/(M+m)\asymp n/M$ (since $m\ll M$), the contribution to
$\delta_n$ 
is
\begin{align*}
\frac{y}{\sqrt n}
\;\lesssim\; &
\underbrace{m^{-\alpha}\sqrt{\ell_p}\,\Psi_{2,\alpha}}_{\text{trunc., Gaussian}}
\;\vee\;
\underbrace{\sqrt{m/M}\,\sqrt{\ell_p}\,\Psi_{2,\alpha}}_{\text{small-block, Gaussian}} \\
& \;\vee\;
\underbrace{p^{1/q}\ell_p m^{1/2-1/q-\alpha}n^{1/q-1/2}
\Psi_{q,\alpha}}_{\text{trunc., polynomial}}
\;\vee\;
\underbrace{p^{1/q}\ell_p (Lm)^{1/q} n^{-1/2}
\Psi_{q,\alpha}}_{\text{small-block, polynomial}}.
\end{align*}
Call the displayed maximum $T_1$. Its Gaussian truncation term is
dominated by the last term of 
$\delta_n$, while its Gaussian
small-block term is contained in the fourth term.  Moreover,
\[
p^{1/q}\ell_p (Lm)^{1/q}n^{-1/2}
\asymp
\{p^{1/q}\ell_p L^{-1/2+1/q}\}\,m^{1/q}/\sqrt M,
\]
so the small-block polynomial term is dominated by the second term of $\delta_n$:
indeed, $m^{2/q}\le m\ll M$ for $q>2$, and hence
$m^{1/q}/\sqrt M=o(1)$.  The truncation polynomial is
the third term of $\delta_n$. 
Hence $T_1\lesssim\delta_n$.
Finally, $|\mathcal R_n|<M+m$ and the same maximal inequality gives a
terminal-block radius
\[
T_{1,\mathrm{term}}
\lesssim
\sqrt{M\ell_p/n}
\;\vee\;
p^{1/q}\ell_p M^{1/q}n^{-1/2}.
\]
The first term is dominated by $L^{-1/6}\ell_p^{7/6}$, and the second by
$p^{1/q}\ell_p L^{-1/2+1/q}$ because $n\asymp LM$ and $q>4$.
Thus the same conclusion holds for $T_X=\sum_{t=1}^nX_t$ with $T_1$
enlarged by $T_{1,\mathrm{term}}$.

\medskip
\noindent\textbf{Step 2: Gaussian approximation of the independent block
sums.}

By construction, $\{\widetilde{Y}_{b}\}_{b=1}^{L}$ are independent: each
$\widetilde{Y}_{b}=\sum_{t\in\mathcal{B}_{b}}X_{t,m}$ depends only on
$\{\varepsilon_{s}:s\in\mathcal{B}_{b}'\}$ with
$\mathcal{B}_{b}'=[\min\mathcal{B}_{b}-m,\max\mathcal{B}_{b}]$, and the
small blocks of size $m$ between consecutive big blocks ensure
$\mathcal{B}_{b}'\cap\mathcal{B}_{b+1}'=\emptyset$. Define
$Z_{b}\sim\mathcal{N}(0,\mathbb{E}[\widetilde{Y}_{b}\widetilde{Y}_{b}^{\top}])$
independently. Define the real-valued block maxima
\[
V_L\coloneqq
\Bigl\|L^{-1/2}\textstyle\sum_{b=1}^{L}\widetilde Y_b\Bigr\|_\infty,
\qquad
G_L\coloneqq
\Bigl\|L^{-1/2}\textstyle\sum_{b=1}^{L}Z_b\Bigr\|_\infty.
\]
To handle the absolute maximum, apply Theorem~3.1 of
\cite{chernozhukov2016empirical} (hereafter CCK) to the $2p$-dimensional augmented
vectors $(\widetilde Y_b^\top,-\widetilde Y_b^\top)^\top$ and
$(Z_b^\top,-Z_b^\top)^\top$. Thus, for every Borel set
$A\subset\mathbb{R}$ and $\delta>0$,
\begin{equation}\label{eq:cck}
\mathbb{P}(V_L\in A)
\;\le\;
\mathbb{P}(G_L\in A^{C_{7}\delta})+g(\delta),
\end{equation}
where
\[
g(\delta)\;\coloneqq\;
\frac{C_{8}\ell_p^{2}}{\delta^{3}\sqrt{L}}
\bigl\{L^{*}+M_{L,1}(\delta)+M_{L,2}(\delta)\bigr\},
\]
with universal constants $C_{7},C_{8}>0$, and
\[
L^{*}\;\coloneqq\;\max_{1\le j\le p}\frac{1}{L}\sum_{b=1}^{L}\mathbb{E}|\widetilde{Y}_{b,j}|^{3},
\]
\[
M_{L,1}(\delta)\;\coloneqq\;\frac{1}{L}\sum_{b=1}^{L}\mathbb{E}\!\left[\max_{1\le j\le p}|\widetilde{Y}_{b,j}|^{3}\,
\mathbf{1}\!\left\{\max_{1\le j\le p}|\widetilde{Y}_{b,j}|\ge\frac{\delta\sqrt{L}}{\ell_p}\right\}\right],
\]
\[
M_{L,2}(\delta)\;\coloneqq\;\frac{1}{L}\sum_{b=1}^{L}\mathbb{E}\!\left[\max_{1\le j\le p}|Z_{b,j}|^{3}\,
\mathbf{1}\!\left\{\max_{1\le j\le p}|Z_{b,j}|\ge\frac{\delta\sqrt{L}}{\ell_p}\right\}\right].
\]

The cited theorem assumes finite coordinatewise third absolute moments;
it does not impose a coordinatewise variance lower bound.  Thus the
nondegeneracy condition in (iii) is not needed for this CCK step (it ensures uniformly nondegenerate Gaussian coordinate variances
for the theorem's subsequent applications).
We bound each term using the following standard inequality.
\begin{lemma}[\citealt{burkholder1988sharp,rio2009moment}]\label{lem:burkholder}
Let $q>1$, $q'=\min(q,2)$, and let $M_{n}=\sum_{t=1}^{n}\xi_{t}$ where
$\xi_{t}\in\mathcal{L}^{q}$ are martingale differences. Then
$\|M_{n}\|_{q}^{q'}\le K_{q}^{q'}\sum_{t=1}^{n}\|\xi_{t}\|_{q}^{q'}$,
where $K_{q}=\max((q-1)^{-1},\sqrt{q-1})$.
\end{lemma}

By Wu's m.d.s.\ decomposition, $X_{t,j}=\sum_{s\le t}\mathcal{P}_{s}X_{t,j}$
with $\mathcal{P}_{s}=\E[\,\cdot\mid\mathcal{F}_{s}]-\E[\,\cdot\mid\mathcal{F}_{s-1}]$,
and the projections satisfy
$\|\mathcal{P}_{s}X_{t,j}\|_{q}\le \delta_{t-s,q,j}$. Applying
Lemma~\ref{lem:burkholder} with $q'=2$ to the m.d.s.\ representation of
$\widetilde{Y}_{b,j}=\sum_{t\in\mathcal{B}_{b}}X_{t,j,m}$ and using
Minkowski's inequality across projections,
\[
\|\widetilde{Y}_{b,j}\|_{q}
\;\le\;
K_{q}\sqrt{M}\sum_{s\ge 0}\delta_{s,q,j}
\;\le\;
K_{q}\sqrt{M}\,\Psi_{q,\alpha}.
\]
Since third moments exist under $q>4$,
\[
L^{*}
\;=\;
\max_{1\le j\le p}\frac{1}{L}\sum_{b=1}^{L}\E|\widetilde{Y}_{b,j}|^{3}
\;\lesssim\;M^{3/2}\Psi_{q,\alpha}^{3}.
\]
Moreover, the union bound across coordinates gives
\[
\E[\max_j|\widetilde Y_{b,j}|^q]
\lesssim pM^{q/2}\Psi_{q,\alpha}^q.
\]
Hence a Markov-type truncation
argument following the proof of Theorem~2.1 in
\cite{chernozhukov2016empirical} gives
\[
M_{L,1}(\delta)
\;\le\;
\frac{\ell_p^{q-3}}{(\delta\sqrt{L})^{q-3}}
\E\!\left[\max_{1\le j\le p}|\widetilde{Y}_{b,j}|^{q}\right]
\;\lesssim\;
\frac{\ell_p^{q-3}}{(\delta\sqrt{L})^{q-3}}\,p\,M^{q/2}\,\Psi_{q,\alpha}^{q}.
\]
For the Gaussian term, write
\[
\sigma_{b,j}^2
\coloneqq
\E[\widetilde Y_{b,j}^2]
\lesssim
M\Psi_{2,\alpha}^2
\]
and define
\[
x_\delta
\coloneqq
\frac{\delta\sqrt L}
{\sqrt M\,\Psi_{2,\alpha}\ell_p}.
\]
For \(G\sim\mathcal N(0,1)\), integration by parts gives
\[
\E\!\left[
|G|^3\mathbf 1\{|G|\geq x\}
\right]
\lesssim
(1+x^2)e^{-x^2/2}.
\]
Moreover,
\[
\max_{j\leq p}|Z_{b,j}|^3
\mathbf 1\!\left\{
\max_{j\leq p}|Z_{b,j}|
\geq u
\right\}
\leq
\sum_{j=1}^p
|Z_{b,j}|^3
\mathbf 1\{|Z_{b,j}|\geq u\}.
\]
Therefore, taking
\(u=\delta\sqrt L/\ell_p\) and applying the preceding bound
coordinatewise yields
\[
M_{L,2}(\delta)
\lesssim
pM^{3/2}\Psi_{2,\alpha}^3
(1+x_\delta^2)
\exp(-c x_\delta^2)
\]
for some constant \(c>0\).
Setting
\[
T_{2}\;\coloneqq\;
M^{1/2}L^{-1/6}\ell_p^{7/6}\Psi_{q,\alpha}
\;\vee\;
p^{1/q}M^{1/2}\ell_p L^{-1/2+1/q}\Psi_{q,\alpha},
\]
direct substitution gives $g(T_2)=o(1)$.  Indeed, the $L^*$ contribution
is $O(\ell_p^{-3/2})$, the $M_{L,1}$ contribution is
$O(\ell_p^{-1})$, and $M_{L,2}$ is exponentially smaller under
$\ell_p^9\ll L$.  From \eqref{eq:cck},
Lemma~4.1 (Strassen's theorem) of
\citet{chernozhukov2016empirical} yields, on a suitably enriched
probability space, a real-valued random variable $G_L^c$ such that
$G_L^c\overset d=G_L$ and
\[
\P\!\left(|V_L-G_L^c|>C_7T_2\right)\to 0.
\]
Multiplying both scalar maxima by $\sqrt{L/n}$ rescales them to the
$n^{-1/2}$ block-sum scale. Since $L\asymp n/M$, Step~2 contributes
\[
T_{2}/\sqrt{M}
\;\asymp\;
L^{-1/6}\ell_p^{7/6}\Psi_{q,\alpha}
\;\vee\;
p^{1/q}\ell_p L^{-1/2+1/q}\Psi_{q,\alpha}
\]
to the rate $\delta_n$. 
\medskip\\
\noindent\textbf{Step 3: Approximation of Gaussian block sums by the
limiting Gaussian.}

Let
\[
\Gamma_m(h)\coloneqq\E[X_{0,m}X_{h,m}^{\top}],\qquad h\in\mathbb Z,
\]
and let $\{X_t^{z,m}\}_{t\in\mathbb Z}$ be a centered stationary Gaussian
process with autocovariance function $\Gamma_m$.  Such a process exists
because $\Gamma_m$ is a covariance kernel.  Moreover,
$\Gamma_m(h)=0$ for $|h|>m$, so $\{X_t^{z,m}\}$ is $m$-dependent.  This
definition is important: the Gaussian process matches the covariance of
the truncated process $\{X_{t,m}\}$, not that of $\{X_t\}$.

\emph{(a) Restoring the omitted Gaussian blocks.}
The Gaussian block vectors
$\{\sum_{t\in\mathcal B_b}X_t^{z,m}\}_{b=1}^L$ are independent and each
has covariance $\E[\widetilde Y_b\widetilde Y_b^\top]$.  Consequently,
the following identity is now exact:
\[
\sum_{b=1}^L Z_b
\overset d=
\sum_{b=1}^L\sum_{t\in\mathcal B_b}X_t^{z,m}.
\]

Recall that $\mathcal R_n=\{L(M+m)+1,\ldots,n\}$ is the terminal block.
On a space carrying $\{X_t^{z,m}\}$, define
\[
S_{n,m}^z
\coloneqq
n^{-1/2}\left(
\sum_{t=1}^nX_t^{z,m}
-\sum_{b=1}^L\sum_{t\in\mathcal B_b}X_t^{z,m}
\right).
\]
Let
$\mathcal I_n
\coloneqq
\left(\bigcup_{b=1}^L\mathcal S_b\right)\cup\mathcal R_n$.
Since
$\sup_{j\leq p}
\sum_{h\in\mathbb Z}
|\Gamma_{m,jj}(h)|
\lesssim
\Psi_{2,0}^2
\leq
\Psi_{2,\alpha}^2$,
we have
\[
\begin{aligned}
\max_{j\leq p}
\Var(S_{n,m,j}^z)
&\leq
\frac{|\mathcal I_n|}{n}
\max_{j\leq p}
\sum_{h\in\mathbb Z}
|\Gamma_{m,jj}(h)|  \\
&\lesssim
\left(
\frac{m}{M}
+
\frac1L
\right)
\Psi_{2,\alpha}^2,
\end{aligned}
\]
because
\(
|\mathcal I_n|\leq Lm+M+m
\)
and \(n\asymp LM\). Hence, for a sufficiently large constant $C$,
\[
T_{3,\mathrm{block}}
\coloneqq
C\Psi_{2,\alpha}\sqrt{(m/M+L^{-1})\ell_p}
\]
and the Gaussian maximal inequality imply
$\P(\|S_{n,m}^z\|_\infty>T_{3,\mathrm{block}})\to0$.

\smallskip
\emph{(b) Covariance matching.}
Set
\[
\Sigma_{m,n}
\coloneqq
\Var\!\left(n^{-1/2}\sum_{t=1}^nX_t^{z,m}\right),
\qquad
\Sigma_n
\coloneqq
\Var\!\left(n^{-1/2}\sum_{t=1}^nX_t\right).
\]
Writing $S_n=\sum_{t=1}^nX_t$ and
$S_{n,m}=\sum_{t=1}^nX_{t,m}$, the partial-sum moment inequalities
$\|S_n\|_2\lesssim\sqrt n\Psi_{2,0}$ and
$\|S_n-S_{n,m}\|_2\lesssim\sqrt n m^{-\alpha}\Psi_{2,\alpha}$,
together with the covariance calculation in
\citet[Lemma~7.3]{ZW15gaussian}, yield
\[
\|\Sigma_{m,n}-\Sigma_n\|_{\max}
\lesssim
\Psi_{2,\alpha}\Psi_{2,0}m^{-\alpha},
\qquad
\|\Sigma_n-\Sigma\|_{\max}
\lesssim
\Psi_{2,\alpha}\Psi_{2,0}v(n).
\]
Here $\Psi_{2,0}\le\Psi_{2,\alpha}$.
Thus, under the uniform norm bound in condition~(ii),
\[
\Delta_{m,n}
\coloneqq
\|\Sigma_{m,n}-\Sigma\|_{\max}
\lesssim m^{-\alpha}+v(n).
\]
Let
$V_{3,n}\coloneqq
\Bigl\|n^{-1/2}\textstyle\sum_{t=1}^{n}X_t^{z,m}\Bigr\|_\infty,
V_{4,n}\coloneqq\|Z\|_\infty$.
Apply Theorem~3.2 of \citet{chernozhukov2016empirical} to the
$2p$-dimensional augmented Gaussian vectors. For every Borel set
$A\subset\R$ and every $r>0$,
\[
\P(V_{3,n}\in A)
\;\le\;
\P(V_{4,n}\in A^r)
+C r^{-1}\sqrt{\Delta_{m,n}\ell_p}.
\]
Taking
\[
r_{4,n}\asymp\{m^{-\alpha}+v(n)\}^{1/3}\ell_p^{2/3}
\]
gives a Borel-neighborhood comparison with remainder
$\eta_{4,n}\lesssim\{(m^{-\alpha}+v(n))/\ell_p\}^{1/6}=o(1)$.
Condition~(iv) also makes $r_{4,n}=o(1)$.

\medskip
\noindent\textbf{Combining the three steps via scalar coupling.}
For a Borel set $A\subset\R$ and $r>0$, write
\[
A^r\coloneqq\{x\in\R:\inf_{a\in A}|x-a|\le r\}.
\]
Define the real-valued statistics
\[
\begin{aligned}
V_{0,n}&\coloneqq
\Bigl\|n^{-1/2}\textstyle\sum_{t=1}^{n}X_t\Bigr\|_\infty,\\
V_{1,n}&\coloneqq
\Bigl\|n^{-1/2}\textstyle\sum_{b=1}^{L}\widetilde Y_b\Bigr\|_\infty,\\
V_{2,n}&\coloneqq
\Bigl\|n^{-1/2}\textstyle\sum_{b=1}^{L}Z_b\Bigr\|_\infty,\\
V_{3,n}&\coloneqq
\Bigl\|n^{-1/2}\textstyle\sum_{t=1}^{n}X_t^{z,m}\Bigr\|_\infty,\\
V_{4,n}&\coloneqq\|Z\|_\infty.
\end{aligned}
\]
The tail couplings from Steps~1 and~3(a), the CCK Borel-neighborhood
comparison in Step~2, and the Gaussian Borel-neighborhood comparison
in Step~3(b) imply that, for every Borel $A\subset\R$,
\[
\P(V_{k-1,n}\in A)
\;\le\;
\P(V_{k,n}\in A^{r_{k,n}})+\eta_{k,n},
\qquad k=1,\ldots,4,
\]
where $\sum_{k=1}^{4}\eta_{k,n}=o(1)$ and
\[
\begin{aligned}
r_{1,n}&\lesssim T_1,\\
r_{2,n}&\lesssim T_2/\sqrt M,\\
r_{3,n}&\lesssim T_{3,\mathrm{block}},\\
r_{4,n}&\lesssim\{m^{-\alpha}+v(n)\}^{1/3}\ell_p^{2/3}.
\end{aligned}
\]
Because $(A^r)^s\subset A^{r+s}$, iterating the four inequalities gives
\[
\begin{aligned}
\P(V_{0,n}\in A)
&\le\P(V_{4,n}\in A^{r_n})+\eta_n,\\
r_n&\coloneqq\sum_{k=1}^{4}r_{k,n},\qquad
\eta_n\coloneqq\sum_{k=1}^{4}\eta_{k,n}=o(1).
\end{aligned}
\]
Substituting the bounds obtained in Steps~1--3 and using conditions
~(iv)--(v), we may choose
\[
r_n\lesssim\delta_n.
\]
Only this one-sided Borel-neighborhood inequality is required: it is exactly the hypothesis of Lemma~4.1 of
\citet{chernozhukov2016empirical}, so a reverse chain of
neighborhood inequalities is unnecessary.
Lemma~4.1 (Strassen's theorem) of
\citet{chernozhukov2016empirical}, applied on $\R$, therefore yields,
on a suitably enriched probability space, a real-valued random
variable $\widetilde Z_n$ such that
\[
\widetilde Z_n\overset d=V_{4,n}=\|Z\|_\infty
\]
and
\[
\P\!\left(|V_{0,n}-\widetilde Z_n|>r_n\right)
\;\le\;\eta_n\;\to\;0.
\]
Since $r_n\lesssim\delta_n$, this is the assertion of
Theorem~A.4.

\end{proof}

\bibliographystyle{apalike}
\bibliography{nonlinearLP_bib}